\documentclass{lmcs}
\pdfoutput=1
\usepackage[utf8]{inputenc}

\usepackage{lastpage}
\lmcsdoi{22}{2}{33}
\lmcsheading{}{\pageref{LastPage}}{}{}%
{May~05,~2025}{Jun.~29,~2026}{}

\usepackage{bold-extra}
\usepackage{amsmath}
\usepackage{amssymb}
\usepackage{mathtools}
\usepackage[dvipsnames]{xcolor}
\usetikzlibrary{topaths,calc}
\usetikzlibrary{patterns}
\usetikzlibrary{math}
\usepackage[labelfont=bf]{caption}
\usepackage{subcaption}
\makeatletter         
\def\figurecaption#1#2{\noindent\hangindent 40pt
                       \hbox to 36pt {\small\sl #1 \hfil}
                       \ignorespaces {\small #2}}

\long\def\@makecaption#1#2{
  \vskip 10pt 
  \settowidth{\@tempdima}{#2}
  \ifdim\@tempdima>0pt
       \setbox\@tempboxa\hbox{#1: #2}
     \else
       \setbox\@tempboxa\hbox{#1 #2}
   \fi
   \ifdim \wd\@tempboxa >\hsize               
       \begin{list}{#1:}{
       \settowidth{\labelwidth}{#1:}
       \setlength{\leftmargin}{\labelwidth}
       \addtolength{\leftmargin}{\labelsep}
        }\item #2 \end{list}\par   
     \else                                    
       \hbox to\hsize{\hfil\box\@tempboxa\hfil}  
   \fi}
\makeatother

\usepackage{stmaryrd}
\usepackage{calc}
\usepackage[section]{placeins}
\usepackage[hidelinks]{hyperref}
\usepackage{csquotes}
\usepackage{thm-restate}
\usepackage[shortcuts]{extdash}
\usepackage{dsfont}

\theoremstyle{plain}\newtheorem{claim}[thm]{Claim}   
\newenvironment{claimproof}{\begin{proof}}{\end{proof}}
\usepackage{etoolbox}

\newcommand{\NN}{\ensuremath \mathbb{N}}
\newcommand{\RR}{\ensuremath \mathbb{R}}

\newcommand{\cG}{\ensuremath{\mathcal G}}
\newcommand{\cL}{\ensuremath{\mathcal L}}
\newcommand{\cT}{\ensuremath{\mathcal T}}

\newcommand{\fI}{\ensuremath{\mathfrak I}}

\DeclareMathOperator{\dep}{dp}
\DeclareMathOperator{\wid}{wd}
\DeclareMathOperator{\tw}{tw}
\DeclareMathOperator{\td}{td}
\DeclareMathOperator{\CR}{CR}
\DeclareMathOperator{\monCR}{mon-CR}
\newcommand{\CRkq}{\CR^k_q}
\newcommand{\monCRkq}{\monCR^k_q}
\DeclareMathOperator{\eCR}{ECR}
\DeclareMathOperator{\moneCR}{mon-ECR}
\newcommand{\eCRkq}{\eCR^k_q}
\newcommand{\moneCRkq}{\moneCR^k_q}
\DeclareMathOperator{\escape}{comp}
\DeclareMathOperator{\escapeE}{comp_E}

\DeclareMathOperator{\qr}{qr}

\DeclareMathOperator{\free}{free}

\DeclareMathOperator{\dom}{dom}
\DeclareMathOperator{\img}{img}
\DeclareMathOperator{\HOM}{Hom}

\DeclareMathOperator{\partitions}{Part}
\newcommand{\partitionsEmpty}{\ensuremath{\partitions_\emptyset}}
\newcommand{\ext}{\ensuremath{\rightarrow}}
\newcommand{\symdiff}{\ensuremath{\mathbin{\triangle}}}

\newcommand{\complementOf}[1]{\ensuremath{\overline{#1}}}
\newcommand{\complementOfB}[1]{\complementOf{#1}}

\let\boundary\relax
\newcommand{\boundary}{\Delta}

\newcommand	{\homInd}			{homomorphism indistinguishability}

\newcommand	{\homIndable}		{homomorphism indistinguishable}

\newcommand	{\homDistCl}		{homomorphism distinguishing closed}

\newcommand	{\homDistClure}		{homomorphism distinguishing closure}

\newcommand{\TW}{\ensuremath{\mathcal{TW}}}
\newcommand{\TD}{\ensuremath{\mathcal{TD}}}
\newcommand{\Ekq}{\ensuremath{\cT_{q}^{k}}}
\newcommand{\EParam}[2]{\ensuremath{\cT_{#2}^{#1}}}
\newcommand{\Lkq}{\ensuremath{\cL_{q}^{k}}}
\newcommand{\LParam}[2]{\ensuremath{\cL_{#2}^{#1}}}
\newcommand{\GEkq}{\ensuremath{\cG\Lkq}}
\newcommand{\GEParam}[2]{\ensuremath{\cG\LParam{#1}{#2}}}

\newcommand{\GEkqLL}{\ensuremath{\cG\Ekq}}

\newcommand{\lFO}{\ensuremath{\mathsf{FO}}}
\newcommand{\lC} {\ensuremath{\mathsf C}}
\newcommand{\lL} {\ensuremath{\mathsf L}}
\newcommand{\lGC}{\ensuremath{\mathsf{GC}}}

\newcommand{\impl}{\ensuremath{\rightarrow}}

\newcommand{\tup}[1]{\ensuremath{\boldsymbol{#1}}}
\renewcommand{\vec}[1]{\ensuremath{\boldsymbol{#1}}}
\newcommand{\fktmid}{\ensuremath{\mid}}
\newcommand{\qg}[1]{\ensuremath{\mathfrak{#1}}}
\newcommand{\qgp}[3]{\ensuremath{\mathfrak{#1}[#3; #2]}}

\DeclareMathOperator{\cl}{cl}
\newcommand{\restrict}[1]{\vert_{#1}}
\newcommand{\parto}{\ensuremath{\rightharpoonup}}

\newcommand{\Fdec}{\ensuremath{(T,\beta)}}
\newcommand{\elimOrd}[2]{$(#1,#2)$-constructible}

\newcommand{\elimDepth}{elimination depth}
\DeclareMathOperator{\labfkt}{lab}
\newcommand{\labels}[1]{\ensuremath{LB(#1)}}
\newcommand{\grid}[2]{\ensuremath{G_{#1 \times #2}}}
\newcommand{\preTreeDec}{pre-tree-decomposition}
\newcommand{\PreTreeDec}{Pre-tree-decomposition}
\newcommand{\submod}{submodular}

\newcommand{\considered}{considered}
\newcommand{\consider}{consider}

\setlist[enumerate, 1]{font=\upshape, noitemsep, nolistsep}
\setlist[enumerate, 2]{font=\upshape, noitemsep, nolistsep}
\setlist[itemize, 1]{noitemsep, nolistsep,font=\upshape}
\setlist[itemize, 2]{noitemsep, nolistsep,font=\upshape}

\renewcommand\phi\varphi
\renewcommand\epsilon\varepsilon

\begin{document}
	
\title[Going deep and going wide]{Going deep and going wide: Counting logic and homomorphism indistinguishability over graphs of bounded treedepth and treewidth}
\titlecomment{This work extends articles \cite{fluck_et_al:LIPIcs.CSL.2024.27,adler_et_al:LIPIcs.MFCS.2024.6} which were presented at CSL and MFCS 2024. It is also part of second author's PhD thesis \cite{Fluck:993329}.}
\thanks{Tim Seppelt: European Union (CountHom, 101077083). Views and opinions expressed are however those of the author(s) only and do not necessarily reflect those of the European Union or the European Research Council Executive Agency. Neither the European Union nor the granting authority can be held responsible for them.}
\author[I.~Adler]{Isolde Adler \lmcsorcid{0000-0002-9667-9841}}[a]
\author[E.~Fluck]{Eva Fluck \lmcsorcid{0000-0002-9643-6081}}[b]
\author[T.~Seppelt]{Tim Seppelt \lmcsorcid{0000-0002-6447-0568}}[c]
\author[G.~L.~Spitzer]{Gian Luca Spitzer \lmcsorcid{0009-0008-0270-506X}}[d]

\address{University of Bamberg, Germany}
\email{isolde.adler@uni-bamberg.de}
\address{RWTH Aachen University, Germany}
\email{fluck@cs.rwth-aachen.de}
\address{IT-Universitetet i København, Denmark}
\email{tise@itu.dk}
\address{Université de Bordeaux, France}
\email{gian-luca.spitzer@u-bordeaux.fr}

\begin{abstract}
	\noindent We study the expressive power of first-order logic with counting quantifiers, especially the $k$-variable and quantifier-rank-$q$ fragment, using homomorphism indistinguishability. 
	Recently, Dawar, Jakl, and Reggio~(2021) proved that two graphs satisfy the same $k$-variable and quantifier-rank-$q$ sentences if and only if they are homomorphism indistinguishable over the class of graphs admitting a $k$-pebble forest cover of depth $q$. 
	After reproving this result using elementary means, we provide a graph-theoretic analysis of this graph class. 
	This allows us to separate it from the intersection of the class of all graphs of treewidth at most $k-1$ and the class of all graphs of treedepth at most $q$, provided that $q$ is sufficiently larger than $k$.
	
	We are able to lift this separation to a (semantic) separation of the respective homomorphism indistinguishability relations. 
	We do this by showing that the graph classes of all graphs of treedepth at most $q$ and of graphs admitting a $k$-pebble forest cover of depth $q$ are homomorphism distinguishing closed, as conjectured by Roberson~(2022).
	
	In order to prove Roberson's conjecture for the class of graphs admitting a $k$-pebble forest cover of depth $q$ we characterise the class in terms of a monotone Cops-and-Robber game.
	The crux is to prove that if Cop has a winning strategy then Cop also has a winning strategy that is monotone.
	To that end, we show how to transform Cop's winning strategy into a pre-tree-decomposition, which is inspired by decompositions of matroids, and then applying an intricate breadth-first `cleaning up' procedure along the pre-tree-decomposition (which may temporarily lose the property of representing a strategy), in order to achieve monotonicity while controlling the number of rounds simultaneously  across all branches of the decomposition via a vertex exchange argument.
\end{abstract}

\maketitle

\section{Introduction}
\label{ch:intro}
Since the 1980s, first-order logic with counting quantifiers $\mathsf{C}$ plays a decisive role in finite model theory.
In this extension of first-order logic with quantifiers $\exists^{\geq t} x$ (\enquote{there exists at least $t$ many $x$}), properties which can be expressed in first-order logic only with formulae of length depending on $t$ can be expressed succinctly.
Of particular interest are the $k$\nobreakdash-variable and quantifier-depth-$q$ fragments $\mathsf{C}^k$ and $\mathsf{C}_q$ of $\mathsf{C}$, which
enjoy rich connections to graph algorithms \cite{Dvorak_recognizing_2010}, algebraic graph theory \cite{dell_lovasz_2018,Grohe_homomorphism-tensors_2022}, optimisation \cite{Grohe_homomorphism-tensors_2022,roberson_lasserre_2023}, graph neural networks~\cite{morris_weisfeiler_2019,Xu_gnn_2019,grohe_logic_2021}, and category theory \cite{dawar_lovasz-type_2021,Abramsky_pebbling_2017}.

The intersection of these fragments, the fragment $\mathsf{C}^k_q \coloneqq \mathsf{C}^k \cap \mathsf{C}_q$ of all $\mathsf{C}$-formulae with $k$-variables and quantifier-depth $q$, has received much less attention \cite{rattan_weisfeiler-leman_2023}. In this work, we study the expressivity of $\mathsf{C}^k_q$ using homomorphism indistinguishability.

Homomorphism indistinguishability is an emerging framework for measuring the expressivity of equivalence relations comparing graphs, cf.\ the monograph \cite{seppelt_homomorphism_2024}. Two graphs $G$ and $H$ are \emph{homomorphism indistinguishable} over a graph class $\mathcal{F}$, in symbols $G \equiv_{\mathcal{F}} H$, if for all $F \in \mathcal{F}$ the number of homomorphisms from $F$ to $G$ is equal to the number of homomorphisms from $F$ to $H$. Many natural equivalence relations between graphs including isomorphism \cite{Lovasz_isomorphism_1967}, quantum isomorphism \cite{mancinska_quantum_2020}, cospectrality \cite{dell_lovasz_2018}, and feasibility of integer programming relaxations for graph isomorphism \cite{Grohe_homomorphism-tensors_2022,roberson_lasserre_2023} can be characterised as homomorphism indistinguishability relations over certain graph classes. Establishing such characterisations is intriguing since it allows to use tools from structural graph theory to study equivalence relations between graphs \cite{roberson_oddomorphisms_2022,seppelt_logical_2023}.
Furthermore, the expressivity of homomorphism counts themselves is of practical interest \cite{nguyen_graph_2020,grohe_word2vec_2020}.

\paragraph{Fragments of counting logic and graph decompositions.}
Equivalence with respect to $\mathsf{C}^{k}$ and $\mathsf{C}_q$ has been characterised by Dvo\v{r}ák~\cite{Dvorak_recognizing_2010} and Grohe~\cite{Grohe_counting_2020} as homomorphism indistinguishability over the classes $\mathcal{TW}_{k-1}$ of graphs of treewidth $\leq k-1$ and $\mathcal{TD}_q$ of graphs of treedepth $\leq q$, respectively.
Recently, Dawar, Jakl, and Reggio \cite{dawar_lovasz-type_2021} proved that two graphs satisfy the same $\mathsf{C}^k_q$-sentences if and only if they are homomorphism indistinguishable over the class $\Ekq \subseteq \mathcal{TW}_{k-1} \cap \mathcal{TD}_q$ of graphs admitting a $k$-pebble forest cover of depth $q$. Their proof builds on the categorical framework of game comonads developed in~\cite{Abramsky_pebbling_2017}.

As a first step, we reprove their result using elementary techniques inspired by Dvo\v{r}\'ak \cite{Dvorak_recognizing_2010}. The general idea is to translate between sentences in $\mathsf{C}^k_q$ and graphs from which homomorphisms are counted in an inductive fashion. By carefully imposing structural constraints, we are able to extend the original correspondence from \cite{Dvorak_recognizing_2010} between $\mathsf{C}^{k}$ and graphs of treewidth at most $k-1$ to $\mathsf{C}_q$ and graphs of treedepth at most $q$, reproducing a result of \cite{Grohe_counting_2020}, and finally to $\mathsf{C}^k_q$ and $\Ekq$. This simple and uniform proof strategy also yields the following result on guarded counting logic $\mathsf{GC}^k_q$.
Guarded counting logic plays a crucial role in the theory of properties of higher arity expressible by graph neural networks \cite{grohe_logic_2021}.
Towards this goal we introduce a new graph class called \GEkqLL, which is closely related to \Ekq.

\begin{restatable}{thm}{thmGuardedEkqGuardedLogicInformal}
	\label{thm:guardedEkq_vs_guarded-logic-informal}
	Two graphs are equivalent over the $k$-variable and quantifier-depth-$q$ fragment $\mathsf{GC}^k_q$ of guarded counting logic if and only if they are \homIndable\ over \GEkqLL.
\end{restatable}

\paragraph{Separating $\Ekq$ and $\mathcal{TW}_{k-1} \cap \mathcal{TD}_q$.}
The main contribution of this work, however, concerns the relationship between the graph class $\Ekq$ and the class $\mathcal{TW}_{k-1} \cap \mathcal{TD}_q$ of graphs which have treewidth at most $k-1$ \emph{and} treedepth at most $q$.
Given the results of \cite{Dvorak_recognizing_2010,Grohe_counting_2020}, one might think that elementary equivalence with respect to sentences in $\mathsf{C}^{k}_q = \mathsf{C}^{k} \cap \mathsf{C}_q$ is characterised by homomorphism indistinguishability with respect to $\mathcal{TW}_{k-1} \cap \mathcal{TD}_q$. The central result of this paper asserts that this intuition is wrong.

As a first step towards this, we prove, building on \cite{Furer_rounds_2001}, that the graph class $\Ekq$ and $ \TW_{k-1}\cap\TD_q$ are distinct if $q$ is sufficiently larger than $k$.
In other words, the existence of both a tree-decomposition of small width and the existence of a forest cover of small depth does not guarantee the existence of a decomposition which has both small width and depth simultaneously.

\begin{restatable}{thm}{thmEkqTwTdInformal}
	\label{thm:Ekq_tw-td-informal}
	For $q$ sufficiently larger than $k$, it holds that $\Ekq\subsetneq \TW_{k-1}\cap\TD_q$.
\end{restatable}

However, this, let us say syntactical, separation of the graph classes $\Ekq$ and $\mathcal{TW}_{k-1} \cap \mathcal{TD}_q$ does not suffice to separate their homomorphism indistinguishability relations semantically.\footnote{Homorphisms counts $\hom(F, -)$ from a graph $F$ can be regarded as logical queries to a graph $G$, which are evaluated as the number of homomorphisms $F \to G$. Analogously, a counting logic sentence $\phi \in \mathsf{C}$ allows to query a graph $G$ by evaluating whether $G$ satisfies $\phi$. In fact, we will show in \autoref{lem:homcounts_in_ckq} and \autoref{cor:qg_from_ckq} that these perspectives are two sides of the same coin.
	The pattern $F$ is the syntax of the query while the function $\hom(F, -)$ is its semantics. For example, the patterns $K_1$ and $K_1 + K_1$ are syntactically different while, semantically, $\hom(K_1, G) = |V(G)|$ and $\hom(K_1 + K_1, G) = |V(G)|^2$ essentially encode the same information about $G$.}
In fact, it could well be that all graphs which are homomorphism indistinguishable over $\Ekq$ are also homomorphism indistinguishable over $\mathcal{TW}_{k-1} \cap \mathcal{TD}_q$.

That distinct graph classes induce---under certain mild assumptions---distinct homomorphism indistinguishability relations was recently conjectured by Roberson~\cite{roberson_oddomorphisms_2022}. 
His conjecture asserts that every graph class which is closed under taking minors and disjoint unions is homomorphism distinguishing closed. Here, a graph class $\mathcal{F}$ is \emph{homomorphism distinguishing closed} if it satisfies the following maximality condition: For every graph $F \not\in \mathcal{F}$, there exist two graphs $G$ and $H$ which are homomorphism indistinguishable over $\mathcal{F}$ but have different numbers of homomorphism from $F$.

Since $\Ekq$, $\mathcal{TW}_{k-1}$, and $\mathcal{TD}_q$ are closed under disjoint unions and minors, the confirmation of Roberson's conjecture would readily imply the semantic counterpart of \autoref{thm:Ekq_tw-td-informal}. Unfortunately, Roberson's conjecture is wide open and has been confirmed only for the class of all planar graphs \cite{roberson_oddomorphisms_2022}, $\mathcal{TW}_{k-1}$ \cite{neuen_homomorphism-distinguishing_2023}, graphs of degree at most two \cite{roberson_oddomorphisms_2022}, paths \cite{roberson_oddomorphisms_2022}, $\{K_{2,h}, K_4\}$-minor-free graphs \cite{seppelt_homomorphism_2024},
and for graph classes which are essentially finite \cite{seppelt_logical_2023}, cf.\ \cite[Chapter 6]{seppelt_homomorphism_2024}.
We add to this short list of examples:

\begin{restatable}{thm}{tdclosed}
	\label{thm:td-closed}
	For $q \geq 1$, the class $\mathcal{TD}_q$ is homomorphism distinguishing closed.
\end{restatable}

Hence, as the intersection of homomorphism distinguishing closed graph classes \cite{roberson_oddomorphisms_2022}, $\TW_{k-1}\cap\TD_q$ is homomorphism distinguishing closed as well.
Finally, we show that $\Ekq$ is also homomorphism distinguishing closed.
\begin{restatable}{thm}{ekqclosed}
	\label{thm:ekq-closed}
	For $k,q \geq 1$, the class $\Ekq$ is homomorphism distinguishing closed.
\end{restatable}
Thereby, we lift the syntactic separation from \autoref{thm:Ekq_tw-td-informal} to a separation of the homomorphism indistinguishability relations $\equiv_{\Ekq}$ and $\equiv_{\TW_{k-1}\cap\TD_q}$ for all $q$ sufficiently larger than $k$.
\begin{restatable}{cor}{corSemantic}\label{cor:semantic}
	For $q$ sufficiently larger than $k$, $\equiv_{\Ekq}$ is strictly coarser than $\equiv_{\TW_{k-1}\cap\TD_q}$.
\end{restatable}

The proofs of Theorems \ref{thm:td-closed} and \ref{thm:ekq-closed} follow a strategy laid out by Neuen \cite{neuen_homomorphism-distinguishing_2023}.
Given a graph $G \not\in \TW_{k-1}$, he shows that the Cai--Fürer--Immerman graphs $G_0$ and $G_1$ of $G$ are homomorphism indistinguishable over $\TW_{k-1}$ and admit a different number of homomorphisms from $G$.
CFI graphs \cite{Cai_optimal_1992,roberson_oddomorphisms_2022} are highly similar graphs constructed from $G$ by replacing edges and vertices by suitable gadgets.
The difference between $G_0$ and $G_1$ is a \enquote{twist} located at the gadget corresponding to some vertex of $G$.
In order to show that $G_0$ and $G_1$ are equivalent in some counting logic fragments,
one must argue that the Duplicator player has a winning strategy in a suitable model comparison game played on $G_0$ and $G_1$.
This strategy essentially amounts to Duplicator hiding the twist from its opponent who tries to locate it.
Thus, the model comparison game on $G_0$ and $G_1$ reduces to a graph searching game played on the graph $G$.

\paragraph{Monotonicity of Cops-and-Robber games.}
Both treewidth and treedepth can be characterised using \emph{graph searching games} as introduced by Parsons and Petrov 
in \cite{Parsons_pursit-evasion_1978,Parsons_searchnumber_1978,Petrov_pursuit_1982}.
In these games a fugitive moving on a graph tries to evade capture by a set of searchers.
Different rules on the movement of both the searchers and the fugitive lead to characterizations of different parameters in directed and undirected graphs \cite{Bienstock_CRgames_1989, Bienstock_pathwidth_1991, Fomin_NCR_2009, Makedon_cutwidth_1989, Mazoit_monotonicity_2008, Johnson_dtw_2001, Nesetril_CRtd_2006, Seymour_graph_1993}.
For treewidth this is the \emph{Cops-and-Robber game} as introduced in \cite{Bienstock_CRgames_1989, Seymour_graph_1993}.
We introduce a suitable variant of this game and we use it to characterise the graph class $\Ekq$.
The equivalence of the first and last assertion of \autoref{thm:Ekq_equiv-informal} is what allows us to prove \autoref{thm:ekq-closed}.

\begin{restatable}{thm}{thmEkqEquivInformal}
	\label{thm:Ekq_equiv-informal}
	For any graph $G$ and any $k,q\in \NN$, the following are equivalent:
	\begin{enumerate}
		\item There exists a $k$-pebble forest cover of depth $\leq q$ of $G$, i.e.\ $G \in \Ekq$.
		\item There exists a tree-decomposition of $G$ with width $\leq k-1$ and depth $\leq q$.
		\item Cop wins the monotone $k$-cops $q$-round Cops-and-Robber game on $G$.
		\item Cop wins the $k$-cops $q$-round Cops-and-Robber game on $G$.
	\end{enumerate}
\end{restatable}

Note that we use `Cop wins' for `Cop has a winning strategy'.
As it is the case for treewidth,
a tree-decomposition straightforwardly gives rise to a winning strategy for Cop.
The challenge is to prove the converse direction, i.e.\ 
that a winning strategy for Cop gives rise to a tree-decomposition.

To that end, it is crucial to show that the Cops-and-Robber game can, without loss of generality, be restricted to its \emph{monotone} version, 
that is, Cop always has a winning strategy in which a previously cleared area never needs to be searched again.
The monotonicity of a graph searching game is a non-trivial property.
There are games where the two variants are equivalent, such as the games corresponding to treewidth \cite{Seymour_graph_1993} and tree depth \cite{Giannopoulou_lifo-search_2012}, as well as games where they are not, such as games corresponding to directed treewidth \cite{Kreutzer_digraph-decomp_2O11} or hypertreewidth \cite{Adler_monotone-marshal_2004, Gottlob_marshals_2003}.
For treewidth the proof of monotonicity in \cite{Seymour_graph_1993} builds on the dual concept of brambles.
In \cite{Mazoit_monotonicity_2008} the authors give a proof of this monotonicity that does not rely on any dual object, \cite{Adler_games_2009} gives a similar proof for a more general variant of the Cops-and-Robber game.
Building on these proof ideas we can show that the game variant that characterises the graph class \Ekq\ is monotone.

In order to do that, we introduce the notion of \emph{\preTreeDec}, based on tree-decompositions of matroids \cite{Hlineny_mtw_2006, Hlineny_addendum-mtw_2009}.
To characterise treewidth using the matroid decomposition, one equips a tree with a bijection from its leafs to the edges of the graph.
Then every node of the tree induces a partition of the edges and the width of such a decomposition is the maximum size of the boundary of such a partition minus one.
We show how one can construct such a \preTreeDec\ from a winning strategy of Cop and give a top-down construction that turns such a \preTreeDec\ into a tree-decomposition.

\paragraph*{Outline.}
We start in \autoref{sec:deep-wide-graph-dec} with three different decompositions and prove that they all yield the same graph class \Ekq.
In \autoref{sec:deep-wide-game} we introduce a variant of the Cops-and-Robber game and show that it exactly captures \Ekq.
This includes a construction that turns a possibly non-monotone strategy of Cop into a tree-decomposition.
We conclude the section with a proof that the graph classes $\Ekq$ and $\TW_{k-1}\cap\TD_q$ are different.
We then turn to the \homInd\ relation induced by $\Ekq$ in \autoref{sec:deep-wide-homind}.
We give a new proof that this relation is exactly $\lC^k_q$-equivalence.
We use the same proof technique to describe $\lGC^k_q$-equivalence in terms of \homInd.
Lastly in \autoref{sec:separation-sem} we prove that we can lift the separation of \Ekq\ and $\TW_{k-1}\cap\TD_q$ to the semantic level, that is we show that the induced \homInd\ relations are different.
We do so by proving that the graph classes $\Ekq$ and $\mathcal{TD}_q$ are \homDistCl.

\section{Preliminaries}
\label{ch:prelim}
In this chapter we will introduce the notation used throughout this paper and state some important definitions and well-known results that this paper is based on.

\subsection{Notation for sets and functions}
By $[k]$ we denote the set $\{1, \dots, k\}$.
For a set $U$ and some $k\geq 1$, we write $2^U$ for the power set of $U$, $U^k$ and $U^{\leq k}$ 
for the set of all tuples over $U$ of length (at most) $k$ and, if $k\leq |U|$, we write $\binom{U}{k}$ and $\binom{U}{\leq k}$ for the set of all subsets of $U$ of size (at most) $k$.
We use bold letters to denote tuples.
The tuple elements are denoted by the corresponding regular letter together with an index.
For example, $\tup{a}$ stands for $(a_1, \dots, a_n)$.
For a fixed universe $U$ and some subset $Y\subseteq U$, we write $\complementOf{Y}$ to denote the complement of $Y$ with respect to $U$, that is $\complementOf{Y}\coloneqq U\setminus Y$.
$\partitions(U)$ denotes the set of all \emph{partitions} of $U$.
We also consider \emph{ordered partitions}, where we allow partitions to contain multiple (but finitely many) copies of the empty set, which we denote by $\partitionsEmpty(U)$.
Let $\pi=(X_1,\ldots,X_d)\in\partitionsEmpty(U)$ and $F\subseteq U$.
For $i\in[d]$, the (ordered) partition
\[
\pi_{X_i\ext F}\coloneqq (X_1\setminus F, \ldots, X_{i-1}\setminus F, X_i\cup F,X_{i+1}\setminus F,\ldots,X_d\setminus F),
\]
is called the \emph{$F$-extension in $X_i$ of $\pi$} (see \cite{Amini_submod-partition_09}).

For a function $f$, we denote the domain of $f$ by $\dom(f)$.
The \emph{image} of $f$ is the set $\img(f) \coloneqq \{f(x) \mid x \in \dom(f)\}$.
The \emph{restriction} of a function $f \colon A \to C$ to some set $B \subseteq A$ is the function $f\restrict{B} \colon B \to C$ with $f\restrict{B}(x) = f(x)$ for $x \in B$.
For functions $f \colon A \to C$, $g \colon B \to C$ that agree on $A \cap B$, we write $f \sqcup g$ for the union of $f$ and $g$, that is, the function mapping $x$ to $f(x)$ if $x \in A$ and to $g(x)$ if $x \in B$.
A partial function is denoted by $f\colon A \parto B$.
This corresponds to a function $f\colon A\to B \cup \{\bot\}$ and we set $\dom(f)\coloneqq \{ x\in A \mid f(x)\neq \bot\}$.
In some applications, a function $f$ depends on a set of parameters $\vec{y}$.
We write $f(\vec{x}\fktmid \vec{y})$ and interpret $f(\phantom{|} \cdot \fktmid \vec{y})$ as a function over the variables $\vec{x}$.
We drop the parameters, if they are clear from context.
For a finite set $U$, we say a set function $\kappa\colon 2^U\rightarrow \RR$ is \emph{normalized} if $\kappa(\emptyset)=0$, $\kappa$ is \emph{symmetric} if $\kappa(X)=\kappa(\complementOf{X})$, for all $X\subseteq U$, and $\kappa$ is \emph{submodular} if $\kappa(X)+\kappa(Y) \geq \kappa(X\cap Y)+\kappa(X\cup Y)$, for all $X,Y\subseteq U$.
A function $w\colon \partitionsEmpty(U)\rightarrow \NN$ is \emph{\submod} if, for all $\pi,\pi'\in \partitionsEmpty(U)$, for all sets $X\in \pi$ and $Y\in \pi'$ with $X\cup Y \neq U$, it holds that
\[w(\pi)+w(\pi')\geq w(\pi_{X\ext \complementOf{Y}})+w(\pi'_{Y\ext \complementOf{X}}).\]
This coincides with the definition of submodularity for symmetric set functions (see also \cite{Amini_submod-partition_09}). Indeed, if for all bipartitions it holds that $f(X,\complementOf{X})=\kappa(X)$, then we get
\begin{align*}
	\kappa(X) + \kappa(Y) = f(X,\complementOf{X}) + f(Y,\complementOf{Y}) &\geq f(X\cup\complementOfB{\complementOf{Y}},\complementOf{X}\setminus\complementOfB{\complementOf{Y}}) + f(Y\setminus\complementOf{X},\complementOf{Y}\cup\complementOf{X})\\
	& = f(X\cup Y,\complementOf{X}\setminus Y) + f(Y\cap X, \complementOf{Y}\cup\complementOf{X})\\
	& = \kappa(X\cup Y) + \kappa(X\cap Y).
\end{align*}

\subsection{Graphs}
A graph $G$ is a tuple $(V(G), E(G))$, with a finite set of vertices $V(G)$ and edges $E(G) \subseteq \binom{V(G)}{\leq 2}$.
We usually write $uv$ or $vu$ to denote the edge $\{u, v\} \in E(G)$.
We write $|G|\coloneqq|V(G)|$.
If \(G\) is clear from the context we write \(V,E\) instead of \(V(G),E(G)\).
Unless otherwise specified, all graphs are assumed to be \emph{loopless}, that is, $E(G)$ contains no singletons. 
We denote the class of all loopless graphs by~$\cG$.

A \emph{$k$-labelled graph} $G$ is a graph together with a partial function $\labfkt_G \colon [k] \parto V(G)$ that assigns labels from the set $[k] = \{1, \dots, k\}$ to vertices of $G$.
A label thus occurs at most once in a graph, a single vertex can have multiple labels, and not all labels have to be assigned.
By $\labels{G}=\img(\labfkt_G)$ we denote the set of labelled vertices of $G$.
We call a graph where every vertex has at least one label \emph{fully labelled}.
We denote the class of all $k$-labelled graphs by $\cG_k$.
For $\ell \in [k]$ and $v \in V(G)$, we write $G(\ell \to v)$ to denote the graph obtained from $G$ by setting $\labfkt_G(\ell) \coloneqq v$.
We can \emph{remove} a label $\ell$ from a graph $G$, which yields a copy $G'$ of $G$ where $\labfkt_{G'}(\ell) = \bot$ and $\labfkt_{G'}(\ell') = \labfkt_G(\ell')$ for all $\ell' \neq \ell$.
The \emph{product}\footnote{Categorically speaking, this is a coproduct of labelled graphs or a pushout of graphs.} $G\odot H$ of two labelled graphs is the graph obtained by taking the disjoint union of $G$ and $H$, identifying vertices with the same label, and suppressing any parallel edges that might be created. Any created self-loops are kept.

For an arbitrary vertex \(v\in V\), we write \(N(v\fktmid G)\coloneqq\{u\mid \{u,v\}\in E\}\) to denote the \emph{neighbourhood of \(v\)}.
For some vertex set $U\subseteq V$, we denote the \emph{boundary} of $U$ by \(N(U\fktmid G)\coloneqq\bigcup_{u\in U} N(u\fktmid G)\setminus U\) to denote the \emph{neighbourhood of \(U\)}.
For $X\subseteq E(G)$, we write $\boundary(X\fktmid G)\coloneqq\{v\in V(G)\mid \text{ there is an } e\in X \text{ and an } e'\notin X \text{ such that } v\in e\cap e' \}$.
For $U,W\subseteq V$ with $U\cap W=\emptyset$, we let $E(U,W\fktmid G) \coloneqq \{\{u,w\}\in E(G)\mid u\in U,w \in W \}$.
Furthermore, we let $E(v\fktmid G) \coloneqq \{ \{u,v\}\in E(G)\mid u\in V(G)\}$ and $E(U\fktmid G) \coloneqq E\cap\binom{U}{2}$.
We drop the parameter if $G$ is clear from the context.

We call $H$ a \emph{subgraph} of $G$ if $H$ can be obtained from $G$ by removing vertices and edges.
$H$ is a \emph{proper subgraph} if $H\neq G$.
For a subset $U\subseteq V(G)$, we write $G[U]$ to denote the \emph{subgraph of $G$ induced by $U$}, that is $G[U]\coloneqq(U,E(U\fktmid G))$.
$H$ is a \emph{minor} of $G$ if it can be obtained from $G$ by removing vertices, removing edges, and \emph{contracting} edges. We contract an edge $uv$ by removing it and identifying $u$ and $v$. For labelled graphs, the new vertex is labelled by the union of labels of $u$ and $v$.

\subsubsection{Trees and Forests}
A \emph{tree} is a graph where any two vertices are connected by exactly one path.
The disjoint union of a set of trees is called a \emph{forest}.
A \emph{rooted tree} $(T, r)$ is a tree $T$ together with some designated vertex $r \in V(T)$, the \emph{root} of $T$.
By \(L(T)\) we denote the set of all \emph{leaves} of \(T\), that is \(L(T)\coloneqq\{ v\in V(T) \mid |N(v)|=1\}\).
We often interpret a rooted tree $(T,r)$ as a directed tree where all edges are directed away from the root and the leaves of a rooted tree are all those vertices with no outgoing neighbours.
For $t\in V(T)\setminus\{r\}$ we write $p_t$ for the parent of $t$.
All vertices that are not leaves are called \emph{inner vertices}.
The \emph{depth} of a rooted tree is equal to the number of vertices on the longest path from the root to a leaf\footnote{In the literature the height of a tree is often the number of edges on the longest path from the root to a leaf. In this paper the number of vertices is the more convenient.}.
The depth of a rooted forest is the maximum depth over all its connected components.

At times, the following alternative definition is more convenient.
We can view a rooted forest $(F, \tup{r})$ as a pair $(V(F), \preceq)$, where $\preceq$ is a partial order on $V(F)$ and for every $v \in V(F)$ the elements of the set $\{u \in V(F) \mid u \preceq v\}$ are pairwise comparable: The minimal elements of $\preceq$ are precisely the roots of $F$, and for any rooted tree $(T, r)$ that is part of $F$ we let $v \preceq w$ if $v$ is on the unique path from $r$ to $w$.

The depth of a rooted forest $(F, \tup{r})$ is then given by the length of the longest $\preceq$-chain.
Then a rooted tree $(T', r')$ is a \emph{subtree} of a tree $(T, r)$ if $V(T') \subseteq V(T)$ and $\preceq^{T'}$ is the restriction of $\preceq^T$ to $V(T')$.
Note that the subgraph of $T$ induced by $V(T')$ might not be a tree, since the vertices of $T'$ can be interleaved with vertices that do not belong to $T'$.
We call a subtree $T'$ of $T$ \emph{connected} if its induced subgraph on $T$ is connected.

\subsubsection{Decompositions}

\begin{defi}
	\label{def:graph_decomp}
	A \emph{tree-decomposition} of a graph $G$ ist a pair $(T,\beta)$, where $T$ is a tree and $\beta\colon V(T)\rightarrow 2^{V(G)}$ is a function from the nodes of $T$ to sets of vertices of $G$ such that
	\begin{enumerate}[label=(TD.\arabic*), labelindent=0pt,itemindent=*,leftmargin=*]
		\item $\bigcup_{t\in V(T)} G[\beta(t)] = G$, and \label{ax:GraphDec1}
		\item for every vertex $v \in G$, the graph $T_v\coloneqq T[\{t\in V(T)\mid v\in \beta(t)\}]$ is connected. \label{ax:GraphDec2}
	\end{enumerate}
	The sets $\beta(t)$ are called the \emph{bags} of this tree-decomposition.
\end{defi}

In order to distinguish easily between $V(T)$ and $V(G)$ in \autoref{def:graph_decomp}, we say $V(G)$ are the \emph{vertices} of $G$ and $V(T)$ are the \emph{nodes} of $T$.
For an example of a tree-decomposition, see \autoref{fig:treewidthtrees}.
The following lemma is a direct implication from \ref{ax:GraphDec2}.

\begin{lem}
	\label{lem:GraphDecConnSubgraph}
	Let $\Fdec$ be a tree-decomposition of a graph $G$.
	Then for any connected subgraph $H$ of $G$, the graph $T_H\coloneqq T[\{t\in V(T)\mid V(H)\cap \beta(t)\neq\emptyset\}]$ is connected. 
\end{lem}

Let $G$ be a graph and $(T,\beta)$ a tree-decomposition of $G$.
The width of the decomposition is $\wid(T,\beta)=\max_{t\in V(T)} |\beta(t)|-1$, the \emph{treewidth} $\tw(G)$ of a graph $G$ is the minimum width over all tree-decompositions of $G$.
Treewidth is a structural graph parameter that measures how close a graph is to being a tree.
We denote the class of graphs of treewidth at most $k$ by $\TW_k$.
\begin{figure}
	\begin{subfigure}[t]{0.6\textwidth}
		\centering
		\begin{tikzpicture}[nodes={draw, minimum size=.25cm, inner sep=1pt}, scale=.8]
			\node[anchor=north east, draw=none] at (-0.5, 1.75) {\llap{\textbf{(a)}}};
			\node at (0, 0) (v1) {
				\begin{tikzpicture}[scale=.1, node font=\tiny]
					\node[draw=none] at (0, 0) {$(1,3)$};
					\node[draw=none] at (0, -5) {$(2,3)$};
				\end{tikzpicture}
			};
			
			\node at (2, 1) (v21) {
				\begin{tikzpicture}[scale=.1, node font=\tiny]
					\node[draw=none] at (0, 0) {$(1,3)$};
					\node[draw=none] at (0, -5) {$(2,3)$};
					\node[draw=none] at (-7, -5) {$(2,2)$};
				\end{tikzpicture}
			};
			\node at (2, -1) (v22) {
				\begin{tikzpicture}[scale=.1, node font=\tiny]
					\node[draw=none] at (0, 0) {$(1,3)$};
					\node[draw=none] at (0, -5) {$(2,3)$};
					\node[draw=none] at (7, -5) {$(2,4)$};
				\end{tikzpicture}
			};
			
			\node at (4, 1) (v31) {
				\begin{tikzpicture}[scale=.1, node font=\tiny]
					\node[draw=none] at (0, 0) {$(1,3)$};
					\node[draw=none] at (-7,0) {$(1,2)$};
					\node[draw=none] at (-7, -5) {$(2,2)$};
				\end{tikzpicture}
			};
			\node at (4, -1) (v32) {
				\begin{tikzpicture}[scale=.1, node font=\tiny]
					\node[draw=none] at (0, 0) {$(1,3)$};
					\node[draw=none] at (7,0) {$(1,4)$};
					\node[draw=none] at (7, -5) {$(2,4)$};
				\end{tikzpicture}
			};
			
			\node at (6, 1) (v41) {
				\begin{tikzpicture}[scale=.1, node font=\tiny]
					\node[draw=none] at (0, 0) {$(1,2)$};
					\node[draw=none] at (0, -5) {$(2,2)$};
					\node[draw=none] at (-7, -5) {$(2,1)$};
				\end{tikzpicture}
			};
			\node at (6, -1) (v42) {
				\begin{tikzpicture}[scale=.1, node font=\tiny]
					\node[draw=none] at (0, 0) {$(1,4)$};
					\node[draw=none] at (0, -5) {$(2,4)$};
					\node[draw=none] at (7, -5) {$(2,5)$};
				\end{tikzpicture}
			};
			
			\node at (8, 1) (v51) {
				\begin{tikzpicture}[scale=.1, node font=\tiny]
					\node[draw=none] at (0, 0) {$(1,2)$};
					\node[draw=none] at (-7,0) {$(1,1)$};
					\node[draw=none] at (-7, 5) {$(2,1)$};
				\end{tikzpicture}
			};
			\node at (8, -1) (v52) {
				\begin{tikzpicture}[scale=.1, node font=\tiny]
					\node[draw=none] at (0, 0) {$(1,4)$};
					\node[draw=none] at (7,0) {$(1,5)$};
					\node[draw=none] at (7, -5) {$(2,5)$};
				\end{tikzpicture}
			};
			
			\draw (v51) -- (v41) -- (v31) -- (v21) -- (v1) -- (v22) -- (v32) -- (v42) -- (v52);
			
		\end{tikzpicture}
		\phantomsubcaption
		\label{fig:treewidthtrees}
	\end{subfigure}
	~
	\begin{subfigure}[t]{0.39\textwidth}
		\centering
		\begin{tikzpicture}[scale=0.5,smallVertex/.style={fill=black, inner sep=1, circle}]
			\node[anchor=north east] at (-3.5, -1) {\llap{\textbf{(b)}}};
			\node[smallVertex] (a03) at (0,-1) {};
			\node[smallVertex] (a13) at (0,-2) {};
			\node[smallVertex] (a01) at (-2,-3) {};
			\node[smallVertex] (a11) at (-2,-4) {};
			\node[smallVertex] (a05) at (2,-3) {};
			\node[smallVertex] (a15) at (2,-4) {};
			\node[smallVertex] (a00) at (-3,-5) {};
			\node[smallVertex] (a10) at (-3,-6) {};
			\node[smallVertex] (a02) at (-1,-5) {};
			\node[smallVertex] (a12) at (-1,-6) {};
			\node[smallVertex] (a06) at (3,-5) {};
			\node[smallVertex] (a16) at (3,-6) {};
			\node[smallVertex] (a04) at (1,-5) {};
			\node[smallVertex] (a14) at (1,-6) {};
			
			\draw (a10) -- (a00) -- (a11) -- (a02) -- (a12);
			\draw (a16) -- (a06) -- (a15) -- (a04) -- (a14);
			\draw (a11) -- (a01) -- (a13) -- (a05) -- (a15);
			\draw (a03) -- (a13);
			
			\foreach \i in {0,1,2,3,4,5,6}{
				\foreach \j in {0,1}{
					\ifthenelse{\NOT \i=0}{
						\tikzmath{
							integer \l;
							\l = \i - 1;}
						\draw[dashed, orange] (a\j\i) -- (a\j\l);
					}{}
				}
				\draw[dashed, orange] (a0\i) -- (a1\i);
			}
		\end{tikzpicture}
		\phantomsubcaption
		\label{fig:treedepthexample}
	\end{subfigure}
	\caption{Tree-decomposition and forest cover of grids.
		\textbf{(a)} A tree-decomposition of the grid $\grid{2}{5}$ of width $3$.
		\textbf{(b)} A  forest cover of the grid $\grid{2}{7}$ of depth $6$.
		The edges of the original grid are dashed.}
	\label{fig:grid-decompositions}
\end{figure}
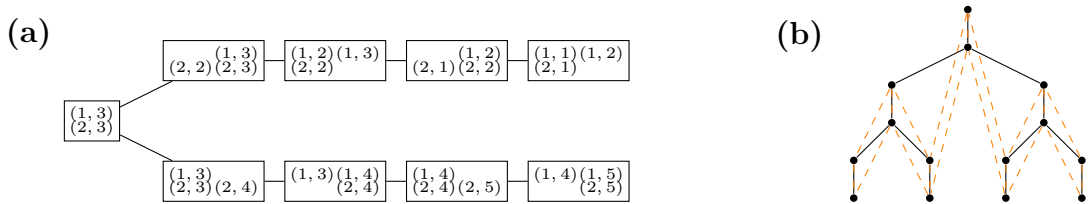

Treedepth can be thought of as measuring how close a graph is to being a star.
Alternatively, we may think of it as extending the notion of height beyond rooted forests.
It is defined for a graph $G$ as the minimum height of a forest $F$ over the vertices of $G$, such that all edges in $G$ have an ancestor-descendant relationship in $F$.

\begin{defi}
	A rooted forest $(F, \tup{r})$ with $V(F) = V(G)$ is a \emph{forest cover} of a graph $G$, if for every edge $uv \in E(G)$ it holds that either $u \preceq v$ or $v \preceq u$.
\end{defi}

The \emph{treedepth} $\td(G)$ of $G$ is the minimum depth of a forest cover of $G$.
We denote the class of all graphs of treedepth at most $q$ by $\TD_q$. 
For an example of a forest cover, see \autoref{fig:treedepthexample}.

It is possible to construct a tree-decomposition from a forest cover $(F, \tup{r})$.
This is achieved by considering a path of bags, each containing the vertices on a path from $\tup{r}$ to the leaves of $F$.
It is not hard to see that there is an ordering of these bags that satisfies the conditions of \autoref{def:graph_decomp}.
This yields the following relation between treedepth and treewidth.

\begin{lem}
	For every graph $G$, it holds that $\tw(G) \leq \td(G) - 1$.
\end{lem}

\subsubsection{Cops-and-Robber}

Both treewidth and treedepth enjoy characterisations in terms of node searching games, the so called \emph{Cops-and-Robber games}.
The general Cops-and-Robber game is played on a graph~$G$ by two players: Cop, controlling a number of cops; and Robber, controlling a single robber.
Cop and the Robber are positioned on vertices of $G$.
The goal of Cop is to place a cop on Robber's position, while Robber tries to avoid capture by moving along paths free from cops.
The players play in rounds where first Cop announces the next position(s) of the cops (with possible restriction on how many cops may be moved and where they may be positioned) and then Robber moves the robber along some path avoiding all vertices where before and after his move there is a cop.
Treewidth can be characterised by the minimum number of cops needed to capture the robber where neither the movement of Cop nor Robber is restricted (see e.g. \cite{Seymour_graph_1993}).
A well-known characterization of treedepth is the minimum number of cops needed in a stationary game where a cop cannot be moved after it is positioned on the graph (see e.g. \cite{Giannopoulou_lifo-search_2012}).
It is equivalent to count the number of vertices that a cop is placed on during a game.
Therefore we use the unified definition of \emph{$q$-rounds $k$-Cops-and-Robber game} $\CRkq(G)$, where $G$ is a non-empty graph and $k,q\geq 1$.
Cop positions are sets $X\in \binom{V(G)}{\leq k}$, Robber positions are vertices $v\in G$.
We define a function $\escape\colon \binom{V(G)}{\leq k}\times V(G) \rightarrow 2^{V(G)}$ as follows
\begin{equation*}
	\escape(X,v) \coloneqq \begin{cases}
		\{v\} & \text{if } v\in X,\\
		V(C) & \begin{split}
			&\text{otherwise, with } C \text{ the connected component }\\
			&\text{of } G\setminus X \text{ such that } v\in V(C).
		\end{split}
	\end{cases}
\end{equation*}
We say $\escape(X,v)$ is the \emph{escape space} of Robber.
We call a pair of Cop position and Robber position or escape space of Robber after one round a \emph{positions of the game} and we write $(X,v)$ or $(X,\escape(X,v))$.
The game is initialized with Cop on the empty set and Robber on a vertex of his choice, that is the initial position is $(\emptyset,V(C))$, for some connected component $C$ of $G$.
In round $i$ Cop can move from position $X_{i-1}$ to $X_i$ if $|X_i\setminus X_{i-1}|\leq1$.
Robber can move to a new position $u\in \escape(X_i\cap X_{i-1},v)$.
Cop wins if at any point the play reaches a position $(X,v)$ with $v\in X$, we say \emph{Robber is caught}.
Robber wins if Cop did not win after $q$ rounds.
If Cop plays in such a way that $\escape(X_i\cap X_{i-1},v)=\escape(X_i,v)$ in every move, we call the game monotone and write $\monCRkq(G)$.

In some applications we need a slightly different Cops-and-Robber game, a game where Robber can hide in an edge, rather than a vertex.
In this game Robber is caught if he hides in an edge where both endpoints are occupied by a cop and Cop wins immediately if the graph contains no edges.
We denote the game by $\eCRkq(G)$ or $\moneCRkq(G)$ respectively.
In this game the escape spaces of Robber are connected subsets of edges, that is we define $\escapeE\colon \binom{V(G)}{\leq k}\times E(G)\rightarrow 2^{E(G)}$ as
\begin{equation*}
	\escapeE(X,uv) \coloneqq \begin{cases}
		\{uv\} & \text{if } u,v\in X,\\
		\{xy\in E(G)\mid x\in V(C)\} & \begin{split}
			&\text{otherwise, with } C \text{ the conn. component } \\&\text{of } G\setminus X \text{ such that } \{v,u\}\cap V(C)\neq\emptyset.
		\end{split}
	\end{cases}
\end{equation*}

It is easy to see that if a graph $G$ contains at least one edge, Cop wins $\CRkq(G)$ if and only if Cop wins $\eCRkq(G)$.
In a graph without edges Cop wins one round earlier.
We can circumvent this difference by playing on a graph $G^\circ$, that is obtained from $G$ by adding all self-loops, we denote this game by $\eCRkq(G^\circ)$.
The game played on the graph $G^\circ$ corresponds to a game on $G$, where Robber can hide both inside a vertex or an edge.
The same holds for the monotone game.

We write $\CR_q(G)$ instead of $\CR_q^q(G)$ and $\CR^k(G)$ instead of $\CR_{|V(G)|}^k(G)$.
Treewidth and treedepth can be expressed in terms of the existence of winning strategies in these games.

\begin{lem}[{\cite[Theorem 1.4]{Seymour_graph_1993}} and {\cite[Theorem 4]{Giannopoulou_lifo-search_2012}}] \label{lem:games}
	Let $G$ be a graph. Let $k,q \geq 1$.
	\begin{enumerate}
		\item $G$ has treewidth at most $k$  if and only if Cop has a winning strategy for $\CR^{k+1}(G)$  if and only if Cop has a winning strategy for $\monCR^{k+1}(G)$.
		\item $G$ has treedepth at most $q$  if and only if Cop has a winning strategy for $\CR_q(G)$  if and only if Cop has a winning strategy for $\monCR_q(G)$.
	\end{enumerate}
\end{lem}

\subsubsection{Logic of graphs}

We will mainly consider \emph{counting first-order logic} $\lC$. $\lC$ extends regular first-order logic $\lFO$ by quantifiers $\exists^{{\geq} t}$, for $t \in \NN$. Consequently, we can build a $\lC$-formula in the usual way from atomic formulae; variables $x_1, x_2, \dots$; logical operators $\land, \lor, \impl, \neg$; and quantifiers $\forall, \exists, \exists^{{\geq} t}$. The atomic formulae in the language of graphs are $E\alpha\beta$ and $\alpha = \beta$ for arbitrary variables $\alpha, \beta$.

An occurrence of a variable $x$ is called \emph{free} if it is not in the scope of any quantifier. The \emph{free variables} $\free(\varphi)$ of a formula $\varphi$ are precisely those that have a free occurrence in $\varphi$. A formula without free variables is called a \emph{sentence}. We often write $\varphi(x_1, \dots, x_n)$ to denote that the free variables of $\varphi$ are among $x_1, \dots, x_n$. For a graph $G$, it usually depends on the interpretation of the free variables whether $G \models \varphi(x_1, \dots, x_n)$. We write $G, v_1, \dots, v_n \models \varphi(x_1, \dots, x_n)$ or $G \models \varphi(v_1, \dots, v_n)$ if $G$ satisfies $\varphi$ when $x_i$ is interpreted by $v_i$. We might also give an explicit \emph{interpretation function} $\fI \colon \free(\varphi) \to V(G)$, writing $G, \fI \models \varphi$.

We generalise the notion of $\lC$-equivalence, writing $G, v_1, \dots, v_n \equiv_\lC H, w_1, \dots, w_n$ to denote that for all formulae $\varphi(x_1, \dots, x_n) \in \lC$ it holds that $G \models \varphi(v_1, \dots, v_n)$ if and only if $H \models \varphi(w_1, \dots, w_n)$. Note that for labelled graphs, such an interpretation function is implicit: If the indices of $\free(\varphi)$ are a subset of the labels of $G$, then we can interpret the variables $x_i$ by the vertex with the label $i$, that is, $\fI(x_i) = \nu(i)$. The semantics of $\lC$ can then be stated succinctly in terms of label assignments.

\begin{defi}[$\lC$ semantics of labelled graphs]
	Let $\varphi \in \lC$ and let $G$ be a labelled graph, such that $\nu(i) \in V(G)$ for all $x_i \in \free(\varphi)$. Then $G \models \varphi$ if
	\begin{itemize}
		\item $\varphi$ is $x_i = x_j$ and $\nu(i) = \nu(j)$,
		\item $\varphi$ is $Ex_ix_j$ and $\nu(i)\nu(j) \in E(G)$,
		\item $\varphi$ is $\neg \psi$ and $G \not\models \psi$,
		\item $\varphi$ is $\psi \lor \theta$ and $G \models \psi$ or $G \models \theta$, or
		\item $\varphi$ is $\exists^{{\geq} t} x_\ell\, \psi(x_\ell)$ and there exist distinct $v_1, \dots, v_t$, such that $G(\ell \to v_i) \models \psi$ for all $i \in [t]$.
	\end{itemize}
\end{defi}

Note that for labelled graphs this is equivalent to extending the standard semantics of $\lFO$ by the following rule: $G, v_1, \dots, v_n \models \exists^{{\geq} t}y \psi(x_1, \dots, x_n, y)$ holds if there exist distinct elements $u_1, \dots, u_t \in V(G)$ such that $G \models \psi(v_1, \dots, v_n, u_i)$ for all $i \in [t]$.

We sometimes write $\exists^{=t}x \varphi(x)$ for $\exists^{{\geq} t} x \varphi(x) \land \neg\exists^{{\geq} t+1}x \varphi(x)$. We also write $\top$ for $\forall x (x=x)$ and $\bot$ for $\neg \top$.
As we already did above, we will often restrict ourselves to the connectives $\neg, \lor$ and the quantifier $\exists^{{\geq} t}$. This set of symbols is indeed equally expressive by De Morgan's laws and observing that $\exists x \varphi(x) \equiv \exists^{{\geq} 1}x \varphi(x)$ and $\forall x \varphi(x) \equiv \neg\exists x \neg\varphi(x)$.

The \emph{quantifier rank} $\qr(\cdot)$ of a formula is defined inductively as follows.
We set $\qr(\varphi) = 0$ for atomic formulae $\varphi$, and otherwise $\qr(\neg\varphi) = \qr(\varphi)$, $\qr(\varphi \lor \psi) = \max \{ \qr(\varphi), \qr(\psi)\}$ and $\qr(\exists^{{\geq} t}x \varphi) = 1 + \qr(\varphi)$.
The \emph{quantifier-rank-$q$ fragment} $\lC_q$ of counting first order logic consists of all formulae of quantifier rank at most $q$.

Instead of restricting the quantifer rank, we can also restrict the number of distinct variables that are allowed to occur in a formula. By $\lC^k$ we denote the \emph{$k$-variable fragment} of $\lC$, consisting of all formulae using at most $k$ different variables. Similarly, we define the \emph{$k$-variable quantifier-rank-$q$ fragment} as $\lC^k_q \coloneqq \lC^k \cap \lC_q$. Note that these are purely syntactic definitions.

\subsubsection{Homomorphisms}

Let $F,G$ be graphs.
A \emph{homomorphism} from $F$ to $G$ is a map $h \colon V(F) \to V(G)$ satisfying $uv \in E(F) \implies h(u)h(v) \in E(G)$.
For $k$-labelled graphs, we additionally require that $h(\labfkt_F(\ell)) = \labfkt_G(\ell)$ for all $\ell \in \dom(\labfkt_F)$.
We denote the set of homomorphisms from $F$ to $G$ by $\HOM(F, G)$.
The number of homomorphisms from $F$ to $G$ we denote by $\hom(F, G) \coloneqq \lvert \HOM(F, G) \rvert$.
We write $\HOM(F, G; a_1 \mapsto b_1, \dots, a_n \mapsto b_n)$ to denote the set of homomorphism $h\colon F \to G$ satisfying $h(a_i) = b_i$ for $i \in [n]$.
Two graphs $G$ and $H$ are \emph{homomorphism indistinguishable} over a graph class $\mathcal{F}$ if $\hom(F, G) = \hom(F, H)$ for all $F \in \mathcal{F}$.
See \cite{seppelt_homomorphism_2024} for a survey on \homInd.

\section{Graph decompositions accounting for treewidth\texorpdfstring{\\}{} and treedepth simultaneously}
\label{sec:deep-wide-graph-dec}
In this section, we reconcile treewidth and treedepth by introducing graph decompositions which account simultaneously for depth and width.
These efforts yield various equivalent characterizations of the graph class $\Ekq$, a subclass both of $\TW_{k-1}$ and $\TD_q$, the classes of graphs of treewidth $\leq k -1$ and treedepth $\leq q$, respectively.

We start with the original definition of the class \Ekq, which incorporates width into forest covers.
This definition has first been introduced as $k$-traversal in \cite{Abramsky_pebbling_2017}.

\begin{defi}
	\label{def:kpebbleforest}
	Let $G$ be  graph and $k\geq 1$.
	A \emph{$k$-pebble forest cover} of $G$ is a tuple $(F,\vec{r},p)$, where $(F,\vec{r})$ is a rooted forest over the vertices $V(G)$ and $p\colon V(G)\rightarrow [k]$ is a pebbling function such that:
	\begin{enumerate}[label=(FC.\arabic*), labelindent=0pt, itemindent=*, leftmargin=*]
		\item \label{ax:pebble-forest-edges} If $uv\in E(G)$, then $u\preceq v$ or $v\preceq u$ in $(F,\vec{r})$.
		\item \label{ax:pebble-forest-pebble} If $uv \in E(G)$ and $u\prec v$ in $(F,\vec{r})$, then for every $w\in V(G)$ with $u\prec w\preceq v$ in $(F,\vec{r})$ it holds that $p(u)\neq p(w)$.
	\end{enumerate}
	$(F,\vec{r},p)$ has depth $q\geq 1$ if $(F,\vec{r})$ has depth $q$.
	We write $\Ekq$ for the class of all graphs $G$ admitting a $k$-pebble forest cover of depth $q$.
\end{defi}

A different way to define the class \Ekq\ is by a new measure of the depth of a tree-decomposition $(T, \beta)$.
Crucially, it does not suffice to take the height of $T$ into account since this notion is not robust.
For example, it is well known that one can alter a tree-decomposition by subdividing any edge multiple times and copying the bag of the child node.
This transformation does neither change the width of the decomposition, nor does it affect the information how to decompose the graph. However, the height of the tree will change drastically under this transformation.
It turns out that the following is the right definition:

\begin{defi}
	\label{def:depth-treedec}
	Let $G$ be a graph.
	A tuple $(T,r,\beta)$ is a \emph{rooted tree-decomposition} of $G$
	if $(T,\beta)$ is a tree-decomposition of $G$ and $r \in V(T)$.
	The \emph{depth of a tree-decomposition $(T,\beta)$} is
	\begin{equation*}
		\dep(T,\beta)\coloneqq \min_{r\in V(T)} \dep(T,r,\beta)
		\qquad \text{with}
		\qquad
		\dep(T,r,\beta)\coloneqq \max_{v \in V(T)} \left| \bigcup_{t\preceq v} \beta(t) \right|.\qedhere
	\end{equation*}
\end{defi}

Lastly we define a construction inspired by Dvo\v{r}\'ak~\cite{Dvorak_recognizing_2010}, that enables us to use their proof technique to study the expressive power of first-order logic with counting quantifiers using homomorphism indistinguishability (see \autoref{fig:elim-example}).

\begin{defi}
	\label{def:Ekq}
	Let $G$ be a (possibly labelled) graph. A \emph{$k$-construction tree} for $G$ is a tuple $(T,\lambda,r)$, where $T$ is a tree rooted at $r$ and $\lambda\colon V(T)\rightarrow \cG_k$ is a function assigning $k$-labelled graphs to the nodes of $T$ such that:
	\begin{enumerate}[label=(CT.\arabic*), labelindent=0pt, itemindent=*, leftmargin=*]
		\item $\lambda(r)=G$, \label{ct1}
		\item all leaves $\ell\in V(T)$ are assigned fully labelled graphs, that is $V(\lambda(\ell))=\labels{\lambda(\ell)}$,
		\item all internal nodes $t\in V(T)$ with exactly one child $t'$ are \emph{elimination nodes}, that is $\lambda(t)$ can be obtained from $\lambda(t')$ by removing one label, and
		\item all internal nodes $t\in V(T)$ with more than one child are \emph{product nodes}, that is $\lambda(t)$ is the product of its children.
	\end{enumerate}
	The \emph{\elimDepth} of a construction tree $(T,\lambda,r)$ is the maximum number of elimination nodes on any path from the root $r$ to a leaf.
	If $G$ has a $k$-construction tree of \elimDepth\ $\leq q$, we say that $G$ is \emph{\elimOrd{k}{q}}.
	For the class of all $k$-labelled \elimOrd{k}{q} graphs we write \Lkq.
\end{defi}

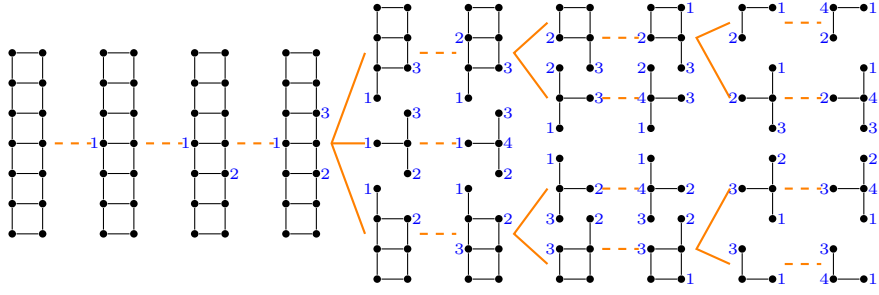
\begin{figure}
	\centering
	\begin{tikzpicture}[scale=0.4,smallVertex/.style={fill=black, inner sep=1, circle}, node font=\tiny]
		
		\foreach \i in {0,1,2,3,4,5,6}{
			\foreach \j in {0,1}{
				\tikzmath{
					integer \newj;
					\newj = \j-3;}
				\node[smallVertex] (b\j\i) at (\newj,\i) {};
				\ifthenelse{\NOT \i=0}{
					\tikzmath{
						integer \l;
						\l = \i - 1;}
					\draw (b\j\i) -- (b\j\l);
				}{}
			}
			\draw (b0\i) -- (b1\i);
		}
		
		\foreach \i in {0,1,2,3,4,5,6}{
			\foreach \j in {0,1}{
				\node[smallVertex] (a\j\i) at (\j,\i) {};
				\ifthenelse{\NOT \i=0}{
					\tikzmath{
						integer \l;
						\l = \i - 1;}
					\draw (a\j\i) -- (a\j\l);
				}{}
			}
			\draw (a0\i) -- (a1\i);
		}
		\node[blue] at (-0.3,3) {$1$};
		
		\foreach \i in {0,1,2,3,4,5,6}{
			\foreach \j in {0,1}{
				\tikzmath{
					integer \newj;
					\newj = \j+3;}
				\node[smallVertex] (b\j\i) at (\newj,\i) {};
				\ifthenelse{\NOT \i=0}{
					\tikzmath{
						integer \l;
						\l = \i - 1;}
					\draw (b\j\i) -- (b\j\l);
				}{}
			}
			\draw (b0\i) -- (b1\i);
		}
		\node[blue] at (2.7,3) {$1$};
		\node[blue] at (4.3,2) {$2$};
		
		\foreach \i in {0,1,2,3,4,5,6}{
			\foreach \j in {0,1}{
				\tikzmath{
					integer \newj;
					\newj = \j+6;}
				\node[smallVertex] (b\j\i) at (\newj,\i) {};
				\ifthenelse{\NOT \i=0}{
					\tikzmath{
						integer \l;
						\l = \i - 1;}
					\draw (b\j\i) -- (b\j\l);
				}{}
			}
			\draw (b0\i) -- (b1\i);
		}
		\node[blue] at (5.7,3) {$1$};
		\node[blue] at (7.3,2) {$2$};
		\node[blue] at (7.3,4) {$3$};
		
		\foreach \i in {0,1,2}{
			\foreach \j in {0,1}{
				\tikzmath{
					integer \newj;
					\newj = \j+9;
					\newi = \i+5.5;}
				\node[smallVertex] (c\j\i) at (\newj,\newi) {};
				\ifthenelse{\NOT \i=0}{
					\tikzmath{
						integer \l;
						\l = \i - 1;}
					\draw (c\j\i) -- (c\j\l);
				}{}
			}
			\draw (c0\i) -- (c1\i);
		}
		\node[smallVertex] (c) at (9,4.5) {};
		\draw (c) -- (c00);
		\node[blue] at (8.7,4.5) {$1$};
		\node[blue] at (10.3,5.5) {$3$};
		
		\foreach \i in {0,1,2}{
			\foreach \j in {0,1}{
				\tikzmath{
					integer \newj;
					\newj = \j+9;
					\newi = \i-1.5;}
				\node[smallVertex] (d\j\i) at (\newj,\newi) {};
				\ifthenelse{\NOT \i=0}{
					\tikzmath{
						integer \l;
						\l = \i - 1;}
					\draw (d\j\i) -- (d\j\l);
				}{}
			}
			\draw (d0\i) -- (d1\i);
		}
		\node[smallVertex] (d) at (9,1.5) {};
		\draw (d) -- (d02);
		\node[blue] at (8.7,1.5) {$1$};
		\node[blue] at (10.3,0.5) {$2$};
		
		\node[smallVertex] (e1) at (9,3) {};
		\node[smallVertex] (e2) at (10,3) {};
		\node[smallVertex] (e3) at (10,4) {};
		\node[smallVertex] (e4) at (10,2) {};
		\draw (e1) -- (e2) -- (e3) (e2) -- (e4);
		\node[blue] at (8.7,3) {$1$};
		\node[blue] at (10.3,2) {$2$};
		\node[blue] at (10.3,4) {$3$};
		
		\foreach \i in {0,1,2}{
			\foreach \j in {0,1}{
				\tikzmath{
					integer \newj;
					\newj = \j+12;
					\newi = \i+5.5;}
				\node[smallVertex] (c\j\i) at (\newj,\newi) {};
				\ifthenelse{\NOT \i=0}{
					\tikzmath{
						integer \l;
						\l = \i - 1;}
					\draw (c\j\i) -- (c\j\l);
				}{}
			}
			\draw (c0\i) -- (c1\i);
		}
		\node[smallVertex] (c) at (12,4.5) {};
		\draw (c) -- (c00);
		\node[blue] at (11.7,6.5) {$2$};
		\node[blue] at (11.7,4.5) {$1$};
		\node[blue] at (13.3,5.5) {$3$};
		
		\foreach \i in {0,1,2}{
			\foreach \j in {0,1}{
				\tikzmath{
					integer \newj;
					\newj = \j+12;
					\newi = \i-1.5;}
				\node[smallVertex] (d\j\i) at (\newj,\newi) {};
				\ifthenelse{\NOT \i=0}{
					\tikzmath{
						integer \l;
						\l = \i - 1;}
					\draw (d\j\i) -- (d\j\l);
				}{}
			}
			\draw (d0\i) -- (d1\i);
		}
		\node[smallVertex] (d) at (12,1.5) {};
		\draw (d) -- (d02);
		\node[blue] at (11.7,1.5) {$1$};
		\node[blue] at (11.7,-0.5) {$3$};
		\node[blue] at (13.3,0.5) {$2$};
		
		\node[smallVertex] (e1) at (12,3) {};
		\node[smallVertex] (e2) at (13,3) {};
		\node[smallVertex] (e3) at (13,4) {};
		\node[smallVertex] (e4) at (13,2) {};
		\draw (e1) -- (e2) -- (e3) (e2) -- (e4);
		\node[blue] at (11.7,3) {$1$};
		\node[blue] at (13.3,2) {$2$};
		\node[blue] at (13.3,4) {$3$};
		\node[blue] at (13.3,3) {$4$};
		
		\foreach \i in {0,1}{
			\foreach \j in {0,1}{
				\tikzmath{
					integer \newj;
					\newj = \j+15;
					\newi = \i+6.5;}
				\node[smallVertex] (c\j\i) at (\newj,\newi) {};
				\ifthenelse{\NOT \i=0}{
					\tikzmath{
						integer \l;
						\l = \i - 1;}
					\draw (c\j\i) -- (c\j\l);
				}{}
			}
			\draw (c0\i) -- (c1\i);
		}
		\node[smallVertex] (c) at (16,5.5) {};
		\draw (c) -- (c10);
		\node[blue] at (14.7,6.5) {$2$};
		\node[blue] at (16.3,5.5) {$3$};
		
		\foreach \i in {0,1}{
			\foreach \j in {0,1}{
				\tikzmath{
					integer \newj;
					\newj = \j+15;
					\newi = \i-1.5;}
				\node[smallVertex] (d\j\i) at (\newj,\newi) {};
				\ifthenelse{\NOT \i=0}{
					\tikzmath{
						integer \l;
						\l = \i - 1;}
					\draw (d\j\i) -- (d\j\l);
				}{}
			}
			\draw (d0\i) -- (d1\i);
		}
		\node[smallVertex] (d) at (16,0.5) {};
		\draw (d) -- (d11);
		\node[blue] at (14.7,-0.5) {$3$};
		\node[blue] at (16.3,0.5) {$2$};
		
		\node[smallVertex] (e1) at (16,1.5) {};
		\node[smallVertex] (e2) at (15,1.5) {};
		\node[smallVertex] (e3) at (15,2.5) {};
		\node[smallVertex] (e4) at (15,0.5) {};
		\draw (e1) -- (e2) -- (e3) (e2) -- (e4);
		\node[blue] at (16.3,1.5) {$2$};
		\node[blue] at (14.7,0.5) {$3$};
		\node[blue] at (14.7,2.5) {$1$};
		
		\node[smallVertex] (e1) at (16,4.5) {};
		\node[smallVertex] (e2) at (15,4.5) {};
		\node[smallVertex] (e3) at (15,5.5) {};
		\node[smallVertex] (e4) at (15,3.5) {};
		\draw (e1) -- (e2) -- (e3) (e2) -- (e4);
		\node[blue] at (16.3,4.5) {$3$};
		\node[blue] at (14.7,3.5) {$1$};
		\node[blue] at (14.7,5.5) {$2$};
		
		\foreach \i in {0,1}{
			\foreach \j in {0,1}{
				\tikzmath{
					integer \newj;
					\newj = \j+18;
					\newi = \i+6.5;}
				\node[smallVertex] (c\j\i) at (\newj,\newi) {};
				\ifthenelse{\NOT \i=0}{
					\tikzmath{
						integer \l;
						\l = \i - 1;}
					\draw (c\j\i) -- (c\j\l);
				}{}
			}
			\draw (c0\i) -- (c1\i);
		}
		\node[smallVertex] (c) at (19,5.5) {};
		\draw (c) -- (c10);
		\node[blue] at (17.7,6.5) {$2$};
		\node[blue] at (19.3,5.5) {$3$};
		\node[blue] at (19.3,7.5) {$1$};
		
		\foreach \i in {0,1}{
			\foreach \j in {0,1}{
				\tikzmath{
					integer \newj;
					\newj = \j+18;
					\newi = \i-1.5;}
				\node[smallVertex] (d\j\i) at (\newj,\newi) {};
				\ifthenelse{\NOT \i=0}{
					\tikzmath{
						integer \l;
						\l = \i - 1;}
					\draw (d\j\i) -- (d\j\l);
				}{}
			}
			\draw (d0\i) -- (d1\i);
		}
		\node[smallVertex] (d) at (19,0.5) {};
		\draw (d) -- (d11);
		\node[blue] at (17.7,-0.5) {$3$};
		\node[blue] at (19.3,0.5) {$2$};
		\node[blue] at (19.3,-1.5) {$1$};
		
		\node[smallVertex] (e1) at (19,1.5) {};
		\node[smallVertex] (e2) at (18,1.5) {};
		\node[smallVertex] (e3) at (18,2.5) {};
		\node[smallVertex] (e4) at (18,0.5) {};
		\draw (e1) -- (e2) -- (e3) (e2) -- (e4);
		\node[blue] at (19.3,1.5) {$2$};
		\node[blue] at (17.7,0.5) {$3$};
		\node[blue] at (17.7,2.5) {$1$};
		\node[blue] at (17.7,1.5) {$4$};
		
		\node[smallVertex] (e1) at (19,4.5) {};
		\node[smallVertex] (e2) at (18,4.5) {};
		\node[smallVertex] (e3) at (18,5.5) {};
		\node[smallVertex] (e4) at (18,3.5) {};
		\draw (e1) -- (e2) -- (e3) (e2) -- (e4);
		\node[blue] at (19.3,4.5) {$3$};
		\node[blue] at (17.7,3.5) {$1$};
		\node[blue] at (17.7,5.5) {$2$};
		\node[blue] at (17.7,4.5) {$4$};
		
		\node[smallVertex] (c1) at (21,6.5) {};
		\node[smallVertex] (c2) at (21,7.5) {};
		\node[smallVertex] (c3) at (22,7.5) {};
		\draw (c1) -- (c2) -- (c3);
		\node[blue] at (20.7,6.5) {$2$};
		\node[blue] at (22.3,7.5) {$1$};
		
		\node[smallVertex] (d1) at (21,-0.5) {};
		\node[smallVertex] (d2) at (21,-1.5) {};
		\node[smallVertex] (d3) at (22,-1.5) {};
		\draw (d1) -- (d2) -- (d3);
		\node[blue] at (20.7,-0.5) {$3$};
		\node[blue] at (22.3,-1.5) {$1$};
		
		\node[smallVertex] (e1) at (21,1.5) {};
		\node[smallVertex] (e2) at (22,1.5) {};
		\node[smallVertex] (e3) at (22,2.5) {};
		\node[smallVertex] (e4) at (22,0.5) {};
		\draw (e1) -- (e2) -- (e3) (e2) -- (e4);
		\node[blue] at (20.7,1.5) {$3$};
		\node[blue] at (22.3,0.5) {$1$};
		\node[blue] at (22.3,2.5) {$2$};
		
		\node[smallVertex] (e1) at (21,4.5) {};
		\node[smallVertex] (e2) at (22,4.5) {};
		\node[smallVertex] (e3) at (22,5.5) {};
		\node[smallVertex] (e4) at (22,3.5) {};
		\draw (e1) -- (e2) -- (e3) (e2) -- (e4);
		\node[blue] at (20.7,4.5) {$2$};
		\node[blue] at (22.3,3.5) {$3$};
		\node[blue] at (22.3,5.5) {$1$};
		
		\node[smallVertex] (c1) at (24,6.5) {};
		\node[smallVertex] (c2) at (24,7.5) {};
		\node[smallVertex] (c3) at (25,7.5) {};
		\draw (c1) -- (c2) -- (c3);
		\node[blue] at (23.7,6.5) {$2$};
		\node[blue] at (25.3,7.5) {$1$};
		\node[blue] at (23.7,7.5) {$4$};
		
		\node[smallVertex] (d1) at (24,-0.5) {};
		\node[smallVertex] (d2) at (24,-1.5) {};
		\node[smallVertex] (d3) at (25,-1.5) {};
		\draw (d1) -- (d2) -- (d3);
		\node[blue] at (23.7,-0.5) {$3$};
		\node[blue] at (25.3,-1.5) {$1$};
		\node[blue] at (23.7,-1.5) {$4$};
		
		\node[smallVertex] (e1) at (24,1.5) {};
		\node[smallVertex] (e2) at (25,1.5) {};
		\node[smallVertex] (e3) at (25,2.5) {};
		\node[smallVertex] (e4) at (25,0.5) {};
		\draw (e1) -- (e2) -- (e3) (e2) -- (e4);
		\node[blue] at (23.7,1.5) {$3$};
		\node[blue] at (25.3,0.5) {$1$};
		\node[blue] at (25.3,2.5) {$2$};
		\node[blue] at (25.3,1.5) {$4$};
		
		\node[smallVertex] (e1) at (24,4.5) {};
		\node[smallVertex] (e2) at (25,4.5) {};
		\node[smallVertex] (e3) at (25,5.5) {};
		\node[smallVertex] (e4) at (25,3.5) {};
		\draw (e1) -- (e2) -- (e3) (e2) -- (e4);
		\node[blue] at (23.7,4.5) {$2$};
		\node[blue] at (25.3,3.5) {$3$};
		\node[blue] at (25.3,5.5) {$1$};
		\node[blue] at (25.3,4.5) {$4$};
		
		\draw[dashed, thick, orange] (-1.6,3) -- (-0.4,3) (1.4,3) -- (2.6,3) (4.4,3) -- (5.6,3) (10.4,6) -- (11.6,6) (10.4,3) -- (11.6,3) (10.4,0) -- (11.6,0) (16.4,6.5) -- (17.6,6.5) (16.4,4.5) -- (17.6,4.5) (16.4,1.5) -- (17.6,1.5) (16.4,-0.5) -- (17.6,-0.5) (22.4,7) -- (23.6,7) (22.4,4.5) -- (23.6,4.5) (22.4,1.5) -- (23.6,1.5) (22.4,-1) -- (23.6,-1);
		\draw[thick, orange] (8.6,0) -- (7.5,3) -- (8.6,6) (7.5,3) -- (8.6,3) (14.6,4.5) -- (13.5,6) -- (14.6,7) (14.6,-1) -- (13.5,0) -- (14.6,1.5) (20.6,4.5) -- (19.5,6.5) -- (20.6,7) (20.6,-1) -- (19.5,-0.5) -- (20.6,1.5);
	\end{tikzpicture}
	\caption{A $4$-construction tree for the grid $\grid{2}{7}$ of elimination depth $6$. Edges entering elimination nodes are dashed.}
	\label{fig:elim-example}
\end{figure}

It turns out that all three notions coincide, that is given some fixed $k,q\geq 1$ for unlabelled graphs they define the same graph class.
The equivalence of \autoref{def:Ekq} and \autoref{def:depth-treedec} is proven by a careful choice of the tree-decomposition where one can identify the bags with the labelled vertices of the construction tree.
For the equivalence of \autoref{def:depth-treedec} and \autoref{def:kpebbleforest}, we follow the proof of \cite[Theorem~19]{Abramsky_relating_2021} and observe that their construction preserves depth.
The following theorem extends \autoref{thm:Ekq_equiv-informal}, equivalences 1 and 2.
\newpage

\begin{thm}
	\label{thm:Ekq_equiv}
	Let $k, q\geq 1$. For every graph $G$, the following are equivalent:
	\begin{enumerate}
		\item $G$ is \elimOrd{k}{q},\label{ekqEquiv1}
		\item $G$ has a tree-decomposition of width $k-1$ and depth $q$,\label{ekqEquiv2}
		\item $G\in\Ekq$, that is $G$ admits a $k$-pebble forest cover of depth $q$.\label{ekqEquiv3}
	\end{enumerate}
\end{thm}

\begin{proof}
	We first show $(1)\Leftrightarrow (2)$.
	
	Let $G$ be \elimOrd{k}{q} witnessed by a $k$-construction tree $(T,\lambda,r)$ of \elimDepth\ $\leq q$.
	For every node $t\in V(T)$ and the corresponding $k$-labelled graph $\lambda(t)$, we write $\beta(t)$ for the set of labelled vertices of $\lambda(t)$.
	We show that $(T,\beta,r)$ is a tree-decomposition of $G$ of width $\leq k-1$ and depth $\leq q$.
	We observe that the number of vertices that are in any bag of a path from the root to some leaf equal the number of times a label was removed on the same path and therefore is bounded by the \elimDepth.
	Thus one only has to show that this is indeed a tree-decomposition.
	For every edge $uv\in E(G)$, there has to be some leaf $\ell\in V(T)$, such that the edge is already present in $\lambda(\ell)$. Hence, $u,v \in \beta(\ell)$.
	The same holds for all isolated vertices and \ref{ax:GraphDec1} holds.
	Finally, we have that $\beta^{-1}(v)$ is the subset of nodes $t\in V(T)$ such that $v\in V(\lambda(t))$ and $v$ is labelled in $\lambda(t)$.
	This subtree is obviously connected, as labels can only be introduced at leaf nodes and only labelled vertices are identified at join nodes, thus \ref{ax:GraphDec2} holds.
	
	Conversely, let $(T,r,\beta)$ be a rooted tree-decomposition of $G$ of width $\leq k-1$ and depth $\leq q$.
	A rooted tree-decomposition is \emph{nice} if every node that is not a leaf is either an introduce node, a forget node or a join node.
	We call a node $t$ 
	\begin{itemize}
		\item \emph{introduce node} if it has exactly one child $s$ and 
		there exists a vertex $v\in V(G) \setminus \beta(s)$ with $\beta(t)=\beta(s)\cup \{v\} $,
		\item \emph{forget node} if it has exactly one child $s$ and there exists a vertex $v\in V(G) \setminus \beta(t)$ with $\beta(s)=\beta(t)\cup \{v\}$, and 
		\item \emph{join node} if it has exactly two children $s_1,s_2$ and $\beta(t)=\beta(s_1)=\beta(s_2)$.
	\end{itemize}
	Additionally, the bag of the root node and of all leaf nodes are required to be empty.
	By \cite{Bodlaender_partial_1998}, if there is a tree-decomposition then there also is a nice tree-decomposition.
	We observe that the technique to make a tree-decomposition nice preserves the depth of the decomposition.
	Thus w.l.o.g.\ $(T,r,\beta)$ is nice.
	We observe that a join node of a nice tree-decomposition is quite similar to the product node of a construction tree, as is the forget node to the eliminate node.
	But we need to get rid of the introduce nodes.
	We construct a new rooted tree-decomposition $(T',r,\beta')$, where at every introduce node $t$ we append a new leaf $\ell_t$, with bag $\beta'(\ell_t)=\beta(t)$.
	Furthermore we set $\beta'(t)=\beta(t)$, for every $t\in V(T)$.
	For every $t\in V(T')$, we now set $\lambda(t):=G[\gamma(t)]$, that is the subgraph induced by all vertices that are in a bag of the subtree above $t$.
	We set $\beta(t)$ to be the labelled vertices of $\lambda(t)$.
	It remains to define a colouring function $c\colon V(G)\rightarrow [k]$, such that for all $t\in V(T')$, it holds that $c|_{\beta(t)}$ is injective.
	We define this colouring via traversing the tree $T'$ from the root to the leafs.
	Whenever $t$ is a forget node with child $s$, such that $\beta(s)\setminus\beta(t)=\{v\}$, we set $c(v)$ to be the smallest value that is not used for any vertex in $\beta(t)$.
	Then $T'$ together with the graphs $H_t$, that are labelled via the function $c|_{\beta(t)}^{-1}$ is the desired elimination order.
	
	Next we show $(2)\Leftrightarrow (3)$.
	To do that, we use the same construction as in the proof of \cite[Theorem~19]{Abramsky_relating_2021}, where the authors gave a proof that a tree-decomposition of width $\leq k-1$ exists if and only if a $k$-pebble forest cover exists.
	We recall their construction and prove that it preserves depth.
	
	Let $(T,\beta)$ be a tree-decomposition of $G$, of width $\leq k-1$ and depth $\leq q$, and $r\in V(T)$ such that $\dep(T,\beta) = \dep(T,r,\beta)$.
	Again w.l.o.g.\ $(T,r,\beta)$ is nice.
	We define the function $\tau\colon V(G) \rightarrow V(T)$ to map every vertex of $G$ to the unique node of $t$ such that $v\in\beta(t)$ that has smallest distance to the root $r$.
	As $(T,r,\beta)$ is nice, this function is injective.
	The forest cover $(F,\vec{r}')$ is induced by the partial order on the image of $\tau$ with respect to $(T,r)$.
	
	We observe that for all $v\in V(G)$ and $t\in V(T)$ from $v\in\beta(t)$ it follows that $\tau(v)\preceq t$, thus \autoref{ax:pebble-forest-edges} holds as every edge of $G$ is covered by some bag.
	The pebbling function $p\colon V(G) \rightarrow [k]$ is defined inductively.
	Assume $p$ is defined for all $v'\prec v$.
	Then we get $p(v)= \min ([k] \setminus \{ p(v')\mid v'\in \beta(\tau(v))\setminus\{v\}\})$.
	Thus \autoref{ax:pebble-forest-pebble} holds again as every edge of $G$ is covered by some bag.
	We observe that the depth of $(F,\vec{r}')$ is the length of the longest chain in the image of $\tau$ with respect to $(T,r)$.
	Thus the depth of $(F,\vec{r}')$ equals the maximum number of vertices from the root $r$ to any leaf of $T$ and thus the depth of $(T,r,\beta)$.
	
	Now let $(F, \vec{r},p)$ be a $k$-pebble forest cover of depth $\leq q$ of $G$.
	We construct a rooted tree $(T,r')$ by introducing a new root $r'$ that connects to all nodes in $\vec{r}$.
	We define $\beta(r')=\emptyset$ and $\beta(t)\coloneqq \{u\preceq t\mid \text{ for all } w\in V(G), u\prec w\preceq t \Rightarrow p(u)\neq p(w) \}$, for every $t\in V(F)$.
	We have that \[\dep(T,r,\beta)= \max_{v \in V(T)} \left| \bigcup_{t\in P_v} \beta(t) \right| = \max_{t\in V(F)}|\{s\in V(F)\mid s\preceq t \}|=\dep(F,\vec{r},p).\qedhere\]
\end{proof}

\begin{cor}
	\label{cor:Ekq-subset-tw-td}
	\label{cor:minor-closed}
	For $k, q\geq 1$, the graph class
	$\Ekq $ is minor\-/closed, closed under taking disjoint unions, and a subclass of $\TW_{k-1}\cap\TD_q$.
\end{cor}
\begin{proof}
	Per definition, every graph $G \in \Ekq$ has a forest cover of depth $q$, so in particular treedepth at most $q$. By \autoref{thm:Ekq_equiv}, it equivalently has a tree-decomposition of width~$k$, so in particular treewidth at most $k-1$. It follows that $\Ekq \subseteq \TW_{k-1} \cap \TD_q$. Taking disjoint unions of graphs in $\Ekq$ corresponds to joining two $k$-construction trees under a product node, which does not increase the elimination depth. The disjoint union of two graphs in $\Ekq$ is thus also contained in $\Ekq$. 
	
	Finally, consider some graph $G$ with tree-decomposition $(T, \beta)$. If $G'$ is obtained from $G$ by removing an edge, then $(T, \beta)$ is still a tree-decomposition of $G'$. If $G'$ is obtained from $G$ by removing a vertex $v$, then $(T, \beta')$ with $\beta'(t) = \beta(t) \setminus \{v\}$ for all $t\in V(T)$ is a tree-decomposition for $G'$. If $G'$ is obtained from $G$ by contracting the edge $uv$, then replacing all instances of $u$ and $v$ in the bags by the newly created vertex yields a tree-decomposition for $G'$. None of these operations increases width or depth of the tree-decomposition, which through \autoref{thm:Ekq_equiv} implies that $\Ekq$ is minor-closed.
\end{proof}

Dawar, Jakl, and Reggio reduced the proofs of the results of Grohe and Dvo\v{r}\'ak \cite{Dvorak_recognizing_2010,Grohe_counting_2020} to a \enquote{combinatorial core} \cite[Remark~17]{dawar_lovasz-type_2021}, which amounts to showing that the classes $\TW_k$ and $\TD_q$ are closed under contracting edges.
To that end, \autoref{cor:minor-closed} illustrates the benefits of characterizing $\Ekq$ in terms of tree-decompositions (\autoref{def:depth-treedec}): Proving that pebble forest covers are preserved under edge contractions requires a non-trivial amount of bookkeeping while the analogous statement for tree-decomposition is straightforward.

\section{A Cops-and-Robber game}
\label{sec:deep-wide-game}

In this section we introduce a characterization of $\Ekq$ in terms of a Cops\-/and\-/Robber game.
We find that Robber does not gain an advantage if Cop is forced to play monotonely.
In order to prove that, we introduce a \emph{\preTreeDec} and lift the definition of depth of a tree-decomposition to matroid tree-decompositions.
Lastly, building upon a result from \cite{Furer_rounds_2001}, we show that $\Ekq$ is a proper subclass of $\TW_{k-1} \cap \TD_q$ if $q$ is sufficiently larger than $k$.
The following Theorem is a formal restatement of \autoref{thm:Ekq_equiv-informal} equivalences 1, 3 and 4.

\begin{thm}
	\label{thm:Ekq-cops}
	Let $G$ be a graph and $q,k\geq 1$.
	The following are equivalent:
	\begin{enumerate}
		\item Cop wins the game $\monCR_q^{k}(G)$,
		\item Cop wins the game $\CR_q^{k}(G)$, and
		\item $G\in \Ekq$.
	\end{enumerate}
\end{thm}

The implication from (1) to (2) is trivial.
The implication from (2) to (3) is where the the crux lies and Sections~\ref{sec:predec}, \ref{sec:game} and
\ref{sec:make-exact} are dedicated to the proof.
In Section~\ref{sec:predec} we introduce the notions of pre-tree-decomposition
and \emph{exact} pre-tree-decomposition,
together with relevant properties. In particular, 
in Lemma~\ref{lem:exact-and-tdec} we show that an exact pre-tree-decomposition gives rise to a tree decomposition and vice versa. This reduces the task to finding an exact pre-tree-decomposition. 
In Section~\ref{sec:game} we show
that a (non-monotone) winning strategy for Cop gives rise to a pre-tree-decomposition. In Section~\ref{sec:make-exact} we then show the pre-tree-decomposition can be turned into an \emph{exact} pre-tree-decomposition using a careful cleaning-up procedure. 
This then completes the proof.

We begin by proving the implication from (3) to (1).

\begin{lem}
	Let $G\in \Ekq$ be a graph and $q,k\geq 1$. Cop wins the game $\monCR_q^{k}(G)$.
\end{lem} 

\begin{proof}
	Suppose without loss of generality that $1\leq k\leq q$.
	
	As $G\in\Ekq$, there is a $k$-pebble forest cover $(F,\vec{r},p)$ of depth at most $q$.
	Write $\preceq$ for the order induced by $(F, \vec{r})$.
	For a set $Y\subseteq V(G)$ and a vertex $x\in V(G)$, 
	write $Y \preceq x$ if $y \preceq x$ for all $y \in Y$.
	Define $x \preceq Y$ analogously.
	For a set $Y \subseteq V(G)$ that is totally ordered by $\preceq$, write $\max_{\preceq}Y$ for its maximal element, i.e.\ the unique element $y \in Y$ such that $Y \preceq y$.
	Note that, for an arbitrary set $Y \subseteq  V(G)$, it may be that there is no such maximal element.
	
	Intuitively, Robber is forced away from a root to a leaf of the pebble forest cover by Cop occupying the reachable vertices below the robber's position.
	This is analogous to the definition of the tree decomposition induced by a pebble forest cover in the proof of \autoref{thm:Ekq_equiv}.
	For a vertex $v \in V(G)$, define the set
	\[
	\beta(v) \coloneqq \{u \preceq v \mid \forall w \in V(G), \text{if } u \prec w \preceq v, \text{then } p(u) \neq p(w)\}.
	\]
	Since there are at most $k$ pebble indices, it holds that $|\beta(v)| \leq k$ for every $v \in V(G)$.
	Furthermore, if $v \prec v'$ are such that there is no  vertex $v \prec v'' \prec v'$,
	then $|\beta(v') \setminus \beta(v)| \leq 1$.
	
	The game $\monCR_q^k(G)$ commences with the position $(\emptyset, y)$, where $y$ is the vertex chosen by Robber.
	Throughout the game, Cop keeps track of a vertex $r_1 \in V(G)$, which is initially set to the $\preceq$-minimal vertex such that $r_1 \preceq y$, i.e.\ $r$ is the root of the connected component of $\preceq$ containing $y$.
	
	The first position Cop assumes is $\beta(r_1)$. 
	At any later stage $i \geq 1$, 
	Cop moves to $\beta(r_{i+1})$ where $r_{i+1}$ is the $\preceq$-minimal vertex such that $r_i \prec  r_{i+1} \preceq y$
	and $y$ is the vertex occupied by Robber.
	
	It nees to be argued that this strategy is well-defined, i.e.\ that, at any stage, Robber's position $y$ satisfies $r_i \preceq y$.
	By definition, this holds for $i = 1$.
	For the inductive step, suppose that $r_i \preceq y$. 
	Cop announces to move to $\beta(r_{i+1})$.
	Now Robber must move to a vertex $y'$ reachable from $y$ via a path that does not traverse $\beta(r_i) \cup \beta(r_{i+1})$. 
	Let $y''$ denote any vertex on such a path that is $\preceq$-comparable with $y$.
	If $y \preceq y''$,
	then Robber is  caught or $r_{i+1} \preceq y''$, as desired.
	Hence, suppose that $y'' \prec y$.
	If $y''$ and $y$ are adjacent in $G$,
	then, by \ref{ax:pebble-forest-pebble},  $y''$ is the $\preceq$-maximal vertex preceeding $y$ that carries the pebble $p(y'')$.
	It follows that $r_{i+1} \preceq y''$ since otherwise $y'' \in \beta(r_{i+1})$ and Robber is caught.
	
	A similar argument applies when $y''$ and $y$ are not adjacent but $y''$ is merely on Robber's paths.
	It follows that $r_{i+1} \preceq y'$, as desired.
	
	Having established that Cop's strategy is well-defined and that Robber always occupies a vertex $y$ such that $r_i \preceq y$,
	it follows that Robber is caught after at most $q$ moves since $r_1 \prec r_2 \prec r_3 \prec \dots$ forms a strict chain whose length is bounded by the depth of $\prec$.
\end{proof}

\subsection{\PreTreeDec, exactness and submodularity}\label{sec:predec}

Here we consider a definition of tree-decompositions that is inspired by tree-decompositions of matroids.
We relax this definition into what we call a \preTreeDec.

\begin{defi}
	Let $G=(V(G),E(G))$ be a graph.
	Let $X\subseteq E(G)$.
	Let $\pi\in\partitionsEmpty(E(G))$.
	We define
	\[
	\boundary(\pi)\coloneqq\{v\in V(G)\mid \exists X\in \pi, v\in\boundary(X) \}.
	\]
	A tuple $(T,r,\beta,\gamma)$, where $(T,r)$ is a rooted tree, $\beta\colon V(T)\rightarrow 2^{V(G)}$ and $\gamma\colon \overrightarrow{E(T)}\rightarrow 2^{E(G)}$, is a \emph{rooted \preTreeDec} if:
	\begin{enumerate}[label=(PD.\arabic*), labelindent=0pt, itemindent=*, leftmargin=*]
		\item $\beta(r)=\emptyset$ and for every connected component $C$ of $G$, there is a child $c$ of the root with $\gamma(r,c)=E(C)$.\label{ax:preTreeRoot}
		\item For every leaf $\ell\in L(T)$ with parent $t$, it holds that $|\gamma(t,\ell)|\leq 1$.\label{ax:preTreeLeaf}
		\item For every internal node $t\in V(T)\setminus L(T)$, we define $\pi_t\coloneqq(\gamma(t,t_1),\ldots,\gamma(t,t_d))$, where $N(t)=\{t_1,\ldots,t_d\}$ an arbitrary enumeration of the neighbours of $t$, and for a leaf $\ell \in L(T)$ with parent $p$ we define $\pi_\ell\coloneqq(\gamma(\ell,p),\complementOf{\gamma(\ell,p)})$.
		For every $t\in V(T)$, the tuple $\pi_t$ is an ordered partition of $E(G)$ and $\boundary(\pi_t)\subseteq \beta(t)$.\label{ax:preTreePart}
		\item For every edge $st\in E(T)$, it holds that $\gamma(s,t)\cap \gamma(t,s) = \emptyset$.\label{ax:preTreeEdge}
	\end{enumerate}
	We call an edge $st\in E(T)$ \emph{exact} if $\gamma(s,t)\cup\gamma(t,s)=E(G)$, we call $(T,r,\gamma,\beta)$ \emph{exact}, if every edge is exact and $\beta(t)= \boundary(\pi_t)$, for all $t\in V(T)$.
	We call $\beta(t)$ the \emph{bag} at node $t$ and $\gamma(s,t)$ the \emph{cone} at edge $st$.
\end{defi}

The function $\gamma$ describes a partition of the edges of the graph at every inner node, whereas the function $\beta$ gives a vertex separator for this partition.
This separator may contain more vertices than necessary at a certain node, which is needed to define the depth of a \preTreeDec, as seen below.
For an edge $st\in E(T)$ with $s\prec t$, we can view $\gamma(s,t)$ as the set of edges that need to be decomposed in the subtree below and $\gamma(t,s)$ as the set of edges that is for sure decomposed somewhere else within the tree.
With this point of view the axioms correspond to the following ideas:
\begin{enumerate}[label=(PD.\arabic*), labelindent=0pt, itemindent=*, leftmargin=*]
	\item We start by separating the different connected components of the graph and assign one distinct subtree to decompose each component. This way we also ensure that all of the graph is decomposed in some subtree.
	\item We want to decompose the graph into single edges. We do allow empty leaves and even empty subtrees, to ease our cleaning up procedure in the following sections. The reader may recall that the root of a rooted tree is never a leaf, by definition.
	\item At every node of the tree we make sure that all edges of the original graph are accounted for and that $\beta$ is indeed a separator for the given partition.
	\item If a parent node assigns an edge to the set that still has to be decomposed, the child node can not assign this edge to the set that is already decomposed somewhere else. But the other direction is possible, if the parent node assigns an edge to the set that is decomposed somewhere else, the child can still assign it to one of its subtrees. If the latter is also not the case the edge is exact.
\end{enumerate}

Similar to the definition of width and depth for tree-decompositions we define the width and depth of a \preTreeDec.
We slightly adapt the definition of depth as \ref{ax:GraphDec2} does not hold in \preTreeDec s.

\begin{defi}
	The \emph{width} of a partition $\pi$ of the edges of a graph is
	\[
	\wid(\pi)\coloneqq|\boundary(\pi)|.
	\]
	The \emph{width} of a \preTreeDec\ is
	\[
	\wid(T,r,\beta,\gamma)\coloneqq\max_{t\in V(T)} |\beta(t)| - 1.
	\]
	The \emph{depth} of a rooted \preTreeDec\ is
	\[
	\dep(T,r,\beta,\gamma) \coloneqq \max_{t\in V(T)} \sum_{r\prec s\preceq t} |\beta(s) \setminus \beta(p_s)|.\qedhere
	\]
\end{defi}

The reader may note that the width of a \preTreeDec\ only gets smaller if one sets $\beta(t)\coloneqq \boundary(\pi_t)$, for all nodes $t\in V(T)$, but the depth can become larger.

\begin{exa}
	\label{ex:pre-tree-dec}
	Let $G$ be the graph that results from a $2\times 5$ grid when contracting the edge $(1,3)(2,3)$ into a single vertex $3$, see \autoref{fig:pre-tree-dec-graph}.
	A \preTreeDec\ of the graph $G^\circ$, that is $G$ with a self-loop attached to each vertex, can be seen in \autoref{fig:pre-tree-dec-dec}.
	The labels of the nodes are the bags and the labels of the edges the cones.
	All edges that are drawn dashed and in orange are not exact and for these edges the cones at both directions are depicted.
	The cones pointing towards and away from leaves are omitted, they are exactly the edge covered by the bag of the leaf.
	For all other edges only the cone for the direction pointing away from the root is depicted, the other cone is its complement.
	Vertices in the bags that are not part of the corresponding boundary are red.
	The blue vertices have a symmetric subtree.
	We omit one due to space restrictions.
	
	The width of this \preTreeDec\ is $4$, the depth is $9$.
	The reader may observe that if the bags where changed to $\beta'(t)\coloneqq \boundary(\pi_t)$ the depth would change to $10$, whilst the width would not change.
\end{exa}

\begin{figure}
	\begin{subfigure}[t]{\textwidth}
} (t);
			\draw[blue, dashdotted] (t) -- (9,-10) -- (11,-10) -- (t);
		\end{tikzpicture}
		\phantomcaption{~}
		\label{fig:pre-tree-dec-dec}
	\end{subfigure}
	\caption{Graph and corresponding \preTreeDec\ from \autoref{ex:pre-tree-dec}.}
\label{fig:pre-tree-dec}
\end{figure}

We show that the width of a partition of the edges as defined above is submodular.
We need this property to show that our main construction does not enlarge the width of the \preTreeDec.

\begin{lemC}[\protect{\cite[Proposition 7]{Amini_submod-partition_09}}]\label{lem:submodular}
	For every graph $G$, the width $\wid$ of partitions is \submod.
	That is,  for all $\pi,\pi'\in \partitionsEmpty(E(G))$, for all sets $X\in \pi$ and $Y\in \pi'$ with $X\cup Y \neq U$, it holds that
	\[\wid(\pi)+\wid(\pi')\geq \wid(\pi_{X\ext \complementOf{Y}})+\wid(\pi'_{Y\ext \complementOf{X}}).\]
\end{lemC}

\begin{proof}
	Let $\pi=(X_1,\ldots,X_d),\pi'=(Y_1,\ldots,Y_d)\in\partitionsEmpty(E(G))$.
	We prove that \[\wid(\pi)+\wid(\pi')\geq \wid(\pi_{X_1\ext \complementOf{Y_1}})+\wid(\pi'_{Y_1\ext \complementOf{X_1}}),\] which is enough to prove the lemma by symmetry.
	
	If $X_1=E(G)$, then $\pi=\pi_{X_1\ext \complementOf{Y_1}}$ and $\pi'=\pi'_{Y_1\ext \complementOf{X_1}}$, thus the lemma holds.
	
	If $X_1=\emptyset$, then $\pi'_{Y_1\ext\complementOf{X_1}}=(E(G),\emptyset,\ldots,\emptyset)$ and thus $\wid(\pi'_{Y_1\ext\complementOf{X_1}})=0$.
	Furthermore $\boundary(\pi_{X_1 \ext \complementOf{Y_1}}) \subseteq \boundary(\pi) \cup \boundary(Y_1) \subseteq \boundary(\pi) \cup \boundary(\pi')$ and thus $\wid(\pi_{X_1 \ext \complementOf{Y_1}}) \leq \wid(\pi) + \wid(\pi')$, thus the lemma holds.
	
	If $Y_1=E(G)$ or $Y_1=\emptyset$ the lemma holds analogously.
	
	Thus let $\emptyset\neq X_1,Y_1\neq E(G)$.
	Trivially we get that $\boundary(\pi_{X_1 \ext \complementOf{Y_1}}) \subseteq \boundary(\pi) \cup \boundary(Y_1)$ and $\boundary(\pi'_{Y_1 \ext \complementOf{X_1}}) \subseteq \boundary(\pi') \cup \boundary(X_1)$.
	Assume there exists some $v\in \boundary(\pi_{X_1 \ext \complementOf{Y_1}})\setminus \boundary(\pi)$, then $v\in \boundary(Y_1)$ and thus $v\in \boundary(\pi')$.
	Furthermore we get that $E(v) \cap X_1=\emptyset$ and thus it holds that $E(v)\subseteq Y_1 \cup \complementOf{X_1}$.
	But then $v\notin \boundary(\pi'_{Y_1 \ext \complementOf{X_1}})$.
	Analogously we can show that $\left(\boundary(\pi'_{Y_1 \ext \complementOf{X_1}})\setminus \boundary(\pi')\right) \cap \boundary(\pi_{X_1 \ext \complementOf{Y_1}})=\emptyset$.
	Thus all in all every vertex that is newly introduced to one of the boundaries $\boundary(\pi_{X_1 \ext \complementOf{Y_1}}), \boundary(\pi'_{Y_1 \ext \complementOf{X_1}})$ is removed from the other and therefore the lemma holds.
\end{proof}

We continue this section with some lemmas, that help us to get comfortable with the definition of a \preTreeDec\ and are useful to prove that our cleaning up procedure in the following sections is correct.
We start with a lemma about the cones along a path of exact edges.
It is a direct consequence of exactness and the fact that the cones incident to a vertex form a partition of the edges.

\begin{lem}
	\label{obs:exact-path}
	Let $(T,r,\beta,\gamma)$ be a \preTreeDec\ of a graph $G$.
	Let $P=t_1,\ldots,t_\ell$ be a path in $T$, such that every edge $t_it_{i+1}$, for $i\in[\ell-1]$, is exact.
	Then it holds that
	$\gamma(t_1,t_2) \supseteq \gamma(t_2,t_3) \supseteq \ldots \supseteq \gamma(t_{\ell-1},t_\ell)$.
\end{lem}

\begin{proof}
	We show that $\gamma(t_1,t_2) \supseteq \gamma(t_2,t_3)$ and then the lemma follows by induction over the path.
	As $\pi_{t_2}$ is a partition of the edges, we get that $\complementOf{\gamma(t_2,t_1)} \supseteq \gamma(t_2,t_3)$.
	By definition of exactness we have that $\complementOf{\gamma(t_2,t_1)}=\gamma(t_1,t_2)$, which concludes the proof.
\end{proof}

The following lemma shows, that \ref{ax:preTreePart} spreads over exact edges, that is any subtree of $T$ that only contains exact edges induces a partition of the edges of the original graph.

\begin{lem}
	\label{obs:exact-subtree}
	Let $(T,r,\beta,\gamma)$ be a \preTreeDec\ of a graph $G$ and let $(T',r')$ be a subtree of $(T,r)$, where $r'$ is the minimal node of $T'$ with respect to $\preceq$, such that all edges of $T'$ are exact.
	We pick arbitrary enumerations of $N(V(T')\fktmid T)\coloneqq \{t_1,\ldots,t_a\}$ and of $L(T)\cap L(T')\coloneqq \{\ell_1,\ldots,\ell_b\}$.
	For $U\coloneqq\{t_1,\ldots,t_a,\ell_1,\ldots,\ell_b\}$, 
	define $s\colon U\rightarrow V(T')$ by mapping $u \in U$ to its unique neighbour in $T'$.
	Then \[ (\gamma(s(t_1),t_1),\ldots, \gamma(s(t_a),t_a),\gamma(s(\ell_1),\ell_1), \ldots,\gamma(s(\ell_b),\ell_b)) \] is an ordered partition of $E(G)$.
\end{lem}

\begin{proof}
	Assume there exists some $e\in E(G)\setminus \bigcup_{u\in U} \gamma(s(u),u)$.
	As $\pi_{r'}$ is a partition of $E(G)$ there is some $s_1\in V(T)$ such that $e\in \gamma(r',s_1)$.
	By \ref{ax:preTreeEdge}, we know that $e\notin \gamma(s_1,r')$.
	If $s_1\notin V(T')$ or $s_1\in L(T')$ we have found a contradiction.
	Otherwise as $\pi_t$ is a partition of $E(G)$ there is some $s_2\in V(T)$ such that $e\in \gamma(s_1,s_2)$ and by \ref{ax:preTreeEdge}, we know that $e\notin \gamma(s_2,s_1)$.
	Again if $s_2\notin V(T')$ or $s_2\in L(T')$ we have found a contradiction.
	Using this argument we inductively construct a sequence $r',s_1,s_2,\ldots$ such that $e\in \gamma(s_i,s_{i+1})$ and $e\notin \gamma(s_{i+1},s_i)$.
	If $s_i\notin V(T')$ or $s_i\in L(T')$, for any $i>1$, we have found a contradiction, thus this sequence is infinite.
	This is a contradiction to $T'$ being a finite tree, thus our assumption was false and $E(G) = \bigcup_{u\in U} \gamma(s(u),u)$.
	
	Now assume there are distinct $t,t'\in U$ and there exists $e\in\gamma(s(t),t)\cap \gamma(s(t'),t')$.
	Let $s_1=s(t),s_2,\ldots,s_c=s(t')$ be the unique path from $s(t)$ to $s(t')$ in $T'$.
	From \autoref{obs:exact-path} we get that $e\in\gamma(s(t),s_2)$.
	Since $s_2\neq t$ this contradicts $\pi_{s(t)}$ being a partition.
	Thus for all distinct $t,t'\in U$ it holds that $\gamma(s(t),t)\cap \gamma(s(t'),t')=\emptyset$.
\end{proof}

The next lemma is needed in th etranslation of exact \preTreeDec s into tree-decompositions.
Furthermore it will help us bound the depth within our cleaning up procedure in the following sections.

\begin{lem}
	\label{obs:exact-subtree-depth}
	Let $(T,r,\beta,\gamma)$ be a \preTreeDec\ of a graph $G$ and let $(T',r')$ be a subtree of $(T,r)$, where $r'$ is the minimal node of $T'$ with respect to $\preceq$, such that all edges of $T'$ are exact.
	For every vertex $v\in V(G)$, it holds that the induced subgraph 
	$T'_v\coloneqq T[t\in V(T')\mid v\in \boundary(\pi_t)]$ is connected.
	In particular, if $r=r'$, for every $t\in V(T')$, it holds that
	\begin{equation*}
		\sum_{\substack{s\preceq t\\s\neq r}} \left\lvert \boundary(\pi_s) \setminus \boundary(\pi_{p_s}) \right\rvert
		= \left\lvert \bigcup_{s\preceq t} \boundary(\pi_s) \right\rvert.
	\end{equation*}
\end{lem}

\begin{proof}
	Assume there exists a $v\in V(G)$ such that $T'_v$ is not connected.
	Let $T_1,T_2$ be two disjoint connected components of $T'_v$ and let $P=t_1,\ldots,t_a$ be the shortest $T_1$-$T_2$-path in $T'$.
	Then $v\notin \boundary(\gamma(t_1,t_2))\subseteq \boundary(\pi_{t_2})$ and thus $E(v)\cap \gamma(t_1,t_2)=\emptyset$.
	As all edges in $P$ are exact it holds that $\gamma(t_1,t_2)\supseteq \gamma(t_a,s)$, for all $s\in N(t_a) \setminus\{t_{a-1}\}$.
	And thus it holds that $E(v)\cap \gamma(t_a,s)=\emptyset$ and $E(v)\subseteq \gamma(t_a,t_{a-1})$.
	This contradicts $v\in\bigcup_{s\in N(t_a)} \boundary(\gamma(t_a,s))$.
\end{proof}

We conclude this section with a lemma that shows that a \preTreeDec\ of a graph $G$ is indeed a relaxation of a tree-decomposition of $G$.
If every edge is exact and all bags are exactly the boundary of the partition, then we can construct a tree decomposition.
We need to start with a \preTreeDec\ of the graph $G^\circ$ with all self-loops added to ensure that every non-isolated vertex does appear in some bag and that the components corresponding to isolated vertices are covered by the \preTreeDec.
On the other hand we can transform a tree-decomposition into a \preTreeDec, by copying the tree-decomposition of each connected component of $G$ and adding leaves that correspond to the edges of $G^\circ$.

\begin{lem}\label{lem:exact-and-tdec}
	\label{lem:tw-ptw}
	Let $k,q\geq 1$.
	Let $G=(V,E)$ be a graph.
	Any tree-decomposition of $G$ of width $\leq k-1$ and depth $\leq q$ gives rise to an exact \preTreeDec\ of $G^\circ$ of width $\leq k-1$ and depth $\leq q$ and vice versa.
\end{lem}

\begin{proof}
	Let $(T,r,\beta,\gamma)$ be an exact \preTreeDec\ of $G^\circ$ of width $\leq k-1$ and depth $\leq q$.
	We define $\beta'\colon V(T)\rightarrow 2^{V(G)}$ as follows
	\begin{equation*}
		\beta'(t)\coloneqq\begin{cases}
			\{v\} & \text{ if } t\in L(T) \text{ and } r \text{ is parent of } t \text{ and } \gamma(r,t)=\{vv\},\\
			\beta(t) & \text{ otherwise.}
		\end{cases}
	\end{equation*}
	\begin{claim}
		$(T,\beta')$ is a tree-decomposition of width $\leq k-1$ and depth $\leq q$.
	\end{claim}
	\begin{claimproof}
		From \ref{ax:preTreeRoot}, \ref{ax:preTreeLeaf} and \autoref{obs:exact-subtree} applied to the complete tree $(T,r)$ we get that for every edge $uv\in E(G^\circ)$ there is some leaf $\ell$ with parent $p$ and $\gamma(p,\ell)=\{uv\}$.
		Thus if $u=v$, then $\beta'(\ell)=\{v\}$ and thus $vv\in E(G[\beta'(\ell)])$.
		Otherwise it holds that $uu,vv\in E(G^\circ)\setminus \{uv\}$ and thus $u,v\in\beta'(\ell)$ and $uv\in E(G[\beta'(\ell)])$.
		All in all we get that \ref{ax:GraphDec1} holds.
		
		By \autoref{obs:exact-subtree-depth} applied to the complete tree $(T,r)$ we know that all $T_v$ are connected.
		Therefore \ref{ax:GraphDec2} also holds and $(T,\beta')$ is a tree-decomposition.
		
		The width and depth are obvious as $k,q\geq 1$.
	\end{claimproof}
	
	Now let $(T,r,\beta)$ be a tree-decomposition of $G$ of width $\leq k-1$ and depth $\leq q$.
	W.l.o.g.\ $\beta$ is \emph{tight}. That is, for all $t\in V(T)$ and $v\in \beta(t)$, it holds that $(T,r,\beta')$ is not a tree-decomposition of $G$, where $\beta'(t)\coloneqq\beta(t)\setminus\{v\}$ and $\beta'(s)=\beta(s)$, for all $s\in V(T)\setminus \{t\}$.
	We construct a new tree $T'$ with root $r'$ and functions $\beta'\colon V(T')\rightarrow 2^{V(G)}$, $\gamma\colon \overrightarrow{E(T')} \rightarrow 2^{E(G^\circ)}$ and $f\colon V(T')\setminus\left(L(T')\cup\{r'\}\right)\rightarrow V(T)$ as follows.
	Let $C$ be a connected component of $G$ and let $V_C\coloneqq \{t\in V(T)\mid V(C)\cap \beta(t)\neq \emptyset\}$.
	By \autoref{lem:GraphDecConnSubgraph} $V_C$ is connected.
	If $C$ contains only an isolated vertex $v$, then $V_C=\{t\}$, for some $t\in V(T)$.
	We add a new node $t_v$ to $T'$ and connect it to the root.
	We set $\beta'(t_v)=\emptyset$, $\gamma(r',t_v)=\{vv\}$ and $\gamma(t_v,r')= E(G^\circ)\setminus \{vv\}$.
	Otherwise let $T_C$ be a copy of the subtree induced by $V_C$ with root $r_C$ and vertices $V^*_C$ and $f|_{V^*_C}\colon V^*_C\rightarrow V_C$ the natural bijection between the copies and their originals.
	We attach $r_C$ to the root $r'$.
	For every $v\in V(C)$, there is some $t_v\in V_C$ such that $v\in\beta(t_v)$, as $C$ is not an isolated vertex.
	We add a new leaf $t'_v$ that we attach to $f|_{V(T_C)}^{-1}(t_v)$ and set $\beta'(t'_v)=\{v\}$, $\gamma(f|_{V^*_C}^{-1}(t_v),t'_v)=\{vv\}$ and $\gamma(t'_v,f|_{V^*_C}^{-1}(t_v))=E(G^\circ)\setminus \{vv\}$.
	For every $e\in E_G(C)$ there is some $t_e\in V_C$ such that $e\subseteq\beta(t_e)$.
	We add a new leaf $t'_e$ that we attach to $f|_{V^*_C}^{-1}(t_e)$ and set $\beta'(t'_e)=e$, $\gamma(f|_{V^*_C}^{-1}(t_e),t'_e)=\{e\}$ and $\gamma(t'_e,f|_{V^*_C}^{-1}(t_e))=E(G^\circ)\setminus \{e\}$.
	For every node $t\in V^*_C$ with parent $p$ we add all edges $e\in E_G(C)$, where $t'_e$ is a descendant of $t$, and all self-loops $vv\in E_{G^\circ}(C)$, where $t'_v$ is a descendant of $t$, to $\gamma(p,t)$.
	Furthermore we set $\gamma(t,p)\coloneqq E(G^\circ)\setminus \gamma(p,t)$ and $\beta'(t)\coloneqq \boundary(\pi_t)\subseteq\beta(f(t))\cap V(C)$.
	By tightness of $\beta$ there is some $v\in \beta(f(\ell))$ such that $T_v=\{f(\ell)\}$, for every $\ell\in L(T_C)$, thus no leaf of $T_C$ is a leaf in $T'$, thus $(T',r',\beta',\gamma)$ satisfies \ref{ax:preTreeLeaf}.
	\ref{ax:preTreeRoot}, \ref{ax:preTreePart} and \ref{ax:preTreeEdge} hold by construction.
	Furthermore every edge is exact by construction.
	Thus we get that $(T',r',\beta',\gamma)$ is an exact \preTreeDec\ of $G^\circ$.
	
	The width is obvious as every bag in $\beta'$ is a subset of some bag in $\beta$.
	To see that the depth bound also holds we observe two things.
	For every leaf $\ell\in L(T')$ with parent $p$ we get that $\beta'(\ell)\setminus\beta'(p)=\emptyset$.
	For every inner node $t\in V(T')\setminus L(T')$ with parent $p$ we get that $\beta'(t)\setminus\beta'(p)\subseteq \beta(f(t))$ and, if $p\neq r'$, $\beta'(t)\setminus\beta'(p)\subseteq \beta(f(t))\setminus \beta(f(p))$, by the tightness of~$\beta$.
\end{proof}

\subsection{From edge game to \preTreeDec}\label{sec:game}

In this section we define how to construct a \preTreeDec\ from a winning strategy of Cop.
We assume that the next move of Cop only depends on the escape space of Robber and not on the exact edge he is positioned in.
We can assume this w.l.o.g. as we can simply pick one of the possible moves of Cop to be applied to all Robber positions as the reachable edges for Robber only depend on the escape space, not the exact edge.

\begin{defi}[strategy tree]
	\label{def:strat-tree}
	Let $G$ be a graph without isolated vertices and let $k,q\in\NN$.
	Let $\sigma\colon V(G)^{\leq k}\times E(G)\rightarrow V(G)^{\leq k}$ be a cop strategy such that for all $X\in V(G)^{\leq k}$, for all $uv\in E(G)$ and for all $u'v'\in\escapeE(X,{uv})$ the strategy agrees on $uv$ and $u'v'$, that is we have that $\sigma(X,uv)=\sigma(X,u'v')$.
	We write $\sigma(X,\escapeE(X,{uv}))$ instead of $\sigma(X,uv)$.
	
	The \emph{strategy tree of $\sigma$} is a \preTreeDec\ $(T,r,\beta,\gamma)$, inductively defined as follows:
	\begin{itemize}
		\item $\beta(r)=\emptyset$,
		\item For every connected component $C$ of $G$, there is a unique child $c$ of the root $r$. Define $\gamma(r,c) \coloneqq E(C)$.
		\item For every node $t\in V(T)\setminus\{r\}$ with parent $s\in V(T)$,
		\begin{itemize}
			\item if the robber player is caught, we set $\beta(t)= e$, where $\gamma(s,t)=\{e\}$,
			\item else $\beta(t)=\sigma(\beta(s),\gamma(s,t))$ and
			\item we add a new child $c$ to $t$ with $\gamma(t,c)=\escapeE(\beta(s)\cap\beta(t),uv)$, for every $uv\in\gamma(s,t)$, where $\escapeE(\beta(s)\cap\beta(t),uv)$ is disjoint from all previously added cones,
			\item we set $\gamma(t,s)\coloneqq
			\begin{cases}
				E(G)\setminus \bigcup_{c \text{ child of }t}\gamma(t,c), &\text{if } t\notin L(T),\\
				E(G)\setminus\gamma(s,t), &\text{if } t\in L(T).
			\end{cases}$
		\end{itemize}
	\end{itemize}
	We call $t\in V(T)$ a \emph{branching node} if the cop player placed a new cop incident to the robber escape space.
	
	Observe that if $t\in V(T)$ is a leaf, then Robber is captured and the depth of $(T,r,\beta,\gamma)$ is $\leq q$ if and only if $\sigma$ is a winning strategy in $\eCRkq(G)$.
\end{defi}

We observe that all inner nodes of the strategy tree directly correspond to Cop moves in a Cops-and-Robber game in the following sense: In every game on $G$ where the cop player plays according to the strategy $\sigma$ and every (reachable) position of the game $(X,\escape(X,uv))$, there exists an inner node $s$ of the tree together with a child $c$ such that $\beta(s)=X$ and $\escape(X,uv)=\gamma(s,c)$.
Additionally for every inner node $s$ with child $c$ the pair $(\beta(s),\gamma(s,c))$ is a position in a Cops-and-Robber game, more precisely in the game where the cop player always plays according to the strategy $\sigma$ and the robber player always moves to an edge in $\gamma(s_i,s_{i+1})$ where $r=s_0,s_1,\ldots,s_\alpha=s$ is the path from the root to $s$ and the current Cop position is $\beta(s_i)$.
Note that w.l.o.g. every child of the root is a branching node, as Cop w.l.o.g. only plays positions that are inside the component Robber chose in the first round.
If the game is played on $G^\circ$, then every branching node that does not correspond to the placement of a cop onto an isolated vertex has more than one child.

\begin{exa}
	Let $G$ be the graph from \autoref{ex:pre-tree-dec} (see \autoref{fig:pre-tree-dec-graph}).
	The \preTreeDec\ depicted in \autoref{fig:pre-tree-dec-dec} is a strategy tree of a winning strategy for Cop in the game $\eCR^5_9(G^\circ)$.
\end{exa}

We observe that the monotone moves of Cop correspond to the exact edges in the strategy tree.

\begin{lem}
	\label{obs:monotone-exact}
	For edge $st \in E(T)$, where $s$ is the parent of $t$ it holds that the move $\sigma(\beta(s),\gamma(s,t))$ is monotone if and only if $st$ is exact.
\end{lem}

\begin{proof}
	Let $e\in \gamma(s,t)$.
	By definition it holds that $\gamma(s,t)\subseteq \escapeE(\beta(s) \cap \sigma(\beta(s), \gamma(s,t)), e)$.
	Now by the definition of monotone it holds that $\escapeE(\beta(s) \cap \sigma(\beta(s), \gamma(s,t)), e) \subseteq \gamma(s,t)$ and thus $\gamma(s,t)= \escapeE(\beta(s) \cap \sigma(\beta(s), \gamma(s,t)), e)$.
	As furthermore by construction it holds that $\gamma(t,s)=\complementOfB{\escapeE(\beta(s) \cap \sigma(\beta(s), \gamma(s,t)), e)}$ it follows that $st$ is exact.
\end{proof}

The following lemma about the self-loops of the graph $G^\circ$ is key to prove that the construction in the next section does not enlarge the depth of the \preTreeDec.
That is, it works as intended.

\begin{lem}
	\label{obs:self-loops-gamma}
	\label{obs:self-loops}
	When considering the game on $G^\circ$ all self-loops $vv$ incident to $\beta(s)$, for some $s\in V(T)\setminus L(T)$, are either contained in $\gamma(s,p_s)$ or there is a child $c$ of $s$ such that $\gamma(s,c)=\{vv\}$.
	Furthermore $s$ has a child $c$ with $\gamma(s,c)=\{vv\}$, for some non-isolated vertex $v$ if and only if $s$ is a branching node and $v\in\beta(s)\setminus\beta(p_s)$.
\end{lem}

\begin{proof}
	Let $t\preceq s$ minimal such that $v$ is contained in all bags along the path from $t$ to $s$.
	Note that $p_t$ exists as $\beta(r)=\emptyset$.
	By minimality of $t$ we know that $v\notin \beta(p_t)$, thus $\{v\}=\beta(t)\setminus\beta(p_t)$.
	Therefore we know that $t$ is an inner node corresponding to a move where the cop player picks a new vertex and $v$ is this vertex.
	We know that $\escapeE(\beta(t),vv)=\{vv\}$.
	If $vv\in\gamma(p_t,t)$, then by construction there is a child $c'$ of $t$, such that $\gamma(t,c')=\{vv\}$, otherwise again by construction $vv\in \gamma(t,p_t)$.
	Note that $t$ is branching if and only if $vv\in\gamma(p_t,t)$.
	If $t=s$ we are done.
	Otherwise let $t_1=t,t_2,\ldots,t_k=s$ be the unique path from $t$ to $s$.
	We know that $vv\notin\gamma(t_1,t_2)$.
	For all $i\in [k-1]$, as $v\in \beta(t_i)$ we get that $\escapeE(\beta(t_i),vv)=\{vv\}$ and by induction over $i$ that $vv\notin \gamma(t_i,t_{i+1})$ and thus by construction $vv\in\gamma(t_{i+1},t_i)$.
	For $i=k-1$ we thus get $vv\in\gamma(s,p_s)$.
	
	Now assume $s$ has a child $c$ with $\gamma(s,c)=\{vv\}$.
	Then by construction $vv\in \gamma(p_s,s)$.
	As $v$ is not isolated we have that $\escapeE(\beta(s),vv)=\{vv\}$ only if $v\in\beta(s)$.
	But as $\gamma(p_s,s)$ is connected we have that $v\notin \beta(p_s)$ and thus $v$ is the vertex Cop picked.
	And since $vv\in \gamma(p_s,s)$, $s$ is branching.
\end{proof}

\subsection{Making a strategy tree exact}\label{sec:make-exact}

Our goal is to prove the following lemma.

\begin{lem}
	\label{lem:CRtoTD}
	Let $G$ be a graph and let $k,q\geq 1$. 
	If Cop wins $\CRkq(G)$, then there is a tree-decomposition of $G$ with width $\leq k$ and depth $\leq q$.
\end{lem}

To prove this, we construct an exact \preTreeDec\ of $G^\circ$ from the strategy tree that is constructed from a winning strategy of Cop in the game $\eCRkq(G^\circ)$, starting at the root $r$ and traversing the tree nodes in a breadth-first-search.
We then use \autoref{lem:tw-ptw} to get the desired tree-decomposition.
When we \consider\ a node $t \in V(T)$, 
we change the \preTreeDec\ so that all edges incident with $t$ are exact afterwards.
Note that, by the choice of the traversal, we only need to consider edges to the children of $t$.

\paragraph*{The construction.}
Let $(T,r,\beta,\gamma)$ be the \preTreeDec\ of $G^\circ$ from a winning strategy of Cop in the game $\eCRkq(G^\circ)$.
Let $s_1,\ldots,s_{n_T}$ be an order of the nodes of $T$ as appearing in a breadth-first search where $s_1=r$. Let $\beta_0\coloneqq\beta$ and $\gamma_0\coloneqq \gamma$.
We inductively construct a sequence $(T,r,\beta_0,\gamma_0), \ldots, (T,r,\beta_{n_T},\gamma_{n_T})$ of \preTreeDec s, such that $(T,r,\beta_{n_T},\gamma_{n_T})$ is exact.
We say $s_i$ is \emph{\considered\ in step $i$}.
Let \[T_i\coloneqq T[\{s_1,\ldots,s_i\}\cup N_T(\{s_1,\ldots,s_i\})].\]
See \autoref{fig:teeEye} for an illustration of $T_i$. (It will become clear that this is the subtree of all nodes where the \preTreeDec\ is modified in or before step $i$. We also point out that edges from $T_i$ to $T\setminus T_i$ may become non-exact during our modification process.)

If $s_i$ is a leaf, there are no children and thus no incident edges that are not exact.
We set $\beta_i\coloneqq\beta_{i-1}$ and $\gamma_i\coloneqq \gamma_{i-1}$.
Otherwise let $t^i_1,\ldots,t^i_{a_i}\in N_T(s_i)$ be an enumeration of all children of $s_i$.

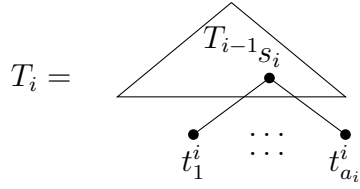
\begin{figure}
	\centering
	\begin{tikzpicture}
		\node at (-1,0) {$T_i =$};
		
		\coordinate (A) at (1.5,1);
		\coordinate (B) at (0,-.25);
		\coordinate (C) at (3,-.25);
		
		\draw (A) -- (B) -- (C) -- (A);
		
		\coordinate (si) at (2,0);
		\node (Ti1) at (1.5,0.5) {$T_{i-1}$};
		\coordinate (t1) at (1,-.75);
		\coordinate (ti) at (3,-.75);
		\coordinate (dots1) at (2,-.75);
		\coordinate (dots2) at (2,-1);
		
		\fill [black] (si) circle (2pt) node[above] {$s_i$};
		\fill [black] (t1) circle (2pt) node[below] {$t_{1}^{i}$};
		\fill [black] (ti) circle (2pt) node[below] {$t_{a_i}^{i}$};
		\node at (dots1) {$\dots$};
		\node at (dots2) {$\dots$};
		
		\draw (si)--(t1);
		\draw (si)--(ti);
	\end{tikzpicture}
	\caption{The subtree $T_i$ appearing in the construction.}  
	\label{fig:teeEye}
\end{figure}

\begin{itemize}
	\item We pick pairwise disjoint $F^i_1,\ldots,F^i_{a_i}\subseteq E(G^\circ)$, with \[F^i_j\subseteq \complementOf{\gamma_{i-1}(t^i_j,s_i)}\cap\complementOf{ \gamma_{i-1}(s_i,t^i_j)},\] such that the partition $\pi^*$ that results from taking the $F^i_j$-extensions in $\gamma_{i-1}(s_i,t^i_j)$, that is adding $F^i_j$ to $\gamma_{i-1}(s_i,t^i_j)$ and removing it from all other sets in order\footnote{The order does not have an impact on the resulting pratition, as all $F^i_j$ are distinct.} $j=1,\ldots,a_i$, has the minimum size boundary. Thus the resulting partition is
	\begin{equation*}
		\begin{split}
			\left(\left(\gamma_{i-1}(s_i,t^i_1)\setminus \bigcup_{j\in[a_i]} F^i_j\right)\cup F^i_1,\ldots,\left(\gamma_{i-1}(s_i,t^i_{a_i})\setminus \bigcup_{j\in[a_i]} F^i_j\right)\cup F^i_{a_i},\right.\\
			\left.\phantom{\bigcup_{j\in[a_i]}}\gamma_{i-1}(s_i,p_{s_i})\setminus\bigcup_{j\in[a_i]} F^i_j \right).
		\end{split}
	\end{equation*}
	If there are multiple optimal choices for $F^i_1,\ldots,F^i_{a_i}$ we select the one that minimizes the size of $\bigcup_{j\in[a_i]} F^i_j$, if there are still several options we break ties arbitrarily.
	\item Let $F^i\coloneqq\bigcup_{j\in[a_i]} F^i_j$ and $F^{*i}_j\coloneqq \left(\left(\complementOf{\gamma_{i-1}(t^i_j,s_i)}\cap \complementOf{\gamma_{i-1}(s_i,t^i_j)}\right)\cup F^i\right) \setminus F^i_j$.
\end{itemize}

For every $p\in V(T_i)$ with child $c$ we set
\begin{equation*}
	\gamma_i(p,c)\coloneqq\begin{cases}
		\left(\gamma_{i-1}(s_i,t^i_j)\setminus F^i\right)\cup F^i_j & \text{ if } (p,c)=(s_i,t^i_j), \text{ for some } j\in[a_i],\\
		\gamma_{i-1}(p,c)\setminus F^{*i}_j & \text{ if } t^i_j= p, \text{ for some } j\in[a_i],\\
		\gamma_{i-1}(p,c)\cup F^i & \text{ if } c\preceq s_i,\\
		\gamma_{i-1}(p,c)\setminus F^i & \text{ otherwise,}
	\end{cases}
\end{equation*}
and
\begin{equation*}
	\gamma_i(c,p)\coloneqq\begin{cases}
		\gamma_{i-1}(c,p) \cup F^{*i}_j & \text{ if } (p,c)=(s_i,t^i_j), \text{ for some } j\in[a_i],\\
		\gamma_{i-1}(c,p) & \text{ if } t^i_j= p, \text{ for some } j\in[a_i],\\
		\gamma_{i-1}(c,p) \setminus F^i & \text{ if } c\preceq s_i,\\
		\gamma_{i-1}(c,p) \cup F^i & \text{ otherwise,}
	\end{cases}
\end{equation*}
and all other $uv\in\overrightarrow{E(T)}$ we set $\gamma_i(u,v)\coloneqq\gamma_{i-1}(u,v)$.
Furthermore we set
\begin{equation*}
	\beta_i(t)\coloneqq\begin{cases}
		\boundary(\pi^i_t) & \text{ if } t\in V(T_i),\\
		\beta_{i-1}(t) & \text{ otherwise.}
	\end{cases}
\end{equation*}

Intuitively in the construction above we push the change at $s_i$ through $T_{i-1}$, that is for all edges in $T_{i-1}$ we add $F^i$ to the directed edge that points away from $s_i$ and remove $F^i$ from the edges in the other direction.
We obtain the following lemma.

\begin{lem}
	\label{obs:only-remove}
	Let $i,j\in[n_T]$ such that $s_i$ is the parent of $s_j$.
	Then $\gamma_{\alpha}(s_i,s_j) \subseteq \gamma(s_i,s_j)$, for all $\alpha<i$.
\end{lem}

\begin{proof}
	We prove this statement by induction over $\alpha$.
	For $\alpha = 0$, the statement obviously holds.
	If $s_i\notin T_{\alpha}$ the statement again obviously holds.
	It suffices to show that, for every $\alpha\in[i-1]$ such that $s_i\in T_{\alpha}$, it holds that $\gamma_{\alpha}(s_i,s_j) \subseteq \gamma_{\alpha-1}(s_i,s_j)$.
	Thus let us consider how $\gamma_{\alpha}(s_i,s_j)$ is computed from $\gamma_{\alpha - 1}(s_i,s_j)$.
	As $\alpha < i$ it holds that $s_j\notin T_{\alpha}$ and $s_i\not\preceq s_{\alpha}$.
	In both of the remaining two cases of the definition we get that $\gamma_{\alpha}(s_i,s_j) \subseteq \gamma_{\alpha-1}(s_i,s_j)$ as desired.
\end{proof}

\begin{figure}
	\hspace*{10pt}
	\begin{subfigure}{0.3\textwidth}
};
			\draw (j) -- (k4);
			
		\end{tikzpicture}
		\phantomcaption{~}
		\label{fig:monotony-construction-after}
	\end{subfigure}
	\caption{Example of the construction, see \autoref{ex:monotony-construction}.}
	\label{fig:monotony-construction}
\end{figure}

\begin{exa}
	\label{ex:monotony-construction}
	We consider the graph $G^\circ$ from \autoref{ex:pre-tree-dec} together with its \preTreeDec\ as depicted in \autoref{fig:pre-tree-dec}.
	
	In \autoref{fig:monotony-construction} we depict some steps of our construction to transform this \preTreeDec\ into an exact \preTreeDec.
	In green we draw nodes and edges that are part of the tree $T_{i-1}$.
	We omit some of the cones within this subtree, as by \autoref{obs:exact-subtree} they are induced by the cones to the leaves and at the edges leaving $T_{i-1}$.
	In orange and dashed we draw edges that are not exact in the intermediate \preTreeDec.
	The dark-blue edges within the edge labels are the edges of the cone that points away from the root, the light-blue edges are those of the cone pointing towards the root.
	In red we depict those edges of the original graph that are covered by neither of the two cones.
	
	Let $s_1,\ldots,s_{68}$ be the nodes of the decomposition tree in breath-first order, that is top to bottom, left to right.
	The first eight steps are trivial as all edges are already exact and only the bag at node $s_9$ changes to the boundary of the partition.
	
	In \autoref{fig:monotony-construction-step9} we see the node $s_9$ together with its parent $p_{s_9}$ and its single child $t^9_1$ before step $i=9$ of the construction.
	We see that we have to choose $F^9_1$ from the set \[\{(1,2)(1,2), (1,2)(2,2), (1,2)3, (2,2)(2,2), (2,2)3, 33, 3(1,4), 3(2,4)\}.\]
	The choice $\{(1,2)(1,2), (1,2)(2,2), (1,2)3, (2,2)(2,2), (2,2)3\}$ minimizes the boundary at $s_9$, which will then be the single vertex $3$.
	The resulting \preTreeDec\ restricted to the vertices in $T_9$ together with its neighbourhood \autoref{fig:monotony-construction-after-step9}.
	We see that one edge leaving $t^9_1$ is no longer exact after this step.
	
	Step $i=10$ is the next interesting step.
	In \autoref{fig:monotony-construction-step10} we see the node $s_{10}$ together with its parent $p_{s_{10}}$ and its four children $t^{10}_1,\ldots,t^{10}_4$ (enumerated from left to right) before step $i=10$ of the construction.
	We see that we have to choose $F^{10}_1$ and $F^{10}_2$ from the following two sets 
	\[\left\{\begin{matrix}
		3(1,4), 3(2,4), (1,4)(1,4), (1,4)(2,4), (1,4)(1,5),\\
		(2,4)(2,4), (2,4)(2,5), (1,5)(1,5), (1,5)(2,5), (2,5)(2,5)
	\end{matrix}\right\},\]
	\[\{(1,2)3, (2,2)3, 33, 3(1,4), 3(2,4), (1,4)(1,4), (1,4)(2,4), (2,4)(2,4)\},\]
	whilst $F^{10}_3=F^{10}_4=\emptyset$, as the corresponding edges are exact.
	The optimal choice here is $F^{10}_1=\emptyset$ and $F^{10}_2=\{33, 3(1,4), 3(2,4), (1,4)(1,4), (1,4)(2,4), (2,4)(2,4)\}$ again resulting in a boundary only consisting of $3$.
	We note that this choice is not unique, we could also choose $F^{10}_1=\{33\}$.
	We note that after this step all cones pointing away from the root at $s_{10}$, except the cone $\gamma_{10}(s_{10},t^{10}_2)$.
	The resulting \preTreeDec\ restricted to the vertices in $T_{10}$ together with its neighbourhood \autoref{fig:monotony-construction-after-step10}.
	
	In every step $i>10$ the optimal choice yields $F^i=\emptyset$.
	The exact \preTreeDec\ and thus also the tree-decomposition that results from the construction can be seen in \autoref{fig:monotony-construction-after}.
	We omit the subtree below $s_{12}$, which only contains empty bags.
	The width after the construction is $3$ and the depth is $5$.
	The reader may note that some edges such as the self-loop $33$ are at leaves that have a distance very far from the first bag they are covered by.
	Our construction never adds new leaves and never adds elements to the cones pointing towards leaves, thus in the resulting exact \preTreeDec\ every edge of the original graph will be covered by some leaf that already covered that edge in the original \preTreeDec.
\end{exa}

\paragraph*{The proof.}

We prove \autoref{lem:CRtoTD} in three steps.
First we prove that the construction indeed yields an exact \preTreeDec. 
Next we show that the width can be bounded as desired and lastly we prove that the construction yields the desired depth.

\begin{lem}
	\label{lem:exact}
	For all $i\in[n_T]$, $(T,r,\beta_i,\gamma_i)$ is a \preTreeDec.
	Furthermore all edges in $E(T_i)$ are exact.
\end{lem}

\begin{proof}
	\ref{ax:preTreeRoot} holds as all edges leaving the root are already exact in $\gamma$, thus we change nothing in step 1 where the root is considered and every $F^i$, with $i>1$, only contains edges from a single component of $G$ by construction.
	
	We observe that the changes from $\gamma_{i-1}$ to $\gamma_{i}$ at some node $t\in V(T_i)\setminus \{s_i\}$ correspond to an $F^i$- or $F^{*i}_j$-extension of $\pi^{i-1}_t$ at the set that corresponds to the edge, that points towards $s_i$.
	Furthermore $\pi^i_s$ is a partition of the edges by construction and for all $t\in V(T_i)$ we set $\beta_i(t)=\boundary(\pi^i_t)$.
	As $\gamma_i$ and $\beta_i$ are equal to $\gamma$ and $\beta$ at all vertices that are not part of $V(T_i)$, this shows by induction that \ref{ax:preTreePart} still holds.
	
	Next we observe that at every edge that is not incident to some $t^i_j$ we add to one direction exactly what we remove from the other direction.
	Furthermore by construction the edges $s_i t^i_j$ are exact after the construction.
	Lastly, for all children $c$ of $t^i_j$, we only remove edges from $\gamma_i(t^i_j,c)$.
	Thus again by induction we get that \ref{ax:preTreeEdge} holds and that all edges of $T_i$ are exact.
	
	It remains to show that \ref{ax:preTreeLeaf} holds.
	Let $i\in[n_t]$ and $j\in[a_i]$ such that $t^i_j\in L(T)$.
	From \autoref{obs:only-remove} we know that $\gamma_{i-1}(s_i,t^i_j)\subseteq \gamma(s_i,t^i_j)$ and thus $|\gamma_{i-1}(s_i,t^i_j)|\leq 1$.
	Furthermore we know that $\gamma_{i-1}(t^i_j,s_i)=\gamma(t^i_j,s_i)= E(G^\circ)\setminus \gamma(s_i,t^i_j)$.
	Therefore we get that $F^i_j\subseteq \gamma(s_i,t^i_j)$ and thus $\gamma_{i}(s_i,t^i_j)\subseteq \gamma(s_i,t^i_j)$.
	By construction, in step $i'\in[n_T]$, such that $s_{i'}=t^i_j$, we do nothing.
	And in all other steps $\alpha > i$, we have that $t^i_j\not\preceq s_\alpha$ and thus we only remove edges from $\gamma_{\alpha}(s_i,t^i_j)$.
	This shows that for all $\alpha\in[n_T]$ we have $|\gamma_{\alpha}(s_i,t^i_j)|\leq 1$.
\end{proof}

Hence, for $i=n_T$, we get that $(T,r,\beta_{n_T},\gamma_{n_T})$ is an exact \preTreeDec. 
Note that it is possible that $\gamma_{n_T}(s,t)$ is empty for an edge $st\in \overrightarrow{E(T)}$.
By Lemma~\ref{lem:tw-ptw} we obtain a tree-decomposition, from this \preTreeDec.
We show below that the width and depth are as stated in the theorem.

Our construction does not change the width of the decomposition.
To prove this we observe that in step $i$ the bound in $s_i$ is minimal.
We then push the change through the subtree $T_i$ and find that if a change would increase the width, we could push this change back to the node $s_i$ and find an even smaller bound there, which contradicts the minimality of our choice.

\begin{lem}
	\label{lem:width}
	$\wid(T,r,\beta_i,\gamma_i)\leq \wid(T,r,\beta,\gamma)$, for all $i\in[n_T]$.
\end{lem}

\begin{proof}
	We prove the statement for all $0\leq i\leq n_T$ by induction.
	The statement clearly holds for $i=0$, as $(T,r,\beta_0,\gamma_0)=(T,r,\beta,\gamma)$.
	For all $i\in[n_T]$, it remains to show that $\wid(T,r,\beta_i,\gamma_i)\leq \wid(T,r,\beta_{i-1},\gamma_{i-1})$.
	Obviously $|\beta_i(t)|=|\beta_{i-1}(t)|$, for all $t\notin T_i$.
	Furthermore by construction $|\beta_i(s_i)|\leq |\beta_{i-1}(s_i)|$.
	Let $j\in[a_i]$, let $X\coloneqq \gamma_i(s_i,t^i_j)$ and let $Y\coloneqq\gamma_{i-1}(t^i_j,s_i)$.
	We observe that \[\pi^{i}_{t^i_j}=\pi^{i-1}_{t^i_j,Y \ext \complementOf{X}}.\]
	Thus it holds that \[|\beta_i(t^i_j)| = |\wid(\pi^{i-1}_{t^i_j, Y \ext \complementOf{X}})| \leq |\wid(\pi^{i-1}_{t^i_j})| \leq |\beta_{i-1}(t^i_j)|\] as otherwise by submodularity, as established in \autoref{lem:submodular}, for the partitions $\pi^{i-1}_{t^i_j}$ and $\pi^i_{s_i}$, we get that \[|\wid(\pi^i_{s_i, X \ext \complementOf{Y}})| < |\wid(\pi^i_{s_i})|,\]
	which contradicts the minimality of the bound for $F^i_1,\ldots,F^i_{a_i}$.
	
	Lastly assume there is a node $t$ in $V(T_i)\setminus\{s_i,t^i_1,\ldots,t^i_{a_i}\}$ such that $|\beta_i(t)|>|\beta_{i-1}(t)|$.
	We assume $t$ is of minimal distance to $s_i$ with this property.
	Let $x_0=t,x_1,\ldots,x_{b}=s_i$ be the path from $t$ to $s_i$.
	By minimality of the distance we know that $|\beta_i(x_1)|\leq |\beta_{i-1}(x_1)|$.
	Additionally we know that all edges on the path from $s_i$ to $x_1$ are exact in $\gamma_i$, as well as the edge $x_1t$ in $\gamma_{i-1}$.
	Now let $Y\coloneqq \gamma_{i-1}(t,x_1)$ and, for all $0\leq \alpha < b$, let $X_\alpha \coloneqq \gamma_i(x_{\alpha + 1},x_{\alpha})$ and $Z_\alpha \coloneqq \gamma_i(x_{\alpha},x_{\alpha + 1})$.
	The transition from $i-1$ to $i$ at $t$ corresponds to $\pi^{i-1}_{t,Y \ext F}=\pi^{i-1}_{t,Y \ext \complementOf{X_0}}$.
	Thus if $\wid(\pi^{i-1}_{t,Y \ext F})=|\beta_i(t)|>|\beta_{i-1}(t)|=\wid(\pi^{i-1}_{t})$ it follows from submodularity that $\wid(\pi^i_{x_1})>\wid(\pi^i_{x_1,X_0 \ext \complementOf{Y}})$.
	As the edge $x_1t$ was exact at step $i-1$, we know that  \[F'\coloneqq\complementOf{Y}\setminus X_0= F^i\setminus Y \subseteq F^i.\]
	We now push this change back to $s_i$ along the path $x_1,\ldots,x_b$ and again we construct a contradiction to the minimality of the bound of $F^i_1,\ldots,F^i_{a_i}$.
	For this, let us assume we have pushed the change to $x_\alpha$, that is we changed $\pi^i_{x_\alpha}$ to $\pi^*_{x_\alpha} = \pi^i_{x_\alpha,X_{\alpha - 1}\ext F'}$ and we know that $\wid(\pi^*_{x_\alpha})<\wid(\pi^{i}_{x_\alpha})$.
	As the edge $x_\alpha x_{\alpha+1}$ is exact in $\gamma_i$, we get that $\pi^*_{x_\alpha,(Z_\alpha \setminus F') \ext \complementOf{X_\alpha}} = \pi^i_{x_\alpha}$.
	Let \[\pi^*_{x_{\alpha+1}}\coloneqq \pi^i_{x_{\alpha+1},X_\alpha \ext \complementOfB{(Z_\alpha \setminus F')}} = \pi^i_{x_{\alpha+1},X_\alpha \ext F'},\] then by submodularity $\wid(\pi^*_{x_{\alpha+1}}) < \wid(\pi^i_{x_{\alpha+1}})$.
	When we have pushed the change to $\alpha=b$, we find the desired contradiction.
\end{proof}

To prove that our construction does not increase the depth we show that in every step $i$ the depth up to the nodes in $T_i$ is bounded by the depth up to these nodes in the original tree.
We prove this by induction on the number of steps.
In step $i$ every change in any bag at some node in $V(T_{i-1})$ is closely related to the change at the \considered\ node $s_i$.
Additionally we find that the vertices in a bag at some child of $s_i$ that is not present in the bag at $s_i$ is exactly the vertex the cop player newly placed in the corresponding move of the game.

\begin{lem}
	\label{lem:depth}
	For all $i\in[n_T]$ and all $t\in V(T_i)$, it holds that
	\begin{equation*}
		\sum_{r\prec s \preceq t} |\beta_i(s) \setminus \beta_i(p_s)|\leq \sum_{r\prec s \preceq t} |\beta(s) \setminus \beta(p_s)|.
	\end{equation*}
\end{lem}

\begin{proof}
	Let $\ell\in[n_T]$.
	As by construction $\beta_\ell(t)=\boundary(\pi^\ell_t)$, for all $t\in V(T_\ell)$, we get from \autoref{obs:exact-subtree-depth} and \autoref{lem:exact} that $|\bigcup_{s \preceq t} \beta_\ell(s)| = \sum_{r\prec s \preceq t} |\beta_\ell(s) \setminus \beta_\ell(p_s)|$.
	Thus it suffices to show that $|\bigcup_{s \preceq t} \beta_i(s)| \leq \sum_{r\prec s \preceq t} |\beta(s) \setminus \beta(p_s)|.$
	
	We prove the statement by induction on the steps $i$.
	Let us recall that by construction $(T,r,\beta_0,\gamma_0)=(T,r,\beta,\gamma)$, thus the statement holds for $i=0$.
	Now assume the statement holds for $i-1$, thus for all $t\in V(T_{i-1})$ it hold that  $|\bigcup_{s \preceq t} \beta_{i-1}(s)| \leq \sum_{r\prec s \preceq t} |\beta(s) \setminus \beta(p_s)|.$
	
	We recall that $V(T_i)=V(T_{i-1})\cup \{t^i_1,\ldots,t^i_{a_i}\}$ and that $s_i\in L(T_{i-1})$.
	For the nodes $t\in V(T_{i-1})$ we can directly build upon the induction hypothesis.
	But the nodes $t^i_j$, with $j\in [a_i]$, are added into the subtree.
	Here we need to compare directly to the original bags, as we can no longer use that in step $i-1$ the depth at these nodes is bounded by the depth in the original strategy tree.
	We can prove for these nodes that every vertex newly placed at one of these nodes in step $i$ is also newly placed in the original strategy.
	Then we can show that the difference between depth at these nodes and their parent in step $i$ is bounded by the difference in the original strategy tree.
	\begin{claim}
		\label{cl:depth-leafs-new}
		Every $j\in[a_i]$ satisfies $\beta_i(t^i_j)\setminus \beta_i(s_i)\subseteq \beta(t^i_j)\setminus \beta(s_i).$
	\end{claim}
	\begin{claimproof}
		Let $v \in \beta_i(t^i_j) \setminus \beta_i(s_i)$.
		As $v \notin \beta_i(s_i)$ we get that $v\notin\boundary(\gamma_i(t^i_j,s_i))$ and thus it holds that $E_{G^\circ}(v)\cap \gamma_i(t^i_j,s_i) = \emptyset$.
		By construction we have that $\gamma_i(t^i_j,s_i)\supseteq \gamma_{i-1}(t^i_j,s_i) = \gamma(t^i_j,s_i)$, and thus $vv \notin \gamma(t^i_j,s_i)$.
		As $v \in \beta_i(t^i_j)=\boundary(\pi^i_{t^i_j})$, there are two distinct children $c_1,c_2$ of $t^i_j$ such that $v\in\boundary(\gamma_i(t^i_j,c_\ell))$ and thus $E_{G^\circ}(v)\cap \gamma_i(t^i_j,c_\ell) \neq \emptyset$, for $\ell = 1,2$.
		By construction we have $\gamma_i(t^i_j,c_\ell)\subseteq \gamma_{i-1}(t^i_j,c_\ell) = \gamma(t^i_j,c_\ell)$, for $\ell = 1,2$.
		And thus $v\in \boundary(\pi_{t^i_j})\subseteq \beta(t^i_j)$.
		By \autoref{obs:self-loops-gamma}, there is a child $c$ of $t^i_j$ such that $\gamma(t^i_j,c)=\{vv\}$ and, by \autoref{obs:self-loops}, $v\in \beta(t^i_j)\setminus \beta(s_i)$.
	\end{claimproof}
	
	The following claim tracks vertices that are added to any bag in $V(T_{i-1})$ at step $i$.
	
	\begin{claim}
		\label{cl:new-vertex-new}
		Let $i\in [n_T]$ and let $t\in V(T_{i-1})$.
		If $v\in \beta_i(t)\setminus\beta_{i-1}(t)$, then $v\in\beta_i(t^*)$, for all $t^*$ on the path from $t$ to $s_i$.
	\end{claim}
	
	\begin{claimproof}
		Let $t^*\neq t$ and let $t'$ be the next node on the path from $t$ to $s_i$.
		It holds that $\gamma_i(t,t')=\gamma_{i-1}(t,t')\cup F^i$.
		As $\gamma_i(t,t')$ is the only set incident to $t$ where edges are added in step $i$, we get that $v\in\boundary(\gamma_i(t,t'))$.
		And from $v\notin\boundary(\gamma_{i-1}(t,t'))$ we get that $v\in\boundary(F^i)$.
		Now suppose that $v\notin \beta_i(t^*)$, and thus also $v\notin \boundary(\gamma_i(t^*,p))$, where $p$ is the next node on the path from $t^*$ to $t$.
		As $v$ is incident to edges in $F^i$ we get that $E_{G^\circ}(v) \cap \gamma_i(t^*,p) = \emptyset$.
		We know from \autoref{lem:exact} that all edges in $T_i$ are exact and thus that $\gamma_i(t',t)\subseteq \gamma_i(t^*,p)$ by \autoref{obs:exact-path}.
		This contradicts $v\in\boundary(\gamma_i(t,t'))=\boundary(\gamma_i(t',t))$ and thus $v\in\beta_i(t^*)$.
	\end{claimproof}
	
	The next claim is used to show that a vertex that disappears from a bag in $V(T_{i-1})$ at step $i$ also disappears from the union of bags that determine the depth at that bag, especially if a vertex disappears from the bag at $s_i$, then it disappears from every bag in $V(T_i)$.  
	
	\begin{claim}
		\label{cl:remove-vertex-new}
		Let $i\in [n_T]$ and let $t\in V(T_{i-1})$. If $v\in \beta_{i-1}(t)\setminus\beta_{i}(t)$, then $v\notin\beta_i(t^*)$, for all $t^*\in V(T_{i-1})$ 
		such that $t$ is contained in the path from $t^*$ to $s_i$.
	\end{claim}
	
	\begin{claimproof}
		We have $E_{G^\circ}(v)\cap F^i\neq \emptyset$.
		
		Let $t=s_i$.
		As $v\notin \beta_i(s_i)$ we get that $v\notin \boundary(\gamma_i(s_i,p_{s_i}))$ and thus for the edges incident to $v$ it holds that
		$E_{G^\circ}(v)\cap \gamma_i(s_i,p_{s_i}) = E_{G^\circ}(v)\cap \gamma_{i-1}(s_i,p_{s_i}) \cap \complementOfB{F^i} = \emptyset$.
		Now let $t^*\in V(T_{i-1})$ and $t'$ be the next node on the path from $t^*$ to $s_i$.
		Then by \autoref{lem:exact} we get that $\gamma_i(t^*,t')\supseteq\gamma_i(p_{s_i},s_i)\supseteq E_{G^\circ}(v)$ 
		and thus $v\notin\beta_i(t^*)$.
		
		Otherwise let $t\neq s_i$.
		Let $t'$ be the next node on the unique path from $t$ to $s_i$.
		As $v\notin\boundary(\gamma_i(t,t'))$ it follows that $E_{G^\circ}(v)\subseteq \gamma_i(t,t')=\gamma_{i-1}(t,t')\cup F^i$, that $v\in\boundary(\gamma_{i-1}(t,t'))$, and that $E_{G^\circ}(v) \cap \gamma_{i-1}(t',t)\subseteq E_{G^\circ}(v)\cap F^i$.
		Assume there is some $t^*\in V(T_{i-1})$ such that $v\in\beta_i(t^*)$.
		We observe that due to \autoref{lem:exact} and because all edges incident to $v$ are contained in $\gamma_i(t,t')$, we get that $t$ is not contained in the path from $t^*$ to $s_i$.
	\end{claimproof}
	
	We are now ready to prove the lemma.
	Towards this, let $i\geq 1$ and assume the statement holds for $i-1$.
	We consider all vertices that appear at a bag at any node in $T_i$ due to the changes in step $i$.
	Observe that if $\beta_i(s_i)=\beta_{i-1}(s_i)$, then there are no changes to the bags at other nodes than the $t^i_j$ by minimality of $|F^i|$, and if $\beta_i(s_i)\neq\beta_{i-1}(s_i)$,  we have $|\beta_i(s_i)|<|\beta_{i-1}(s_i)|$ again by the minimality of the choice.
	
	Let $t\in V(T_{i-1})$.
	Let $U \coloneqq \left(\bigcup_{s \preceq t} \beta_{i}(s)\right) \setminus \left(\bigcup_{s \preceq t} \beta_{i-1}(s)\right)$.
	Let $t^*$ be the greatest common ancestor of $t$ and $s_i$.
	As $t^*$ is on every path from some predecessor of $t$ to $s_i$, from \autoref{cl:new-vertex-new} we know that $u\in\beta_i(t^*)\setminus \beta_{i-1}(t^*)$, for all $u\in U$.
	Let $W = \beta_{i-1}(t^*)\setminus \beta_{i}(t^*)$.
	As by \autoref{lem:width} it holds that $|\beta_i(t^*)|\leq |\beta_{i-1}(t^*)|$, we know that $|U| \leq |W|$.
	By \autoref{cl:remove-vertex-new} we get that $W\subseteq \left(\bigcup_{s \preceq t} \beta_{i-1}(s)\right) \setminus \left(\bigcup_{s \preceq t} \beta_{i}(s)\right)$.
	By this \emph{vertex exchange} we conclude that $\left|\bigcup_{s \preceq t} \beta_{i}(s)\right| \leq \left|\bigcup_{s \preceq t} \beta_{i-1}(s)\right|$.
	
	Otherwise $t=t^i_j$ for some $j\in [a_i]$.
	We get $\bigcup_{s\preceq{t}} \beta_{i}(s) = \bigcup_{s\preceq{s_i}} \beta_{i}(s) \cup \beta_i(t) \setminus \beta_i(s_i)$, by construction.
	We have shown that $|\bigcup_{s\preceq{s_i}} \beta_{i}(s)| \leq \sum_{r\prec s\preceq{s_i}} |\beta(s) \setminus \beta(p_s)|$ 
	and by \autoref{cl:depth-leafs-new} we have $\beta_i(t) \setminus \beta_i(s_i) \subseteq \beta(t)\setminus \beta(s_i)$.
	Thus we get $|\bigcup_{s\preceq{t}} \beta_i(s)| \leq \sum_{r\prec s\preceq{t}} |\beta(s) \setminus \beta(p_s)|$.	
\end{proof}

\begin{proof}[Proof of \autoref{lem:CRtoTD}]
	Let $G=(V(G),E(G))$ be a graph, let $k,q\geq 1$ and assume Cop wins $\CRkq(G)$.
	Then Cop wins $\eCRkq(G^\circ)$ and there is some winning strategy \[\sigma\colon V(G^\circ)^{\leq k}\times E(G^\circ)\rightarrow V(G^\circ)^{\leq k},\]
	where $\sigma(X,uv)=\sigma(X,u'v')$ for all $X\in V(G^\circ)^{\leq k}$ and $u'v'\in\escapeE(X,uv)$.
	From this winning strategy we construct a strategy tree, thus \preTreeDec\ of $G^\circ$.
	Combining Lemmas \ref{lem:exact}, \ref{lem:width} and \ref{lem:depth} we get that there exists an exact \preTreeDec\ of $G^\circ$ of width $\leq k$ and depth $\leq q$, if the cop player wins $\eCRkq(G^\circ)$.
	The theorem then follows from \autoref{lem:tw-ptw}.
\end{proof}

The reader may note that the way we define monotony is known as robber-monotony in some literature.
There also is an alternative definition called cop-monotony, where Cop is not allowed to move to a vertex that he moved away from in an earlier round.
For standard Cop-and-Robber game these two notions are known to be equivalent.
Our proofs show that they are also equivalent in our game, as the strategy that results from a tree-decomposition is both robber-monotone and cop-monotone.

\newcommand{\numberOfRoundsK}{\ensuremath{\frac{(k-1)(\ell - (k-1) + 2)}{4}}}
\newcommand{\sizeLargeComponent}{\ensuremath{\frac{\ell h - h - 2}{2}}}
\newcommand{\sizeNonGood}{\ensuremath{\frac{(h-1)(h+2)}{2}}}
\newcommand{\boundL}{\ensuremath{\lceil \frac{q}{k-1}\rceil (k+4)}}

\subsection{Separating \texorpdfstring{$\Ekq$}{Ekq} from \texorpdfstring{$\TW_{k-1} \cap \TD_q$}{TWk-1 intersection TDq}}

We aim to show that the graph class $\Ekq$ is a proper subclass of $\TW_{k-1} \cap \TD_q$.
One can only hope to separate the classes if $q$ is larger than $k$, since $\EParam{q}{q}=\TD_q=\TW_{q-1}\cap\TD_q$.
Using the characterizations of $\Ekq$ via a Cops-and-Robber game, we show that there indeed are graphs for which one can not simultaneously bound the width to treewidth and the depth to treedepth.

The graph we consider is the $(h\times\ell)$-grid \grid{h}{\ell}, with $h<\ell$. 
Its vertex set is $[h] \times [\ell]$ and it contains an edge between $(i,j)$ and $(i',j')$ iff $|i-i'| + |j-j'| = 1$ for $i,i' \in [h]$ and $j,j' \in [\ell]$.
It is well known that $\tw(\grid{h}{\ell})=h$ and $\td(\grid{h}{\ell})\leq h \lceil\log(\ell+1)\rceil$, cf.\@ \autoref{fig:grid-decompositions}. We give a lower bound to the number of cop rounds that Robber can win in a, possibly non-monotone, game $\CR^h(G)$, which is linear in both $\ell$ and $h$.

\begin{lem}
	\label{lem:lower-bound-rounds}
	For $1<h<\ell-2$ and $q\leq\frac{h(\ell  - h + 2)}{4}$, Robber wins the game $\CR_q^{h+1}(\grid{h}{\ell})$.
\end{lem}

The proof of \autoref{lem:lower-bound-rounds} builds upon \cite{Furer_rounds_2001}.
The winning strategy of the robber is to always stay in the component with the most vertices.
To prove that this strategy benefits Robber, we first make some structural observations on the grid~\grid{h}{\ell} with $h>1$ rows and $\ell$ columns and its separators of size $\leq h+1$.
For $X\subseteq V(\grid{h}{\ell})$, we call a connected component of $\grid{h}{\ell}\setminus X$ \emph{good} if it contains at least one full column of \grid{h}{\ell}.

\begin{lem} \label{obs:good-component}
	For $1 < h < \ell - 1$ and $X\in \binom{V(\grid{h}{\ell})}{\leq h+1}$, there exists a good component in $\grid{h}{\ell} \setminus X$.
\end{lem}

\begin{proof}
	Since there are at most $h+1$ vertices in $X$ and the grid has at least $h+2$ columns, there exists a column which does not contain any vertices from $X$. This column is contained in a good component of $\grid{h}{\ell} \setminus X$.
\end{proof}

The next lemma shows that $\grid{h}{\ell}\setminus X$ can never contain more than two good components if $X\in\binom{V(\grid{h}{\ell})}{\leq h+1}$.
And if it does contain two good components, there can only be one single vertex that is in any other component.

\begin{lem}
	\label{lem:third-comp-single}
	Let $h,\ell>1$ and let $X\in \binom{V(\grid{h}{\ell})}{\leq h+1}$. If there are two good components in $\grid{h}{\ell}\setminus X$, then there is at most one additional component and this component has size 1.
\end{lem}

\begin{proof}
	Let $C_1,C_2$ be the two good components.
	Since $X$ separates $C_1$ and $C_2$, $X$ must contain at least one vertex from every row. Thus, at most one row contains two vertices from $X$.
	Let $C_3$ be the set of vertices that are not contained in $X\cup C_1\cup C_2$.
	Assume $C_3$ contains a vertex from row $i\in[h]$.
	Then row $i$ contains vertices from $C_1$, $C_2$, and $C_3$. 
	Since these are all distinct components of $\grid{h}{\ell}\setminus X$, row $i$ also contains at least two, thus exactly two, vertices of $X$.
	
	Hence, $C_3$ intersects at most one row, say row~$r$.
	As $C_3$ cannot contain a vertex of any other row, all neighbours of $C_3$ that are not in row $r$ need to be in $X$.
	As $h>1$, at least one of $r-1$ or $r+1$ is in $[h]$.
	Thus there is a row in $[h]\setminus \{r\}$ where at least $|C_3|$ vertices are contained in $X$.
	Since all rows besides $r$ have exactly one vertex in $X$, we conclude that $|C_3|\leq 1$.
\end{proof}

We use \autoref{lem:third-comp-single} to show that Robber can find some large component and that Cop can never remove more than two vertices from this component as long as his component is good.
Next we show that a large component is always good.

\begin{lem}
	\label{lem:size-non-good}
	Let $1<h<\ell-2$ and let $X\in \binom{V(\grid{h}{\ell})}{\leq h+1}$.
	A connected component $C$ of $\grid{h}{\ell}\setminus X$ that contains more than $\sizeNonGood$ vertices is good.
\end{lem}

\begin{proof}
	Let $C$ be a connected component of $\grid{h}{\ell}\setminus X$ which is not good.
	We prove that $C$ contains at most $\sizeNonGood$ vertices.
	
	By \autoref{obs:good-component}, there exists a good component $C'$ in $G \setminus X$ containing some column $j \in [\ell]$. Without loss of generality, all columns intersecting $C$ lie left of $j$, i.e.\@ have index strictly smaller than $j$.
	
	\begin{claim} \label{cl:size-non-good}
		For $1\leq i<h$,
		$C$ contains at most $i+1$ columns with at least $h-i$ vertices.
	\end{claim}
	\begin{claimproof}
		Assume $C$ would contain at least $i+2$ columns with at least $h-i$ vertices.
		$X$ contains at least one vertex of each of these columns, as $C$ is not good.
		Let $j'$ be the largest index of a column intersecting $C$ and let $I \subseteq [h]$ be the set of rows that intersect $C$ in column $j'$, that is $(i',j')\in C$, for all $i'\in I$.
		We define vertex disjoint $C$-$C'$-paths by $(i',j'),(i',j'+1),\ldots,(i',j)$ for all $i' \in I$.
		By Menger's Theorem, $X$ contains at least one vertex of each of these paths and all those vertices are in some column that is strictly larger than $j'$.
		All in all, $X$ contains at least $i+2+|I|\geq i+2+h-i=h+2$ vertices, which contradicts the choice of $X$.
	\end{claimproof}
	
	For $1 \leq i \leq h$, write $m_i$ for the number of columns in $C$ which contain at least $i$ vertices. Then $m_i - m_{i+1}$ is the number of columns in $C$ which contain exactly $i$ vertices. 
	By \autoref{cl:size-non-good}, $m_i \leq h +1 - i$ for all $1 \leq i \leq h-1$.
	Hence,
	\[
	|C| = \sum_{i = 1}^{h-1} (m_i - m_{i+1}) i = \sum_{i=1}^{h-1} m_i \leq (h-1)(h+1) - \sum_{i=1}^{h-1} i = \sizeNonGood,
	\]
	as desired.
\end{proof}

We observe that this bound is tight as one can realize a component $C$ that is not good and contains exactly $i+1$ columns with at least $h-i$ vertices, for all $1\leq i<h$, by $X=\{(1,1),(1,2),(2,3),\ldots,(h,h+1)\}$.
We are now ready to prove \autoref{lem:lower-bound-rounds}.

\begin{proof}[Proof of \autoref{lem:lower-bound-rounds}]
	The strategy of Robber is to always move to the largest component.
	\begin{claim}
		After $q'< \frac{\ell h - h^2 + 2h}{4}$ cop introductions the size of the robber component is at least $\sizeLargeComponent-2(q'-h) > \sizeNonGood$.
	\end{claim}
	\begin{claimproof}
		Since $h < \ell-2$, for every Cop position $X\in \binom{V(\grid{h}{\ell})}{\leq h+1}$, there are at least two columns where no cop is positioned.
		In order to catch the robber, the cops have to move to a position such that these two columns lie in distinct good components of $\grid{h}{\ell}\setminus X$.
		Since such an $X$ contains a vertex of every row, Cop needs to introduce at least $h$ cops.
		
		By the pigeonhole principle and \autoref{lem:third-comp-single}, there is some component $C$ of size at least $\sizeLargeComponent$ and it can be reached by Robber.
		Before Robber reaches this component, the largest component of the graph contained even more vertices.
		If the cops move back to a position $X$ with only one good component, the escape space of the robber is again larger than $\sizeLargeComponent$.
		If the cops move outside of the robber component, its size does not change.
		The only way to shrink the size of the escape space is to move the next cop into the robber component.
		Such a move could, in addition to the vertex where the cop is placed onto, remove some connected component from the robber escape space.
		But we know from \autoref{lem:third-comp-single} that this component can only be a singleton as long as the larger remaining component is good.
		Thus from this point onward the size of the robber component shrinks by at most two per round as long as the robber is still in a good component.
		Since 
		\begin{equation*}
			\sizeLargeComponent - 2(q' - h)> \sizeLargeComponent -2\left(\frac{\ell h - h^2 + 2h}{4} - h\right) = \sizeNonGood,
		\end{equation*}
		Robber can choose a good component by \autoref{lem:size-non-good}.
	\end{claimproof}
	
	As $\sizeNonGood>1$, for $h>1$, Robber is not caught after $q\leq \frac{\ell h - h^2 + 2h}{4}$ rounds and he wins the game.
\end{proof}

The bound of \autoref{lem:lower-bound-rounds} turns out to be tight up to an additive term that only depends on $h$, for all $h>3$.

\begin{lem}
	\label{lem:upper-bound-rounds}
	For $3<h<\ell-3$ and $q\geq \frac{\ell h}{4} + h +1$,
	Cop wins the game $\CR_q^{h+1}(\grid{h}{\ell})$.
\end{lem}

\begin{proof}
	In the first $h$ rounds, Cop places the cops on the following vertices in arbitrary order $X\coloneqq\{(1,\lfloor\frac{\ell}{2}\rfloor-\lfloor\frac{h}{2}\rfloor+1),(2,\lfloor\frac{\ell}{2}\rfloor-\lfloor\frac{h}{2}\rfloor+2),\ldots,(h,\lfloor\frac{\ell}{2}\rfloor+\lceil\frac{h}{2}\rceil)\}$.
	
	W.l.o.g.\ we assume that the robber is in position $\escape(X,{(h,1)})$.
	In the next round Cop places the last cop on $(h,\lfloor\frac{\ell}{2}\rfloor+\lceil\frac{h}{2}\rceil - 2)$.
	In round $r$, as long as $(1,1)\notin X_r$, there is some $i\in[h]$ and $1<j\leq \lfloor\frac{\ell}{2}\rfloor-\lfloor\frac{h}{2}\rfloor+1$ such that the cops are in some position $X_{r}=\{(1,j), (2,j+1), \ldots, (i,j+i-1), (i,j+i-3), (i+1,j+i-2), \ldots, (h,j+h-3)\}$.
	The robber is either still in position $\escape({X_{r}},{(h,1)})$ or he moved to the singleton component $\{(i,j+i-2)\}$.
	If the robber moves to the singleton component, he can be caught in the next round as $h+1>4$ and the singleton has at most four neighbours.
	Thus assume that the robber is in position $\escape({X_{r}},{(h,1)})$.
	Then the cop player moves the cop on position $(i,j+i-1)$ to $(i-1,\max(1,j+i-4))$ if $i>1$, or to position $(h,j+h-5)$.
	That is in $h$ rounds the cops move iteratively to a new diagonal that is two columns closer to vertex $(h,1)$.
	Else if $(1,1)\in X_r$, Cop continues this strategy of iteratively moving the (now partial) diagonal two columns closer to the vertex $(h,1)$ until the robber is caught.
	All in all this strategy places $\lceil(\lfloor\frac{\ell}{2}\rfloor-\lfloor\frac{h}{2}\rfloor+i)/2\rceil$ times a cop in row $i\in [h]$, if the robber avoids singleton components as long as possible.
	Therefore the number of rounds done by Cop is
	\begin{align*}
		1 + \sum_{i=1}^{h} \lceil(\lfloor\frac{\ell}{2}\rfloor-\lfloor\frac{h}{2}\rfloor+i)/2\rceil &\leq 1 + \frac12 \sum_{i=1}^{h} (\frac{\ell}{2}-\frac{h-1}{2}+i+1)\\
		&= 1 + \frac{\ell h-h^2+3h}{4} +\frac12\sum_{i=1}^{h} i = \frac{\ell h}{4} + h + 1
	\end{align*}
	
	If the robber starts in position $\escape(X,{(1,\ell)})$ the strategy for the cops is symmetric, but since $\lfloor\frac{\ell}{2}\rfloor-\lfloor\frac{h}{2}\rfloor+1\geq \ell - \lfloor\frac{\ell}{2}\rfloor - \lceil\frac{h}{2}\rceil$, Cop does not need longer to catch the robber.
\end{proof}

For $h=1$, the proof idea of \autoref{lem:lower-bound-rounds} is not applicable as on a path there are separators of size two that separate the path into three components of roughly equal size.
Despite that, one may observe that such a separator does not benefit Cop as from such a position he would always have to combine two of these components into a larger one.
Thus the cop player can only move along the path and shrink the escape space of Robber by one vertex.
This case is covered in the original proof of \cite{Furer_rounds_2001}.

\begin{lemC}[{\cite[Theorems~5~and~7]{Furer_rounds_2001}}]
	\label{lem:rounds-path}
	Let $\ell\geq 1$. Robber wins the game $\CR_q^2(\grid{1}{\ell})$ if and only if $q\leq\lceil \frac{\ell-1}{2}\rceil$.
\end{lemC}

Using Lemmas \ref{lem:lower-bound-rounds} and \ref{lem:rounds-path}, we can construct a graph $G\in\TW_{k-1}\cap\TD_q$ such that Robber wins $\CR_q^k(G)$, if $q$ is sufficiently larger than $k$.

\begin{lem}
	\label{lem:CR_tw-td}\hfill
	\begin{enumerate}
		\item For  $2\leq k-1\leq\frac{q}{3 +\log q}$, there exists a connected graph $G\in\TW_{k-1}\cap\TD_q$ such that Robber wins $\CR_q^k(G)$.
		\item For $q \geq 3$, there exists a connected graph $G'\in \TW_{1}\cap\TD_{q}$ such that Robber wins $\CR_q^2(G')$.
	\end{enumerate}
\end{lem}

\begin{proof}
	We first construct $G'\in \TW_{1}\cap\TD_q$.
	Consider the path $G'\coloneqq\grid{1}{2^q-1}$.
	It is well known that $\tw(\grid{1}{2^q-1})=1$ and $\td(\grid{1}{2^q-1}) =  \lceil \log(2^q-1+1)\rceil = q$, see  \cite[(6.2)]{Nesetril_sparsity_2012}, thus $G' \in \TW_{1}\cap\TD_q$.
	By \autoref{lem:rounds-path}, Robber wins against two cops if the game is played for $q' \leq \lceil \frac{2^q-1-1}{2}\rceil=2^{q-1} -1$ rounds.
	Thus Robber wins $\CR_q^2(G')$ since $q \geq 3$ by assumption.
	
	Next we construct $G\in\TW_{k-1}\cap\TD_q$.
	Since $\frac{q}{k-1} \geq \log q > 1$, we may pick an integer~$\ell$ such that $\frac{q}{k-1}(k+1) \leq \ell \leq \frac{q}{k-1}(k+2)$.  
	Define the graph $G$ as the grid $\grid{k-1}{\ell}$.
	Then $\td(G) \leq (k-1)(\lceil \log (\ell+1)\rceil)$, cf.\@ \autoref{fig:treedepthexample}.
	Using that $\ell \leq \frac{q}{k-1}(k+2)$ and $q \geq k+1$, we get that 
	\begin{align*}
		(k-1)(\lceil \log (\ell+1)\rceil) &\leq (k-1)\left( \log \left(\frac{q}{k-1}(k+2) + 1 \right)+1\right)\\
		&\leq (k-1)\left( \log (q)+\log\left(\frac{k+2}{k-1} + \frac{1}{k+1} \right)+1\right)\\
		&\leq (k-1)\left( \log (q)+2.5 \right) \leq q,
	\end{align*}
	where the penultimate inequality holds since $k \geq 3$ by assumption.
	Hence, $G \in \TD_q$ and clearly $G \in \TW_{k-1}$.
	
	Finally, we derive that the bound in \autoref{lem:lower-bound-rounds} applies, yielding that Robber wins the game $\CR_q^k(G)$.
	Since $\ell \geq \frac{q}{k-1}(k+1)$ and $q \geq k-1$.
	\begin{align*}
		\numberOfRoundsK 
		& \geq \frac{q(k+1) - (k-1)(k-3)}{4} \\ 
		& \geq q + \frac{(k-1)(k-3) - (k-1)(k-3)}{4} \\
		& = q.
	\end{align*}
	Hence, Robber wins $\CR_q^k(G)$.
\end{proof}

This now enables us to separate the graph classes \Ekq\ and $\TW_{k-1}\cap\TD_q$, for $q$ sufficiently larger than $k$.
The following is a formal restatement of \autoref{thm:Ekq_tw-td-informal}:

\begin{restatable}{thm}{ekqtwtdsyntax}
	\label{thm:Ekq_tw-td}
	For $q\geq 3$, it holds that $\EParam{2}{q}\subsetneq \TW_{1}\cap\TD_q$, and for $k,q\geq 1$ such that $2\leq k-1\leq\frac{q-4}{1+\log q}$, it holds that $\Ekq\subsetneq \TW_{k-1}\cap\TD_q$.
\end{restatable}

\begin{proof}
	By \autoref{thm:Ekq-cops} and \autoref{lem:CR_tw-td}.
\end{proof}

\section{Homomorphism indistinguishability and logical equivalence}
\label{sec:deep-wide-homind}

In this section, we want to investigate $\Ekq$ in terms of homomorphism indistinguishability.
It turns out that the representation of graphs in $\Ekq$ in terms of construction trees offers a perspective suitable for obtaining characterisations of logical equivalence.
The general idea will be to use these trees to inductively construct $\lC$-formulae that capture homomorphism counts.
Not only does this approach generalise results from \cite{Dvorak_recognizing_2010,Grohe_counting_2020}, it also yields an intuitive characterisation of $\lC^k_q$-equivalence and thereby provides a more elementary proof of a result from \cite{dawar_lovasz-type_2021}.

Moreover, the constructive nature of our proof strategy proves fruitful in obtaining additional characterisations of fragments of $\lC$.
The general idea is to impose natural restrictions on the construction trees, such that a fragment $\lL \subsetneq \lC$ already suffices to capture homomorphism counts.
By choosing these restrictions carefully, the resulting subclass of $\Ekq$ is then still large enough to capture $\lL$-equivalence.
We illustrate this point by giving a characterisation of \emph{guarded} counting logic $\lGC$.

We conclude by \emph{semantically} separating $\Ekq$ and $\mathcal{TW}_{k-1} \cap \mathcal{TD}_q$.
More formally, we show \autoref{cor:semantic}, i.e.\ that, for $q$ sufficiently larger than $k$, there exist graphs $G$ and $H$ which are homomorphism indistinguishable over $\Ekq$ but have different numbers of homomorphisms from some graph in $\mathcal{TW}_{k-1} \cap \mathcal{TD}_q$.

\subsection{Homomorphism indistinguishability over \texorpdfstring{$\Ekq$}{Ekq} is \texorpdfstring{$\lC^k_q$}{Ckq}-equivalence}
\label{subsec:ckq}

In his 2010 paper \cite{Dvorak_recognizing_2010}, Dvo\v{r}ák showed that $\lC^k$-equivalence is equivalent to homomorphism indistinguishability over $\TW_k$.
It turns out that his techniques generalize remarkably well to construction trees.
To begin with, we make a few observations on how the operations that make up construction trees interact with homomorphism counts.

The following is not hard to see.

\begin{obs}\label{prp:homproducts}
	For labelled graphs $F_1, F_2, G$, it holds that \[\hom(F_1F_2, G) = \hom(F_1, G) \cdot \hom(F_2, G).\]
\end{obs}

This is because any two homomorphisms $g \colon F_1 \to G$ and $h \colon F_2 \to G$ must agree on vertices with the same label, so $g \sqcup h$ is well-defined and a homomorphism from $F_1F_2$ to $G$.
Moreover, for $h \in \HOM(F_1F_2, G)$ the restrictions $h\restrict{V(F_1)}$ and $h\restrict{V(F_2)}$ are homomorphisms.
We can also relate the homomorphism counts from graphs $F$ and $F'$, whenever $F'$ is obtained from $F$ by removing some label $\ell$.
Then in any homomorphism $h \colon F' \to G$ the image of $\labfkt_F(\ell)$ is no longer necessarily $\labfkt_G(\ell)$.
Hence, we can obtain $\hom(F, G)$ by moving the label $\ell$ to different vertices in $G$ and tallying up the homomorphisms from $F$ to those graphs.

\begin{obs}\label{prp:homlabeldel}
	Let $F'$ be the graph obtained from $F$ by removing a single label $\ell$.
	Then $\hom(F', G) = m$ if and only if there exists a decomposition $m = \sum_ic_im_i$ with $c_i, m_i \in \NN$, such that:
	\begin{itemize}
		\item There exist exactly $c_i$ vertices $v$ with $\hom(F, G(\ell \rightarrow v)) = m_i$.
		\item There exist exactly $c \coloneqq \sum_i c_i$ vertices $v$ with $\hom(F, G(\ell \rightarrow v)) \neq 0$.
	\end{itemize}
\end{obs}

Finally, observe that when $F$ is fully labelled there can be at most one homomorphism $h \colon F \to G$, which is entirely determined by the label positions in $G$.

\begin{obs}
	Let $F$ be a fully labelled graph and let $L_F$ denote the set of labels.
	Then there exists a unique homomorphism $h \colon F \to G$ if for all labels $i, j \in L_F$
	\begin{itemize}
		\item $\labfkt_F(i) = \labfkt_F(j) \implies \labfkt_G(i) = \labfkt_G(j)$,
		\item $\labfkt_F(i)\labfkt_F(j) \in E(F) \implies \labfkt_G(i)\labfkt_G(j) \in E(G)$.
	\end{itemize}
\end{obs}

The crucial insight is that these conditions are all definable in $\lC$.
The condition for fully labelled graphs in particular can be expressed as a conjunction of atomic formulae using at most $|L_F|$ different variables.
This allows us to prove the following lemma by induction over a construction tree.

\begin{lem}
	\label{lem:homcounts_in_ckq}
	Let $F \in \Lkq$ be a $k$-labelled graph, and let $m\in \NN$.
	There exists a formula $\phi_m \in \lC^k_q$ such that for each $k$-labelled graph $G$ with $L_F \subseteq L_G$, $G \models \phi_m$ if and only if $\hom(F, G) = m$.
\end{lem}

\begin{proof}
	Since $F$ is \elimOrd{k}{q}, there is a $k$-construction tree $(T, \lambda, r)$ for $F$ of \elimDepth{} at most $q$.
	We construct $\phi_m$ by induction over the structure of $T$.
	Without loss of generality, we assume that $T$ is a binary tree.
	
	Let $v \in V(T)$ be a leaf of $T$.
	Since $H \coloneqq \lambda(v)$ is fully labelled, for any graph $G$ with $L_{H} \subseteq L_G$ there is either a unique homomorphism from $H$ to
	$G$ or none at all.
	We thus let $\phi_1^v$ be the conjunction of formulae $x_i = x_j$ for $\labfkt_{H}(i) = \labfkt_{H}(j)$ and $Ex_ix_j$ for $\labfkt_{H}(i)\labfkt_{H}(j) \in E(H)$.
	We then let $\phi^v_0 = \neg \phi^v_1$ and $\phi^v_m = \bot$ for $m > 1$.
	Note that for all $m$, $\phi^v_m$ uses at most $k$ distinct variables and has quantifier-rank $0$, so $\phi_m^v \in \lC^k_q$.
	
	Now let $v \in V(T)$ be a node with two children $w_1, w_2$.
	Since $T$ is a construction tree, we have $\lambda(v) = \lambda(w_1)\lambda(w_2)$.
	Per induction hypothesis, there are formulae $\phi^{w_1}_m, \phi^{w_2}_m \in \lC^k_q$, such that $\hom(\lambda(w_i), G) = m$ if and only if $G \models \phi^{w_i}_m$ for appropriately labelled graphs $G$, $i \in \{1,2\}$, and $m \in \NN$.
	If $m \geq 1$ we let $\phi^v_m$ be the disjunction of formulae $\phi^{w_1}_{m_1} \land \phi^{w_2}_{m_2}$ for all $m_1, m_2$ with $m = m_1m_2$.
	For $m = 0$ we let $\phi_m^v = \phi^{w_1}_0 \lor \phi^{w_2}_0$.
	This boolean combination does not alter the quantifier rank, so $\phi^v_m$ still has quantifier-rank at most $q$.
	Moreover, $\phi^{w_1}_m$ and and $\phi^{w_2}_m$ both have variables among $x_1, \dots, x_k$, so $\phi^v_m \in \lC^k_q$.
	
	Finally, suppose $v \in V(T)$ has only one child $w$.
	This implies that we can obtain $\lambda(v)$ from $\lambda(w)$ by removing a label $\ell$.
	Per induction hypothesis, there are formulae $\phi^w_m \in \lC^k_q$ that capture that there are $m$ homomorphisms from $\lambda(w)$ to an appropriately labelled graph $G$.
	By \autoref{prp:homlabeldel}, we can define $\phi^v_m$ as the disjunction over all decompositions $m = \sum_{i=1}^t c_im_i$, for $c_i, m_i \in \NN$ and $c \coloneqq \sum c_i$, of formulae
	\begin{equation*}
		\exists^{= c} x_\ell ~\neg\phi^w_0 \land \bigwedge_{i \in [t]} \exists^{=c_i}x_l ~\phi^w_{m_i}.
	\end{equation*}
	
	If $v$ is the root of $T$, i.e. $\lambda(v) = F$, we let $\phi_m = \phi_m^v$.
	Observe that each elimination step increases the quantifier-rank of $\phi_m$ by one, and at the leafs the quantifier-rank is 0.
	Since $T$ has \elimDepth{} at most $q$, it holds that $\phi_m \in \lC^k_q$.
\end{proof}

Ideally, we would like to prove the converse in a similar manner.
Given some $\lC^k_q$ formula $\psi$ that distinguishes two graphs $G$ and $H$, we want to construct a graph $F \in \Lkq$ with $\hom(F, G) \neq \hom(F, H)$ by induction over the structure of $\psi$.
While graphs are too rigid in this regard, such a construction will be possible using \emph{formal linear combinations} of graphs.\footnote{These linear combinations are called \enquote{quantum graphs} in \cite{Dvorak_recognizing_2010}.} A formal linear combination is just a set of tuples $(c_i, G_i)$ of scalar values $c_i$ and graphs $G_i$, that we interpret as a linear combination of the $G_i$ with corresponding coefficient $c_i$.

For a class of (labelled) graphs $\mathcal{F}$, we let $\RR\mathcal{F}$ be the class of finite formal linear combinations with real coefficients of graphs $F \in \mathcal{F}$.
We linearly extend the $\hom$ function to $\RR\cG$ by defining
\begin{equation*}
	\hom(\qg{F}, G) = \hom(\sum_i c_iF_i, G) \coloneqq \sum_i c_i \cdot \hom(F_i, G),
\end{equation*}
for $\qg{F} = \sum_i c_iF_i \in \RR\mathcal{F}$.

The following observation shows that homomorphism indistinguishability over $\mathcal{F}$ and over $\RR\mathcal{F}$ is essentially the same. This allows us to reason about linear combinations instead of graphs.
This argument is also used in \cite{Dvorak_recognizing_2010}.

\begin{obs}\label{prp:qgdistinguishing}
	Let $G, H$ be graphs and let $\qg{F} \in \RR\mathcal{F}$.
	If $\hom(\qg{F}, G) \neq \hom(\qg{F}, H)$, then there is already an $F \in \mathcal{F}$ with $\hom(F, G) \neq \hom(F, H)$.
\end{obs}

The product of two linear combinations is defined in the natural way, that is,
\begin{equation*}
	\left(\sum_i c_i F_i\right)\left(\sum_i c'_i F'_i\right) = \sum_{ij} c_ic'_j F_iF'_j.
\end{equation*}
We also remove any graphs with loops that might have been created from the resulting linear combination, to ensure that it only contains loopless graphs.
This does not impact the homomorphism count into a loopless graph $H$, as there are no homomorphisms from a graph with a self-loop into a loopless graph.
This definition preserves the property on homomorphism counts of products, that is $\hom(\qg{F}_1\qg{F}_2, H) = \hom(\qg{F}_1, H)\hom(\qg{F}_2, H)$ and admits the following interpolation lemma.

\begin{lemC}[{\cite[Lemma 5]{Dvorak_recognizing_2010}}]
	\label{lem:qginterpolation}
	Let $\mathcal{F}$ be a class of graphs and let $\qg{F} \in \RR\mathcal{F}$.
	If $S^-, S^+$ are disjoint finite sets of real numbers, then there exists a linear combination of graphs $\qgp{F}{S^+}{S^-} \in \RR\mathcal{G}$, such that for any graph $G$
	\begin{itemize}
		\item $\hom(\qgp{F}{S^+}{S^-}, G) = 1$ if $\hom(\qg{F}, G) \in S^+$, and
		\item $\hom(\qgp{F}{S^+}{S^-}, G) = 0$ if $\hom(\qg{F}, G) \in S^-$.
	\end{itemize}
	Moreover, if $\mathcal{F}$ is closed under taking products then $\qgp{F}{S^+}{S^-} \in \RR\mathcal{F}$.
\end{lemC}

With this result, we may construct for a formula $\psi \in \lC^k_q$ and $n \in \NN$ a linear combination $\qg{F}_{\psi, n}$ such that for all graphs $G$ of size $n$ it holds that $\hom(\qg{F}_{\psi, n}, G) = 1$ if $G \models \psi$ and $\hom(\qg{F}_{\psi, n}, G) = 0$ otherwise.
We say that $\qg{F}_{\psi, n}$ \emph{models $\psi$ for graphs of size $n$}.

\begin{lem}
	\label{lem:qg_from_ckq}
	Let $k, q \geq 1$ and let $\phi$ be a $\lC^k_q$-formula.
	Then for every $n \geq 1$ there exists an $\qg{F} \in \RR\Lkq$ modelling $\phi$ for graphs of size $n$.
\end{lem}

\begin{proof}
	The proof is by induction over the structure of $\phi$.
	If $\phi = [x_i = x_j]$, we let $\qg{F} = F$ be the graph consisting of a single vertex $v$ with $\labfkt_F(i) = \labfkt_F(j) = v$.
	If $\phi = Ex_ix_j$, we let $\qg{F} = F$ be the graph consisting of two adjacent vertices $v_1, v_2$ with $\labfkt_F(i) = v_1$ and $\labfkt_F(j) = v_2$, unless $i = j$, in which case we let $\qg{F} = 0$.
	It is not hard to see that there exists a (unique) homomorphism from $\qg{F}$ to a loopless graph $G$ iff $G \models \varphi$.
	In all these cases $\qr(\phi) = 0$, and since $F$ is always fully labelled, it holds that $F$ is \elimOrd{k}{0}.
	Consequently, $\qg{F} \in \RR\LParam{k}{0}$.
	
	If $\phi = \neg \psi$, then there exists per induction hypothesis an $\qg{F}_\psi \in \RR\Lkq$ modelling $\psi$ for graphs of order $n$.
	We use the interpolation construction from \autoref{lem:qginterpolation} and let $\qg{F} = \qgp{F_\psi}{\{0\}}{\{1\}}$.
	Since $\Lkq$ is closed under taking products, we have $\qg{F} \in \RR\Lkq$.
	
	If $\phi = \psi \lor \theta$, let $\qg{F_\psi}, \qg{F_\theta}$ be defined as above.
	Then $\qg{F} \coloneqq \qgp{F'}{\{1, 2\}}{\{0\}}$ where $\qg{F'} = \qg{F_\psi} + \qg{F_\theta}$ models $\phi$ for graphs of order $n$.
	Again, by \autoref{lem:qginterpolation} it is $\qg{F} \in \RR\Lkq$.
	
	Finally, consider the case $\phi = \exists^{{\geq} t} x_\ell \psi$.
	Let $\qg{F_\psi} = \sum_{i}c_iF_{\psi, i}$ be a linear combination modelling $\psi$ for graphs of size $n$.
	Since $\qr(\psi) = \qr(\phi) - 1$, from the induction hypothesis we get that $\qg{F_\psi} \in \RR\LParam{k}{q-1}$.
	We let $\qg{F_\psi'}$ be the graph obtained from $\qg{F_\psi}$ by removing the label $\ell$ from all $F_{\psi, i}$.
	Then
	\begin{align*}
		\hom(\qg{F_\psi'}, G) &= \sum_{v \in V(G)}\sum_i c_i\hom(F_{\psi, i}, G(\ell \to v))\\
		&= \sum_{v \in V(G)}\hom(\qg{F_\psi}, G(\ell \to v)),
	\end{align*}
	so $\qg{F} = \qgp{F_\psi'}{\{t, \dots, n\}}{\{0, \dots, t-1\}}$ models $\phi$ for graphs of order $n$.
	Moreover, it is easy to see that the $F_{\psi, i}'$, obtained by removing a label from $F_{\psi, i}$, are \elimOrd{k}{q}.
	Consequently, $\qg{F} \in \RR\Lkq$.
\end{proof}

The construction used above has the property that labels in the components of $\qg{F}$ correspond to free variables of $\phi$.
This correspondence yields the following corollary.

\begin{cor}\label{cor:qg_from_ckq}
	Let $k, q \geq 1$ and let $\phi$ be a $\lC^k_q$-sentence.
	Then for every $n \geq 1$ there exists an $\qg{F} \in \RR\Ekq$ modelling $\phi$ for graphs of size $n$.
\end{cor}

We can now prove the main result of this section.

\begin{thm}\label{thm:ckq_equivalence}
	Let $k,q \geq 1$.
	Two graphs $G$ and $H$ are $\lC^k_q$-equivalent if and only if they are homomorphism indistinguishable over $\Ekq$.
\end{thm}
\begin{proof}
	Suppose there was a graph $F \in \Ekq \subseteq \Lkq$ with $\hom(F, G) \neq \hom(F, H)$.
	Then by \autoref{lem:homcounts_in_ckq} there exist $\lC^k_q$ sentences $\varphi^F_m$ such that $G \models \varphi^F_m$ iff $\hom(F, G) = m$.
	Consequently, there exists an $m$ with $G \models \varphi^F_m$ and $H \not\models \varphi^F_m$, so $G$ and $H$ cannot satisfy the same $\lC^k_q$ sentences.
	
	Suppose now there was a sentence $\phi \in \lC^k_q$ with $G \models \phi$ and $H \not\models \phi$.
	W.l.o.g. we assume $|G| = |H| = n$.
	Then, by \autoref{cor:qg_from_ckq}, there is an $\qg{F} \in \RR\Ekq$ that models $\phi$ for graphs of size $n$, that is, $\hom(\qg{F}, G) \neq \hom(\qg{F}, H)$.
	By \autoref{prp:qgdistinguishing}, this already implies the existence of an $F \in \Ekq$ with $\hom(F, G) \neq \hom(F, H)$.
\end{proof}

By dropping the restriction on one of the parameters in \autoref{thm:ckq_equivalence}, we recover the original results of Dvo\v{r}\'ak~\cite{Dvorak_recognizing_2010} and Grohe~\cite{Grohe_counting_2020}:

\begin{cor}
	Let $k, q \geq 1$.
	Let $G$ and $H$ be graphs.
	\begin{enumerate}
		\item $G$ and $H$ are $\lC^k$-equivalent iff they are homomorphism indistinguishable over the graph class $\mathcal{TW}_{k-1}$.
		\item $G$ and $H$ are $\lC_q$-equivalent iff they are homomorphism indistinguishable over the graph class $\mathcal{TD}_q$.
	\end{enumerate}
\end{cor}

\subsection{Guarded fragments}
\label{subsec:gckq}

Given the constructive nature of these proofs, it is interesting to investigate whether the same strategy can be used to obtain results for different fragments of $\lC$ by restricting construction trees in some way.
An example where this works well is guarded counting logic $\lGC$.

In the guarded fragment $\lGC$, quantifiers are restricted to range over the neighbours of a vertex, that is they are equipped with a guard.
Formally, we require that quantifiers only occur in the form $\exists^{{\geq} t}y (Exy \land \psi(z_1, \dots, z_n, y))$, where $x$ and $y$ are distinct variables.

Since $\lGC$-formulae necessarily have a free variable, it is not immediately obvious how to define $\lGC$-equivalence on graphs.
One option is to pair graphs together with a distinguished vertex and write $G, v \equiv_\lGC H, w$.
This works, and we obtain, in fact, a characterisation for precisely this relation.
However, we would prefer to study the landscape of homomorphism indistinguishability relations on graphs without restrictions.
The following natural definition of \textsf{GC}-equivalence allows us to lift our result to graphs without a dinstinguished vertex.

\begin{defi}[\textsf{GC}-equivalence]\label{def:gcequiv}
	Let $G$ and $H$ be unlabelled graphs.
	We say that $G$ and $H$ are $\textsf{GC}$-equivalent if there exists a bijection $f \colon V(G) \to V(H)$ such that
	\[G, v \models \varphi(x) \iff H, f(v) \models \varphi(x)\]
	for all $v \in V(G)$ and $\varphi \in \mathsf{GC}$.
	In this case we write $G \equiv_{\mathsf{GC}} H$.
\end{defi}

To apply our proof strategy to $\lGC$, we need to restrict the construction trees such that guarded quantifiers suffice to express homomorphism counts.
Observe that the quantifiers are only needed to describe how the number of homomorphisms $F \to G$ changes by removing a label from $F$.
More precisely, we use the fact that removing a label $\ell$ from $F$ is the same as moving it around in $G$ and tallying up the resulting homomorphisms.
Now if $\ell$ is adjacent to some other label $\ell'$, then the only positions of $\ell$ in $G$ that contribute to the final homomorphism count are adjacent to $\ell'$.
Consequently, it will suffice to quantify over the neighbours of $\ell'$.

\begin{defi}
	Let $k, q \geq 1$.
	By $\GEkq$ we denote the class of $k$-labelled graphs that admit a $k$-construction tree of elimination-depth $q$ with the additional restriction that labels can only be removed if they have a labelled neighbour.
\end{defi}

We observe that in \autoref{fig:elim-example} there are nodes where labels without labelled neighbors are removed.
In \autoref{fig:guard-elim-example}, we depict a construction tree without such nodes of the same graph.
We remark that all graphs in $\GEkq$ are labelled, as a single label can never be removed.
Under these restrictions, the argument from \autoref{lem:homcounts_in_ckq} goes through using only guarded quantifiers.

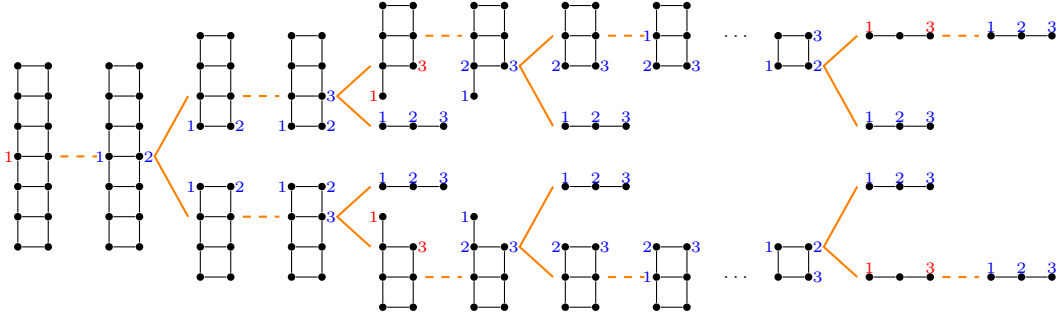
\begin{figure}
	\centering
	\begin{tikzpicture}[scale=0.4,smallVertex/.style={fill=black, inner sep=1, circle}, node font=\tiny]
		\foreach \i in {0,1,2,3,4,5,6}{
			\foreach \j in {0,1}{
				\node[smallVertex] (a\j\i) at (\j,\i) {};
				\ifthenelse{\NOT \i=0}{
					\tikzmath{
						integer \l;
						\l = \i - 1;}
					\draw (a\j\i) -- (a\j\l);
				}{}
			}
			\draw (a0\i) -- (a1\i);
		}
		\node[red] at (-0.3,3) {$1$};
		
		\foreach \i in {0,1,2,3,4,5,6}{
			\foreach \j in {0,1}{
				\tikzmath{
					integer \newj;
					\newj = \j+3;}
				\node[smallVertex] (b\j\i) at (\newj,\i) {};
				\ifthenelse{\NOT \i=0}{
					\tikzmath{
						integer \l;
						\l = \i - 1;}
					\draw (b\j\i) -- (b\j\l);
				}{}
			}
			\draw (b0\i) -- (b1\i);
		}
		\node[blue] at (2.7,3) {$1$};
		\node[blue] at (4.3,3) {$2$};
		
		\foreach \i in {0,1,2,3}{
			\foreach \j in {0,1}{
				\tikzmath{
					integer \newj;
					\newj = \j+6;
					integer \newi;
					\newi = \i+4;}
				\node[smallVertex] (c\j\i) at (\newj,\newi) {};
				\ifthenelse{\NOT \i=0}{
					\tikzmath{
						integer \l;
						\l = \i - 1;}
					\draw (c\j\i) -- (c\j\l);
				}{}
			}
			\draw (c0\i) -- (c1\i);
		}
		\node[blue] at (5.7,4) {$1$};
		\node[blue] at (7.3,4) {$2$};
		
		\foreach \i in {0,1,2,3}{
			\foreach \j in {0,1}{
				\tikzmath{
					integer \newj;
					\newj = \j+6;
					integer \newi;
					\newi = \i-1;}
				\node[smallVertex] (d\j\i) at (\newj,\newi) {};
				\ifthenelse{\NOT \i=0}{
					\tikzmath{
						integer \l;
						\l = \i - 1;}
					\draw (d\j\i) -- (d\j\l);
				}{}
			}
			\draw (d0\i) -- (d1\i);
		}
		\node[blue] at (5.7,2) {$1$};
		\node[blue] at (7.3,2) {$2$};
		
		\foreach \i in {0,1,2,3}{
			\foreach \j in {0,1}{
				\tikzmath{
					integer \newj;
					\newj = \j+9;
					integer \newi;
					\newi = \i+4;}
				\node[smallVertex] (c\j\i) at (\newj,\newi) {};
				\ifthenelse{\NOT \i=0}{
					\tikzmath{
						integer \l;
						\l = \i - 1;}
					\draw (c\j\i) -- (c\j\l);
				}{}
			}
			\draw (c0\i) -- (c1\i);
		}
		\node[blue] at (8.7,4) {$1$};
		\node[blue] at (10.3,4) {$2$};
		\node[blue] at (10.3,5) {$3$};
		
		\foreach \i in {0,1,2,3}{
			\foreach \j in {0,1}{
				\tikzmath{
					integer \newj;
					\newj = \j+9;
					integer \newi;
					\newi = \i-1;}
				\node[smallVertex] (d\j\i) at (\newj,\newi) {};
				\ifthenelse{\NOT \i=0}{
					\tikzmath{
						integer \l;
						\l = \i - 1;}
					\draw (d\j\i) -- (d\j\l);
				}{}
			}
			\draw (d0\i) -- (d1\i);
		}
		\node[blue] at (8.7,2) {$1$};
		\node[blue] at (10.3,2) {$2$};
		\node[blue] at (10.3,1) {$3$};
		
		\foreach \i in {0,1,2}{
			\foreach \j in {0,1}{
				\tikzmath{
					integer \newj;
					\newj = \j+12;
					integer \newi;
					\newi = \i+6;}
				\node[smallVertex] (c\j\i) at (\newj,\newi) {};
				\ifthenelse{\NOT \i=0}{
					\tikzmath{
						integer \l;
						\l = \i - 1;}
					\draw (c\j\i) -- (c\j\l);
				}{}
			}
			\draw (c0\i) -- (c1\i);
		}
		\node[smallVertex] (c) at (12,5) {};
		\draw (c) -- (c00);
		\node[red] at (11.7,5) {$1$};
		\node[red] at (13.3,6) {$3$};
		
		\foreach \j in {0,1,2}{
			\tikzmath{
				integer \newj;
				\newj = \j + 12;}
			\node[smallVertex] (e\j) at (\newj,4) {};
			\node[smallVertex] (f\j) at (\newj,2) {};
			\ifthenelse{\NOT \j=0}{
				\tikzmath{
					integer \l;
					\l = \j - 1;}
				\draw (e\j) -- (e\l);
				\draw (f\j) -- (f\l);
			}{}
		}
		\node[blue] at (12,4.3) {$1$};
		\node[blue] at (13,4.3) {$2$};
		\node[blue] at (14,4.3) {$3$};
		\node[blue] at (12,2.3) {$1$};
		\node[blue] at (13,2.3) {$2$};
		\node[blue] at (14,2.3) {$3$};
		
		\foreach \i in {0,1,2}{
			\foreach \j in {0,1}{
				\tikzmath{
					integer \newj;
					\newj = \j+12;
					integer \newi;
					\newi = \i-2;}
				\node[smallVertex] (d\j\i) at (\newj,\newi) {};
				\ifthenelse{\NOT \i=0}{
					\tikzmath{
						integer \l;
						\l = \i - 1;}
					\draw (d\j\i) -- (d\j\l);
				}{}
			}
			\draw (d0\i) -- (d1\i);
		}
		\node[smallVertex] (d) at (12,1) {};
		\draw (d) -- (d02);
		\node[red] at (11.7,1) {$1$};
		\node[red] at (13.3,0) {$3$};
		
		\foreach \i in {0,1,2}{
			\foreach \j in {0,1}{
				\tikzmath{
					integer \newj;
					\newj = \j+15;
					integer \newi;
					\newi = \i+6;}
				\node[smallVertex] (c\j\i) at (\newj,\newi) {};
				\ifthenelse{\NOT \i=0}{
					\tikzmath{
						integer \l;
						\l = \i - 1;}
					\draw (c\j\i) -- (c\j\l);
				}{}
			}
			\draw (c0\i) -- (c1\i);
		}
		\node[smallVertex] (c) at (15,5) {};
		\draw (c) -- (c00);
		\node[blue] at (14.7,5) {$1$};
		\node[blue] at (14.7,6) {$2$};
		\node[blue] at (16.3,6) {$3$};
		
		\foreach \i in {0,1,2}{
			\foreach \j in {0,1}{
				\tikzmath{
					integer \newj;
					\newj = \j+15;
					integer \newi;
					\newi = \i-2;}
				\node[smallVertex] (d\j\i) at (\newj,\newi) {};
				\ifthenelse{\NOT \i=0}{
					\tikzmath{
						integer \l;
						\l = \i - 1;}
					\draw (d\j\i) -- (d\j\l);
				}{}
			}
			\draw (d0\i) -- (d1\i);
		}
		\node[smallVertex] (d) at (15,1) {};
		\draw (d) -- (d02);
		\node[blue] at (14.7,1) {$1$};
		\node[blue] at (14.7,0) {$2$};
		\node[blue] at (16.3,0) {$3$};
		
		\foreach \i in {0,1,2}{
			\foreach \j in {0,1}{
				\tikzmath{
					integer \newj;
					\newj = \j+18;
					integer \newi;
					\newi = \i+6;}
				\node[smallVertex] (c\j\i) at (\newj,\newi) {};
				\ifthenelse{\NOT \i=0}{
					\tikzmath{
						integer \l;
						\l = \i - 1;}
					\draw (c\j\i) -- (c\j\l);
				}{}
			}
			\draw (c0\i) -- (c1\i);
		}
		\node[blue] at (17.7,6) {$2$};
		\node[blue] at (19.3,6) {$3$};
		
		\foreach \j in {0,1,2}{
			\tikzmath{
				integer \newj;
				\newj = \j + 18;}
			\node[smallVertex] (e\j) at (\newj,4) {};
			\node[smallVertex] (f\j) at (\newj,2) {};
			\ifthenelse{\NOT \j=0}{
				\tikzmath{
					integer \l;
					\l = \j - 1;}
				\draw (e\j) -- (e\l);
				\draw (f\j) -- (f\l);
			}{}
		}
		\node[blue] at (18,4.3) {$1$};
		\node[blue] at (19,4.3) {$2$};
		\node[blue] at (20,4.3) {$3$};
		\node[blue] at (18,2.3) {$1$};
		\node[blue] at (19,2.3) {$2$};
		\node[blue] at (20,2.3) {$3$};
		
		\foreach \i in {0,1,2}{
			\foreach \j in {0,1}{
				\tikzmath{
					integer \newj;
					\newj = \j+18;
					integer \newi;
					\newi = \i-2;}
				\node[smallVertex] (d\j\i) at (\newj,\newi) {};
				\ifthenelse{\NOT \i=0}{
					\tikzmath{
						integer \l;
						\l = \i - 1;}
					\draw (d\j\i) -- (d\j\l);
				}{}
			}
			\draw (d0\i) -- (d1\i);
		}
		\node[blue] at (17.7,0) {$2$};
		\node[blue] at (19.3,0) {$3$};
		
		\foreach \i in {0,1,2}{
			\foreach \j in {0,1}{
				\tikzmath{
					integer \newj;
					\newj = \j+21;
					integer \newi;
					\newi = \i+6;}
				\node[smallVertex] (c\j\i) at (\newj,\newi) {};
				\ifthenelse{\NOT \i=0}{
					\tikzmath{
						integer \l;
						\l = \i - 1;}
					\draw (c\j\i) -- (c\j\l);
				}{}
			}
			\draw (c0\i) -- (c1\i);
		}
		\node[blue] at (20.7,7) {$1$};
		\node[blue] at (20.7,6) {$2$};
		\node[blue] at (22.3,6) {$3$};
		
		\foreach \i in {0,1,2}{
			\foreach \j in {0,1}{
				\tikzmath{
					integer \newj;
					\newj = \j+21;
					integer \newi;
					\newi = \i-2;}
				\node[smallVertex] (d\j\i) at (\newj,\newi) {};
				\ifthenelse{\NOT \i=0}{
					\tikzmath{
						integer \l;
						\l = \i - 1;}
					\draw (d\j\i) -- (d\j\l);
				}{}
			}
			\draw (d0\i) -- (d1\i);
		}
		\node[blue] at (20.7,-1) {$1$};
		\node[blue] at (20.7,0) {$2$};
		\node[blue] at (22.3,0) {$3$};
		
		\foreach \i in {0,1}{
			\foreach \j in {0,1}{
				\tikzmath{
					integer \newj;
					\newj = \j+25;
					integer \newi;
					\newi = \i+6;}
				\node[smallVertex] (c\j\i) at (\newj,\newi) {};
				\ifthenelse{\NOT \i=0}{
					\tikzmath{
						integer \l;
						\l = \i - 1;}
					\draw (c\j\i) -- (c\j\l);
				}{}
			}
			\draw (c0\i) -- (c1\i);
		}
		\node[blue] at (24.7,6) {$1$};
		\node[blue] at (26.3,6) {$2$};
		\node[blue] at (26.3,7) {$3$};
		
		\foreach \i in {0,1}{
			\foreach \j in {0,1}{
				\tikzmath{
					integer \newj;
					\newj = \j+25;
					integer \newi;
					\newi = \i-1;}
				\node[smallVertex] (d\j\i) at (\newj,\newi) {};
				\ifthenelse{\NOT \i=0}{
					\tikzmath{
						integer \l;
						\l = \i - 1;}
					\draw (d\j\i) -- (d\j\l);
				}{}
			}
			\draw (d0\i) -- (d1\i);
		}
		\node[blue] at (24.7,0) {$1$};
		\node[blue] at (26.3,0) {$2$};
		\node[blue] at (26.3,-1) {$3$};
		
		\foreach \j in {0,1,2}{
			\tikzmath{
				integer \newj;
				\newj = \j + 28;}
			\node[smallVertex] (c\j) at (\newj,7) {};
			\node[smallVertex] (d\j) at (\newj,-1) {};
			\node[smallVertex] (e\j) at (\newj,4) {};
			\node[smallVertex] (f\j) at (\newj,2) {};
			\ifthenelse{\NOT \j=0}{
				\tikzmath{
					integer \l;
					\l = \j - 1;}
				\draw (c\j) -- (c\l);
				\draw (d\j) -- (d\l);
				\draw (e\j) -- (e\l);
				\draw (f\j) -- (f\l);
			}{}
		}
		\node[blue] at (28,4.3) {$1$};
		\node[blue] at (29,4.3) {$2$};
		\node[blue] at (30,4.3) {$3$};
		\node[blue] at (28,2.3) {$1$};
		\node[blue] at (29,2.3) {$2$};
		\node[blue] at (30,2.3) {$3$};
		\node[red] at (28,7.3) {$1$};
		\node[red] at (30,7.3) {$3$};
		\node[red] at (28,-0.7) {$1$};
		\node[red] at (30,-0.7) {$3$};
		
		\foreach \j in {0,1,2}{
			\tikzmath{
				integer \newj;
				\newj = \j + 32;}
			\node[smallVertex] (e\j) at (\newj,7) {};
			\node[smallVertex] (f\j) at (\newj,-1) {};
			\ifthenelse{\NOT \j=0}{
				\tikzmath{
					integer \l;
					\l = \j - 1;}
				\draw (e\j) -- (e\l);
				\draw (f\j) -- (f\l);
			}{}
		}
		\node[blue] at (32,7.3) {$1$};
		\node[blue] at (33,7.3) {$2$};
		\node[blue] at (34,7.3) {$3$};
		\node[blue] at (32,-0.7) {$1$};
		\node[blue] at (33,-0.7) {$2$};
		\node[blue] at (34,-0.7) {$3$};
		
		\draw[dashed, thick, orange] (1.4,3) -- (2.6,3) (7.4,5) -- (8.6,5) (7.4,1) -- (8.6,1) (13.4,7) -- (14.6,7) (13.4,-1) -- (14.6,-1) (19.4,7) -- (20.6,7) (19.4,-1) -- (20.6,-1) (30.4,7) -- (31.6,7) (30.4,-1) -- (31.6,-1);
		\draw[thick, orange] (5.6,1) -- (4.5,3) -- (5.6,5) (11.6,4) -- (10.5,5) -- (11.6,6) (11.6,0) -- (10.5,1) -- (11.6,2) (17.6,4) -- (16.5,6) -- (17.6,7) (17.6,-1) -- (16.5,0) -- (17.6,2) (27.6,4) -- (26.5,6) -- (27.6,7) (27.6,-1) -- (26.5,0) -- (27.6,2);
		\node[] at (23.6,7) {$\ldots$};
		\node[] at (23.6,-1) {$\ldots$};
	\end{tikzpicture}
	\caption{A guarded $3$-construction tree of elimination depth $7$ for the grid $\grid{2}{7}$ with one labelled vertex. Edges entering elimination nodes are dashed. At every labelled graph, those labels that may be removed are marked blue, those that may not be removed are marked red. The dotted omitted part of the construction tree follows the same pattern.}
	\label{fig:guard-elim-example}
\end{figure}

\begin{lem}
	\label{lem:gchomcap}
	Let $F \in \GEkq$.
	Then for each $m \geq 0$ there is a formula $\varphi_m \in \mathsf{GC}^k_q$ such that for appropriately labelled graphs $G$ it holds that $\hom(F, G) = m$ iff $G \models \varphi_m$.
\end{lem}

\begin{proof}
	The construction proceeds along the same lines as \autoref{lem:homcounts_in_ckq}.
	In fact, we only have to reconsider label deletions.
	Suppose $F$ is obtained from a graph $F'$ by removing a label $\ell$.
	Then per induction hypothesis there exist formulae $\varphi'_m \in \mathsf{GC}_{q-1}^k$ encoding the number of homomorphisms from $F'$ to an arbitrary graph $G$.
	Moreover, per definition there is a label $\ell'$ in $F'$ with $\labfkt(\ell)\labfkt(\ell') \in E(F')$.
	We then define the formula $\varphi_m \in \lGC_{q}^k$ as the disjunction over all decompositions $m = \sum_i c_im_i$, $c = \sum_i c_i$ of formulae
	\begin{equation*}
		\theta \coloneqq \exists^{=c}x_\ell (Ex_{\ell}x_{\ell'} \land \neg \varphi'_0) \land \bigwedge_i \exists^{=c_i}x_\ell(Ex_{\ell}x_{\ell'} \land \varphi'_{m_i}).
	\end{equation*}
	Note that this is only differs from the construction in \autoref{lem:homcounts_in_ckq} by the added guards $Ex_{\ell}x_{\ell'}$.
	To see that this does not limit the strength of the formula, suppose
	\begin{equation*}
		G \models \theta' \coloneqq \exists^{=c}x_\ell \neg \varphi'_0 \land \bigwedge_i \exists^{=c_i}x_\ell\varphi'_{m_i}.
	\end{equation*}
	Then per definition there exist exactly $c_i$ vertices $v_1, \dots, v_{c_i}$ such that for each such vertex it holds that $G(\ell \to v_j) \models \varphi'_{m_i}$.
	From the induction hypothesis we know this is equivalent to $\hom(F', G(\ell \to v_j)) = m_i > 0$.
	But for each $h \in \HOM(F', G(\ell \to v_j))$, we have $h(\labfkt(\ell))h(\labfkt(\ell')) = v_jh(\labfkt(\ell')) \in E(G(\ell \to v_j))$.
	Consequently $G(\ell \to v_j) \models Ex_{\ell}x_{\ell'}$ for all $j$, and thus $G \models \theta$.
	
	Conversely, if $G \models \theta$ but $G \not\models \theta'$, then there must exist a vertex $w\in V(G)$ such that $G(\ell \to w) \models \varphi'_{m_i}$ but $G(\ell \to w) \not\models Ex_{\ell}x_{\ell'}$.
	But then again per induction hypothesis there exists a homomorphism from $F'$ to $G(\ell \to w)$, and thus $w\labfkt(\ell') \in E(G(\ell \to w))$ and $G(\ell \to w) \models Ex_{\ell}x_{\ell'}$.
\end{proof}

The proof of the converse---showing that there exists for each $\psi \in \lGC^k_q$ an $\qg{F} \in \RR\GEkq$ modelling $\psi$---also goes through nearly unchanged.

\begin{lem}
	\label{lem:qg_from_gckq}
	Let $\varphi \in \lGC^k_q$.
	There is an $\qg{F} \in \RR\GEkq$ modelling $\varphi$ for graphs of size $n$.
\end{lem}

\begin{proof}
	Consider the construction used in the proof of \autoref{lem:qg_from_ckq}.
	We show that the same construction yields an $\qg{F} \in \RR\GEkq$ when restricting ourselves to $\lGC$ formulae.
	First note that $\GEkq$ is still closed under products, so we can use the interpolation construction without any restrictions.
	It thus suffices to consider the case that $\varphi = \exists^{{\geq} t} x_\ell (Ex_{\ell}x_{\ell'} \land \theta)$.
	
	Per assumption there then exists an $\qg{F_\theta} \in \RR\GEParam{k}{q-1}$ such that for graphs $G$, it holds that $\hom(\qg{F_\theta}, G) = 1$ if and only if $G \models \theta$.
	Moreover, there is a graph that models the term $Ex_{\ell}x_{\ell'}$ -- namely the two vertex graph $F_E = (\{x, y\}, {xy})$ with $\labfkt_{F_E}(\ell) = x$ and $\labfkt_{F_E}(\ell') = y$.
	
	The product $\qg{F'_\theta} = \qg{F}_\theta \cdot F_E$ then models $Ex_{\ell}x_{\ell'} \land \theta$.
	But note that in all graphs of $\qg{F'_\theta}$, $\labfkt(\ell)$ has a labelled neighbour -- namely $\labfkt(\ell')$.
	This means that the linear combination $\qg{F''}$ obtained by removing the label $\ell$ from all graphs in the linear combination is still in $\RR\GEkq$.
	Completely analogously to \autoref{lem:qg_from_ckq},
	\begin{equation*}
		\qg{F_\varphi} \coloneqq \qgp{F_\theta''}{\{t, \dots, n\}}{\{0, \dots, t-1\}},
	\end{equation*}
	models $\varphi$ for graphs of size $n$.
\end{proof}

The analogues of these two lemmas already sufficed to prove \autoref{thm:ckq_equivalence}.
Here, however, we still need to be mindful of any remaining labels.
Concretely, \autoref{lem:gchomcap} and \autoref{lem:qg_from_gckq} imply the following for $\lGC$ sentences.

\begin{cor}
	\label{cor:gc-thm-with-labels}
	Let $G, v$ and $H, w$ be graphs together with a single labelled vertex.
	Then the following are equivalent.
	\begin{enumerate}
		\item For all $\psi(x) \in \lGC^k_q$, it holds $G, v \models \psi(x) \iff H, w \models \psi(x)$.
		\item $\hom(F, G) = \hom(F, H)$ for all $F \in \GEkq$.
	\end{enumerate}
\end{cor}

While this is already a nice result, ideally we would like to make a statement about general, unlabelled, graphs.
Fortunately, simply removing all labels from $F \in \GEkq$ turns out to induce the equivalence relation described in \autoref{def:gcequiv}.
Let us denote by $\GEkqLL$ the class of graphs in $\GEkq$ with all labels removed.

We prove the two directions of the equivalence in \autoref{thm:guardedEkq_vs_guarded-logic-informal} separately.
First we prove the equivalence with respect to guarded counting logic implies equivalence with regard to \homInd.

\begin{lem}\label{lem:labelsFw}
	For unlabelled graphs $G$ and $H$ such that $G \equiv_{\lGC^k_q} H$,
	it holds that
	\[\hom(\GEkqLL, G) = \hom(\GEkqLL, H).\]
\end{lem}

\begin{proof}
	Let $G, H$ be unlabelled graphs with $G \equiv_{\lGC^k_q} H$.
	Then there exists a bijection $f \colon V(G) \to V(H)$ such that for all $\psi(x) \in \lGC^k_q$ it holds that
	\begin{equation*}
		G, v \models \psi(x) \iff H, f(v) \models \psi(x)
	\end{equation*}
	for all $v \in V(G)$.
	This holds in particular for the formulae $\phi^F_m(x_\ell)$ encoding homomorphism counts from a graph $F \in \GEkq$.
	We thus get
	\begin{align*}
		\hom(F, G(\ell \to v)) = m &\iff G, v \models \phi^F_m(x_\ell)\\
		&\iff H, f(v) \models \phi^F_m(x_\ell)\\
		&\iff \hom(F, H(\ell \to f(v))) = m.
	\end{align*}
	or equivalently
	\begin{equation*}
		\hom(F, G(\ell \to v)) = \hom(F, H(\ell \to f(v))).
	\end{equation*}
	
	Now for some $F \in \GEkq$, we let $F^-$ be the graph obtained from $F$ by removing all labels.
	Then $F^- \in \GEkqLL$, and moreover for each $X' \in \GEkqLL$ there is a graph $X \in \GEkq$ with $X' = X^-$.
	For an arbitrary $F^- \in \GEkqLL$ it then holds
	\begin{align*}
		\hom(F^-, G) &= \sum_{v \in V(G)}\hom(F, G(\ell \to v))\\
		&= \sum_{v \in V(G)} \hom(F, H(\ell \to f(v)))\\
		&= \sum_{w \in V(H)} \hom(F, H(\ell \to w))\\
		&= \hom(F^-, H).
		\qedhere
	\end{align*}
\end{proof}

Before we can prove the backwards direction of the equivalence we need the following folklore lemma.

\begin{lem}[Folklore, cf.\ {\cite[Lemma~4.1.11]{seppelt_homomorphism_2024}}] \label{lem:interpolation}
	Let $I$ and $J$ be finite sets.
	Let $\mathcal{F}$ be a set of pairs of functions $(a, b)$ where $a \colon I \to \mathbb{R}$ and $b \colon J \to \mathbb{R}$.
	Suppose that $\mathcal{F}$ is closed under multiplication, i.e.\@ if $(a, b), (a', b') \in \mathcal{F}$, then $(a \cdot a', b \cdot b') \in \mathcal{F}$ where $a \cdot a'$ denotes the point-wise product of $a$ and $a'$.
	Then the following are equivalent:
	\begin{enumerate}
		\item For all $(a, b) \in \mathcal{F}$, $\sum_{i \in I} a(i) = \sum_{j \in J} b(j)$,
		\item There exists a bijection $\pi \colon I \to J$ such that $a = b \circ \pi$ for all $(a,b) \in \mathcal{F}$.
	\end{enumerate}
\end{lem}

We are now able to prove the backwards direction using \autoref{cor:gc-thm-with-labels}.

\begin{lem}\label{lem:labelsBw}
	Let $G, H$ be unlabelled graphs with $\hom(\GEkqLL, G) = \hom(\GEkqLL, H)$.
	Then $G \equiv_{\lGC^k_q} H$.
\end{lem}
\begin{proof}
	Write $\mathcal{F} $ for the set of pairs of functions $V(G) \to \mathbb{R}$, $v \mapsto \hom(F, G(1 \to v))$ and $V(H) \to \mathbb{R}$, $v \mapsto \hom(F, H(1 \to v))$ for all $F \in \GEkq$ with one label.
	Then $\mathcal{F}$ is as in \autoref{lem:interpolation}: For functions $v \mapsto \hom(F, G(1 \to v))$ and $v \mapsto \hom(F', G(1 \to v))$, their pointwise product is given by $v \mapsto \hom(FF', G(1 \to v))$, and it is $F \odot F' \in \GEkq$.
	Then property $1$ holds because for $F \in \GEkq$ with a single label, we have \[\sum_{v \in V(G)} \hom(F, G(1 \to v)) = \hom(F', G),\] where $F' \in \GEkqLL$ is obtained from $F$ by removing its label.
	Hence, there exists a bijection $\pi \colon V(G) \to V(H)$ from $G$ to $H$ such that, for all $v \in V(G)$ and $F \in \GEkqLL$, it holds that $\hom(F, G(1 \to v)) = \hom(F, H(1 \to \pi(v)))$.
	\autoref{cor:gc-thm-with-labels} then yields the claim.
\end{proof}

\autoref{thm:guardedEkq_vs_guarded-logic-informal} is an immediate corollary of Lemmas~\ref{lem:labelsFw} and~\ref{lem:labelsBw}.
We restate the theorem here more formally.

\begin{thm}
	\label{thm:guardedEkq_vs_guarded-logic}
	Let $k, q \geq 1$.
	Two graphs $G$ and $H$ are $\lGC^k_q$-equivalent
	if and only if 
	they are homomorphism indistinguishable over $\GEkqLL$.
\end{thm}

We remark that in \cite{Abramsky_comonadic_2021} the logic $\lGC$ was studied with comonadic means.
In this work, winning strategies for Duplicator in guarded bisimulation games were characterised as coKleisli morphisms with respect to a suitably defined comonad.
This is in contrast to the comonadic Lov\'asz-type theorem of \cite{dawar_lovasz-type_2021} which applies to logical equivalences which can be characterised as coKleisli isomorphisms.
Thus, \autoref{thm:guardedEkq_vs_guarded-logic} does not seem to be immediate from \cite{Abramsky_comonadic_2021,dawar_lovasz-type_2021}.

\section{Separating \texorpdfstring{$\Ekq$}{Ekq} from \texorpdfstring{$\mathcal{TW}_{k-1} \cap \mathcal{TD}_q$}{TWk-1 intersection TDq} semantically}
\label{sec:separation-sem}

By \autoref{thm:Ekq_tw-td}, the graph class $\Ekq$ is a proper subclass of $\mathcal{TW}_{k-1} \cap \mathcal{TD}_q$, when $q$ is sufficiently larger than $k$.
Despite that, it could well be that the homomorphism indistinguishability relations of the two graph classes (and thus also $\mathsf{C}^k_q$-equivalence) coincide, i.e.\@ $G \equiv_{\Ekq} H$ if and only if $G \equiv_{\mathcal{TW}_{k-1} \cap \mathcal{TD}_q} H$ for all graphs $G$ and $H$.
In this section we show that this is not the case.

\corSemantic*

In general, establishing that the homomorphism indistinguishability relations $\equiv_{\mathcal{F}_1}$ and $\equiv_{\mathcal{F}_2}$ of two graph classes $\mathcal{F}_1 \neq \mathcal{F}_2$ are distinct is a notoriously hard task, cf.\ \cite[Chapter~6]{seppelt_homomorphism_2024}.
Pivotal tools for accomplishing this were introduced by Roberson in \cite{roberson_oddomorphisms_2022}.
He defines the \emph{homomorphism distinguishing closure} $\cl(\mathcal{F})$ of a graph class $\mathcal{F}$ as the graph class
\[
\cl(\mathcal{F}) \coloneqq \{ F \in \mathcal{G} \mid \forall G, H\in \mathcal{G} .\ G \equiv_{\mathcal{F}} H \implies \hom(F, G) = \hom(F, H)\}.
\]
Here, $\mathcal{G}$ denotes the class of all graphs.
A graph class $\mathcal{F}$ is \emph{homomorphism distinguishing closed} if $\mathcal{F} = \cl(\mathcal{F})$.
In this case, for every $F \not\in \mathcal{F}$ there exist two graphs $G$ and $H$ homomorphism indistinguishable over $\mathcal{F}$ and satisfying that $\hom(F, G) \neq \hom(F, H)$.
Therefore, homomorphism distinguishing closed graph classes may be thought of as maximal in terms of homomorphism indistinguishability.

Roberson conjectures that \emph{every graph class which is closed under taking minors and disjoint unions is homomorphism distinguishing closed}.
This conjecture is generally open.
For an overview of the progress made towards a proof, see
\cite[Chapter~6]{seppelt_homomorphism_2024}.
We add to the short list of known homomorphism distinguishing graph classes by proving the following:

\tdclosed*

\ekqclosed*

The proof of \autoref{thm:ekq-closed} follows the proof in \cite{neuen_homomorphism-distinguishing_2023} of the assertion that the class $\mathcal{TW}_{k}$ is homomorphism distinguishing closed for all $k \geq 0$.
Central to it is a construction of highly similar graphs from \cite{roberson_oddomorphisms_2022} which is reminiscent of the CFI-construction \cite{Cai_optimal_1992}.

With these ingredients, it suffices to prove that Duplicator wins the model comparison game characterizing $\mathsf{C}^k_q$-equivalence on these CFI-like graphs constructed over a graph not contained in \Ekq.
To that end, we build a Duplicator strategy from a Robber strategy for the game $\CRkq(G)$, where the idea is to hide the difference between the two graphs in the vertices that correspond to the Robber position and to push this difference along the paths that Robber takes to escape.
The connection between model comparison and node searching games via CFI-constructions is well-known \cite{Dawar_power_2007, Hella_logical_1996}.
Crucial for the aforementioned argument is that Robber wins the \emph{non-monotone} node searching game.
Indeed, it cannot be assumed that Cop plays monotonously since he must shadow Spoiler's moves.

We start by recalling the construction of highly similar graphs from \cite{roberson_oddomorphisms_2022}.
Let $G$ be a graph and $U \subseteq V(G)$.
Write $\delta_{v, U} \coloneqq |\{v\} \cap U|$ for every $v\in V(G)$.
The graph $G_U$ has vertices $(v, S)$ for every $v \in V(G)$ and $S \subseteq E(v)$ with $|S| \equiv \delta_{v, U} \mod 2$ where $E(v)$ denotes the set of edges incident to $v$.
It contains an edge $(v, S)(u, T)$ whenever $uv \in E(G)$ and $uv \not\in S \symdiff T$ where $S\symdiff T$ denotes the symmetric difference of $S$ and $T$.
Write $\rho \colon G_U \to G$ for the homomorphism sending $(v, S)$ to $v$.

Note that the same construction also appears in \cite{Furer_rounds_2001} and may also be referred to as CFI-construction with inner vertices only.
Central is the following result of \cite{roberson_oddomorphisms_2022} regarding homomorphisms to these graphs.

\begin{lemC}[{\cite[Corollary~3.7]{roberson_oddomorphisms_2022}}] \label{lem:roberson3.7}
	For a connected graph $G$ and $U \subseteq V(G)$, the following are equivalent:
	\begin{enumerate}
		\item $|U|$ is even,
		\item $G_\emptyset \cong G_U$,
		\item $\hom(G, G_\emptyset) = \hom(G, G_U)$.
	\end{enumerate}
\end{lemC}

In virtue of \cite[Lemma~3.2]{roberson_oddomorphisms_2022}, we may write $G_0$ for $G_\emptyset$ and $G_1$ for any of the isomorphic $G_U$ with $U \subseteq V(G)$ of odd size.
The main lemma is the following:

\begin{lem}\label{prop:closedness-ekq-core}
	Let $k,q \geq 1$ and $G$ be a connected graph.
	If $G \notin \Ekq$ then $G_0 \equiv_{\Ekq} G_1$.
\end{lem}

To prove this lemma we first need to define the model comparison game characterizing $\mathsf{C}^k_q$-equivalence and show how a winning strategy for Robber in $\CR^k_q(G)$ can be turned into a winning strategy for Duplicator.

Let $G$ and $H$ be graphs and $k \geq 1$.
For a partial function $\gamma \colon [k] \rightharpoonup V(G) \times V(H)$, write $\gamma_G \colon [k] \rightharpoonup V(G)$ and $\gamma_H \colon [k] \rightharpoonup V(H)$ for the maps obtained from $\gamma$ by projecting to the respective components.
Then $\gamma$ is a \emph{partial isomorphism} if for all $i, j \in \dom(\gamma)$, $\gamma_G(i) = \gamma_G(j) \Leftrightarrow \gamma_H(i) = \gamma_H(j)$ and $\gamma_G(i)\gamma_G(j) \in E(G) \Leftrightarrow \gamma_H(i)\gamma_H(j) \in E(H)$.
Note that $\gamma = \emptyset$ is a partial isomorphism for any two graphs $G$ and $H$.

The \emph{bijective $k$-pebble game} on graphs $G$ an $H$ is played by two players Spoiler and Duplicator.
The positions are partial functions $\gamma \colon [k] \rightharpoonup V(G) \times V(H)$.
Spoiler wins in $0$ rounds, if $\gamma$ is not a partial isomorphism.
If $|V(G)| \neq |V(H)|$, then Spoiler wins in $1$ round.
At any round $i \geq 1$ starting in position $\gamma$,
\begin{itemize}
	\item Spoiler picks $p \in [k]$,
	\item Duplicator supplies a bijection $f \colon V(G) \to V(H)$,
	\item Spoiler picks $v \in V(G)$.
\end{itemize}
The position is updated to the partial map $\gamma' \coloneqq \gamma[p \mapsto (v, f(v))]$ with domain $\dom(\gamma) \cup \{p\}$ and
\[
\gamma' \colon q \mapsto \begin{cases}
	(v, f(v)), & \text{if } q = p,\\
	\gamma(q), & \text{otherwise}.
\end{cases}
\]
Spoiler wins after round $i$ if $\gamma'$ is not a partial isomorphism.
Otherwise Duplicator wins the $i$-round bijective $k$-pebble game.

The following \autoref{thm:hella} is implicit in many sources \cite{Cai_optimal_1992, Hella_logical_1996}.
As we were not able to find an exact reference, we choose to give a proof sketch:

\begin{thm} \label{thm:hella}
	Let $k\geq 1$ and $q \geq 0$.
	Let $G$ and $H$ be graphs and $\gamma \colon [k] \rightharpoonup V(G) \times V(H)$.
	Then the following are equivalent:
	\begin{enumerate}
		\item Duplicator wins the $q$-round bijective $k$-pebble game starting in position $\gamma$,
		\item For all formulae $\phi(\boldsymbol{x}) \in \mathsf{C}^k_q$ with $|\gamma|$ many free variables, $G, \gamma_G \models \phi(\boldsymbol{x})$ if and only if $H, \gamma_H \models \phi(\boldsymbol{x})$.
	\end{enumerate}
\end{thm}

\begin{proof}
	By induction on $q$.
	For $q = 0$, Duplicator wins the $0$-round game iff the starting position is a partial isomorphism.
	This is precisely what can be defined using formulae in $\mathsf{C}^{k}_0$ with $|\gamma|$ free variables.
	
	Let $q > 1$.
	Suppose Duplicator wins the $q$-round bijective game from $\gamma$.
	Let $\phi(\boldsymbol{x}) \in \mathsf{C}^k_q$ be a formula with $|\gamma|$ free variables.
	We may suppose without loss of generality that $\phi$ is of the form $\exists^{=n} x_p.\ \psi(\boldsymbol{x})$ for some $n \geq 0$ and $p \in [k]$.
	If Spoiler picks $p$, then Duplicator can supply a bijection $f \colon V(G) \to V(H)$ such that for all $v \in V(G)$, Duplicator wins the $(q-1)$-round bijective game from $\gamma[p \mapsto (v, f(v))]$.
	By the inductive hypothesis, $G, \gamma_G[p \mapsto v] \models \psi(\boldsymbol{x})$ if and only if $H, \gamma_H[p \mapsto f(v)] \models \psi(\boldsymbol{x})$ for all $v \in V(G)$ and all $\psi \in \mathsf{C}^k_{q-1}$.
	Hence, $G, \gamma_G \models \phi(\boldsymbol{x})$ if and only if $H,\gamma_H \models \phi(\boldsymbol{x})$ as desired.
	
	Conversely, suppose that $G, \gamma_G \models \phi(\boldsymbol{x})$ if and only if $H, \gamma_H \models \phi(\boldsymbol{x})$ for all $\phi \in \mathsf{C}^k_q$.
	We claim that Duplicator wins the $q$-round game starting in $\gamma$.
	Suppose Spoiler picks $p \in [k]$.
	Then there exists a bijection $f \colon V(G) \to V(H)$ such that $G, \gamma_G[p \mapsto v] \models \psi(\boldsymbol{x})$ iff $H, \gamma_H[p \mapsto f(v)] \models \psi(\boldsymbol{x})$ for all $\psi \in \mathsf{C}^k_{q-1}$ and all $v \in V(G)$.
	This is because occurrences of $\mathsf{C}^k_{q-1}$-types can be counted in $\mathsf{C}^k_q$.
	Duplicator may play this bijection.
	For any choice $v \in V(G)$ of Spoiler, it follows inductively that Duplicator wins the $(q-1)$-round game starting from the pebble positions $\gamma[p \mapsto (v, f(v))]$.
\end{proof}

\begin{lem}[{\cite[Lemma~4.3]{neuen_homomorphism-distinguishing_2023}, cf.\@ \cite{Dawar_power_2007}}]
	\label{lem:neuen4.3}
	Let $G$ be a connected graph.
	
	Let $u, v\in V(G)$.
	Let $P$ be a path in $G$ from $u$ to $v$.
	Then there exists an isomorphism $\phi \colon G_{\{ u \}} \to G_{\{v\}}$ such that 
	\begin{enumerate}
		\item $\rho(\phi(w, S)) = w$ for all $(w, S) \in V(G_{\{u\}})$ and
		\item $\phi(w, S) = (w, S)$ for all $(w, S) \in V(G_{\{u\}})$ with $w \in V(G) \setminus P$.
	\end{enumerate}
\end{lem}

We can now conduct the proof of \autoref{prop:closedness-ekq-core} similar to \cite{Furer_rounds_2001}.

\begin{proof}[Proof of \autoref{prop:closedness-ekq-core}]
	Let $G\notin\Ekq$.
	By \autoref{thm:Ekq-cops}, we get that Robber has a winning strategy in the non-monotone $\CR^k_q(G)$.
	Given this strategy, we provide a winning strategy for Duplicator in the $q$-round bijective $k$-pebble game on $G_\emptyset$ and $G_U = G_{\{u_0\}}$ for some fixed vertex $u_0 \in V(G)$ with initial position $\emptyset$.
	Clearly, $|V(G_\emptyset)| = |V(G_U)|$, so the game commences.
	
	For a position $\gamma \colon [k] \rightharpoonup V(G_\emptyset) \times V\left(G_{\{u_0\}}\right)$, we interpret the partial map $\rho \circ \gamma_{G_\emptyset}$ as a position of the cop player in $\CR^k_q(G)$.
	More precisely, the cops are placed on the vertices $\{\rho(\gamma_{G_\emptyset}(p)) \mid p \in \dom(\gamma)\}$.

	Throughout the game, Duplicator maintains the following invariant.
	For $1 \leq i \leq q$, before the $i$-th round of the bijective $k$-pebble game, with the current position being $\gamma \colon [k] \rightharpoonup V(G) \times V(H)$, there is a vertex $u \in V(G)$ and an isomorphism $\phi \colon G_{\{u\}} \to G_{\{u_0\}}$ such that
	\begin{enumerate}[label=(I.\arabic*), labelindent=0pt, itemindent=*, leftmargin=*]
		\item $\rho(\phi(w, S)) = w$ for all $(w, S) \in V(G_{\{u\}})$,\label{ax:inv1}
		\item $u$ does not appear in the image of $\rho \circ \gamma_{G_\emptyset} = \rho \circ \gamma_{G_{\{u_0\}}}$,\label{ax:inv2}
		\item $\phi$ sends pebbled vertices to pebbled vertices, i.e.\@ $\phi \circ \gamma_{G_\emptyset} = \gamma_{G_{\{u_0\}}}$,\label{ax:inv3}
		\item Robber wins the non-monotone $\CR^k_{q-i+1}(G)$ with Cop placed on $\rho \circ \gamma_{G_\emptyset}$ and Robber starting in $u \in V(G)$.\label{ax:inv4}
	\end{enumerate}
	
	To clarify the indexing, observe that initially (before the first round) with position $\gamma = \emptyset$ we require Robber to be able to win $\CR^k_q(G)$ where no cop is placed on the graph.
	
	Initially, let $u \in V(G)$ denote the vertex chosen by Robber.
	
	Let $P$ denote a shortest path from $u$ to $u_0$.
	By \autoref{lem:neuen4.3}, there exists an isomorphism $\phi \colon G_{\{u\}} \to G_{\{u_0\}}$ satisfying all stipulated properties.
	
	The invariant is maintained as follows: Let $\gamma$ denote the current position and let $u \in V(G)$ be the current position of Robber.
	When Spoiler picks a pebble $p \in [k]$, Duplicator constructs a bijection $f \colon V(G_\emptyset) \to V\left(G_{\{u_0\}}\right)$ as follows:
	For every $v \in V(G)$ Duplicator creates a copy of the Cops-and-Robber game where Cop updates his position to $(\rho \circ \gamma_{G_\emptyset})[p \mapsto v]$ (by possibly first lifting the cop associated with pebble $p$).
	Write $u_v \in V(G)$ for the vertex where robber would escape to.
	Let $P_v$ denote a shortest path from $u$ to $u_v$ avoiding the vertices which appear in $\img (\rho \circ \gamma_{G_\emptyset})[p \mapsto v] \cap \img \left(\rho \circ \gamma_{G_\emptyset}\right)$.
	Let $\psi_v$ denote the isomorphism $G_{\{u_v\}} \to G_{\{u\}}$ from \autoref{lem:neuen4.3} for $P_v$.
	Duplicator plays $f \colon x \mapsto \phi(\psi_{\rho(x)}(x))$, that is every vertex $x$ is mapped via the bijection that corresponds to the copy of the game where Cop moved a cop to $\rho(x)$.
	
	Now Spoiler picks a vertex $x \in V(G_\emptyset)$.
	Then the map $\phi' \coloneqq \phi \circ \psi_{\rho(x)}$ satisfies the properties of the invariant and $f(x) = \phi'(x)$.
	Indeed, \ref{ax:inv1} is immediate from \autoref{lem:neuen4.3}.
	\ref{ax:inv2} holds because the robber was not yet captured, i.e.\@ $\rho(x) \neq u'$ in the notation from above.
	\ref{ax:inv3} holds since $P_v$ avoids all vertices in $\img (\rho \circ \gamma_{G_\emptyset})[p \mapsto v] \cap \img \rho \circ \gamma_{G_\emptyset}$ and $f(x) = \phi'(x)$.
	Clearly, \ref{ax:inv4} holds.
	Duplicator keeps the copy of the Cops-and-Robber game where Cop played to $(\rho \circ \gamma_{G_\emptyset})[p \mapsto v]$.
	
	It remains to argue that the updated $\gamma$ is a partial isomorphism.
	By the properties of the invariant, $\phi \colon G_{\{u\}} \to G_{\{u_0\}}$ is an isomorphism and $u$ does not appear in the image of $\rho \circ \gamma_{G_\emptyset} = \rho \circ \gamma_{G_{\{u_0\}}}$.
	Hence, $\phi$ restricts to an isomorphism $G_\emptyset - \rho^{-1}(u) \to G_{\{u_0\}} - \rho^{-1}(u)$.
	The partial map $\gamma$ coincides with this isomorphism by \ref{ax:inv3}.
\end{proof}

We are finally able to prove the \homDistClure\ of \Ekq\ and $\TD_q$.

\begin{proof}[Proof of \autoref{thm:ekq-closed}]
	By \autoref{prop:closedness-ekq-core} and \autoref{lem:roberson3.7},
	for every connected graph $G \not\in \Ekq$,
	it is $G_0 \equiv_{\Ekq} G_1$ and
	$\hom(G, G_0) \neq \hom(G, G_1)$.
	It remains to deal with disconnected graphs $G \not\in \Ekq$.
	This technicality is taken care of by \cite[Corollary~7.1.5]{seppelt_homomorphism_2024}
	when observing that $\Ekq$ is closed under disjoint unions and taking summands by \autoref{cor:minor-closed}.
\end{proof}

\autoref{thm:td-closed} follows directly from \autoref{thm:ekq-closed}, as $\TD_q=\EParam{q}{q}$.
We finally derive \autoref{cor:semantic}.

\begin{proof}[Proof of \autoref{cor:semantic}]
	For graph classes $\mathcal{F}_1$ and $\mathcal{F}_2$,
	the homomorphism indistinguishability relations $\equiv_{\mathcal{F}_1}$ and $\equiv_{\mathcal{F}_2}$ coincide
	if and only if $\cl(\mathcal{F}_1) = \cl(\mathcal{F}_2)$.
	By \cite{neuen_homomorphism-distinguishing_2023} and Theorems \ref{thm:ekq-closed} and \ref{thm:td-closed}, it holds that
	$\mathcal{TW}_{k-1}$, $\Ekq$, and $\mathcal{TD}_q$ are homomorphism distinguishing closed.
	As the intersection of homomorphism distinguishing closed sets is homomorphism distinguishing closed 
	\cite[Lemma~6.1]{roberson_oddomorphisms_2022},
	the set $\mathcal{TW}_{k-1} \cap \mathcal{TD}_q$ is also homomorphism distinguishing closed.
	Hence, $\equiv_{\Ekq}$ and  $\equiv_{\mathcal{TW}_{k-1} \cap \mathcal{TD}_q}$
	coincide if and only if $\Ekq = \mathcal{TW}_{k-1} \cap \mathcal{TD}_q$.
	The desired statement now follows from \autoref{thm:Ekq_tw-td}.
\end{proof}

\section{Conclusion}
\label{sec:deep-wide-conclusion}

We study the expressive power of the counting logic fragment $\mathsf{C}^k_q$ with tools from homomorphism indistinguishability.
After giving an elementary and uniform proof of theorems from \cite{dawar_lovasz-type_2021, Dvorak_recognizing_2010, Grohe_counting_2020}, we show that the graph class $\Ekq$, whose homomorphism indistinguishability relation characterises $\mathsf{C}^k_q$-equivalence, is a proper subclass of $\TW_{k-1} \cap \TD_q$.
Finally, we show that \Ekq\ and $\TD_q$ are \homDistCl\ which implies that homomorphism indistinguishability over $\Ekq$ is not the same as homomorphism indistinguishability over $\TW_{k-1} \cap \TD_q$.

One key ingredient to the above results is finding a node searching game that can characterise the graph classes.
Exploring whether intertwining node searching and model comparison games can help to confirm Roberson's conjecture in other cases seems a tempting direction for future research.

With slight reformulations, our results might yield insights into the ability of the Weisfeiler--Leman algorithm to determine subgraph counts after a fixed number of rounds \cite{neuen_homomorphism-distinguishing_2023,rattan_weisfeiler-leman_2023}.

To use this connection of the games it is important to characterise the graph class using a node searching game that allows non-monotone moves.
Towards this characterization for the class \Ekq\ we find a new proof for the monotony of the Cops-and-Robber game for treewidth using a breadth-first \enquote{cleaning up} procedure along the pre-tree-decomposition (which may temporarily lose the property of representing a strategy).
As an interesting observation we obtain that cop moves into the back country, i.\,e.~to positions that are not part of the boundary, can be ignored and the depth of the exact \preTreeDec\ is the number of cops placed into the robber escape space:
We observe that in the proof of \autoref{cl:depth-leafs-new} of \autoref{lem:tw-ptw} where we compute how much larger the depth at some node $t^i_j$ at step $i$ is than at the \considered\ node $s_i$ the depth increases only if the node $t^i_j$ is branching by \autoref{obs:self-loops} as $t^i_j$ has a child where the cone contains only a self-loop and hence this is a move into the robber space.
\begin{cor}
	$\dep(T,r,\beta_{n_T},\gamma_{n_T})\leq \max_{\ell\in L(T)} |\{t\preceq \ell\mid t \text{ is branching} \}|$.
\end{cor}

In the future, it would be interesting to know if it is possible to give a proof that entirely argues with game strategies, that is a procedure where all intermediate steps still represent strategies.
We also leave open whether a dual object similar to brambles can be defined for bounded depth treewidth. 
Finally, given a winning strategy for $k$ cops in $q$ rounds, it would be interesting to know if it is possible to bound the number of cops necessary for winning in only $q-1$ rounds or the number of rounds given only $k-1$ cops in terms of $k$ and $q$, given that the cop player still can win.
This might also give insight into the question how the \homInd\ relations of \Ekq\ and $\EParam{k'}{q'}$ relate to each other.

\section*{Acknowledgment}
\noindent The authors thank the anonymous reviewers for their fruitful feedback.

\bibliography{literature}
\bibliographystyle{alphaurl}

\end{document}